\documentclass[10pt,reqno]{amsart}
\RequirePackage[OT1]{fontenc}
\RequirePackage{amsthm,amsmath}
\RequirePackage[colorlinks,linkcolor=blue,citecolor=blue,urlcolor=blue]{hyperref}
\usepackage{amssymb}
\usepackage{anysize}
\usepackage{hyperref}
\usepackage{extarrows}
\usepackage{multirow}
\usepackage{stmaryrd}
\usepackage{xcolor}
\usepackage{amsmath}
\usepackage{enumitem}
\usepackage{mathtools} 
\usepackage{dsfont} 
\usepackage{comment}
\usepackage{soul}
\usepackage{braket}
\usepackage{bm}

\usepackage{tikz}
\usetikzlibrary{decorations.pathreplacing,decorations.markings}
\usepackage{pgfplots,pgfmath}

\DeclareFontFamily{U}{mathx}{}
\DeclareFontShape{U}{mathx}{m}{n}{<-> mathx10}{}
\DeclareSymbolFont{mathx}{U}{mathx}{m}{n}
\DeclareMathAccent{\widecheck}{0}{mathx}{"71}

\numberwithin{equation}{section}
\theoremstyle{plain} 
\newtheorem{theorem}{Theorem}[section]
\newtheorem{lemma}[theorem]{Lemma}
\newtheorem{corollary}[theorem]{Corollary}
\newtheorem{proposition}[theorem]{Proposition}

\newtheorem{remark}[theorem]{Remark} 

\theoremstyle{definition}
\newtheorem{definition}[theorem]{Definition}
\newtheorem{assumption}[theorem]{Assumption}

\usepackage{caption} 
\renewcommand{\Re}{\mathrm{Re}\,}
\renewcommand{\Im}{\mathrm{Im}\,}

\newcommand{\E}{{\mathbb E }}

\newcommand{\R}{{\mathbb R }}
\newcommand{\N}{{\mathbb N}}

\newcommand{\HH}{{\mathbb H}}
\renewcommand{\P}{{\mathbb P}}
\newcommand{\C}{{\mathbb C}}

\newcommand{\s}{{\mathfrak{s}}}
\newcommand{\BM}{\mathfrak{B}}
\newcommand{\dift}{ {\text{\scriptsize $\Delta$}} t}

\newcommand{\ii}{\mathrm{i}}
\newcommand{\ee}{\mathrm{e}}

\newcommand{\dd}{\mathrm{d}}

\newcommand{\sgn}{\mathrm{sgn}}
\newcommand{\sfin}{s_{\mathrm{final}}}
\newcommand{\tfin}{t_{\mathrm{final}}}
\newcommand{\tinit}{t_{\mathrm{init}}}

\newcommand{\vertiii}[1]{{\left\vert\kern-0.3ex\left\vert\kern-0.3ex\left\vert #1 
		\right\vert\kern-0.3ex\right\vert\kern-0.3ex\right\vert}}

\newcommand{\nc}{\normalcolor}

\newcommand{\bs}{\boldsymbol}

\def\Tr{\mathrm{Tr}}

\def\<{\langle}
\def\>{\rangle}

\renewcommand{\mathbf}[1]{\bs{#1}}
 
\marginsize{25mm}{25mm}{25mm}{26mm}

\usepackage[style=numeric,backend=bibtex,natbib=true,maxnames=5,giveninits=true]{biblatex}

\DeclareFieldFormat*{title}{#1}
\AtEveryBibitem{\clearfield{month}}
\AtEveryBibitem{\clearfield{number}}
\AtEveryBibitem{\clearfield{publisher}}
\AtEveryBibitem{\clearfield{url}}
\DeclareFieldFormat*{eprint}{
	\hfil\penalty90\hfilneg\space \ifhyperref{
		\href{https://arxiv.org/abs/#1}{arXiv\addcolon#1} 
	}{
		arXiv\addcolon\nolinkurl{#1}
	} 
}
\AtEveryBibitem{\clearfield{doi}}
\AtEveryBibitem{\clearfield{issn}}
\allowdisplaybreaks

\title[On a Rosenzweig-Porter-type model]{On a Rosenzweig-Porter-type model}
 \author[Cipolloni \and Erd\H{o}s \and Henheik]{}
 \date{\today}
\begin{document}
	\maketitle
	
	\vspace{0.25cm}

\renewcommand{\thefootnote}{\fnsymbol{footnote}}

\noindent
\mbox{}%
\hfill%
\begin{minipage}{0.21\textwidth}
	\centering
	{Giorgio Cipolloni}\footnotemark[1]\\
	\footnotesize{\textit{cipolloni@axp.mat.uniroma2.it}}
\end{minipage}
\hfill%
\begin{minipage}{0.21\textwidth}
	\centering
	{L\'aszl\'o Erd\H{o}s}\footnotemark[2]\\
	\footnotesize{\textit{lerdos@ist.ac.at}}
\end{minipage}
\hfill%
\begin{minipage}{0.21\textwidth}
	\centering
	{Joscha Henheik}\footnotemark[3]\\
	\footnotesize{\textit{joscha.henheik@unige.ch}}
\end{minipage}
\hfill%
\mbox{}%
\footnotetext[1]{Dipartimento di Matematica, Universitá  di Roma Tor Vergata, Via della Ricerca Scientifica, 1, 00133 Roma RM, Italy.}
\footnotetext[2]{Institute of Science and Technology Austria, Am Campus 1, 3400 Klosterneuburg, Austria.}
		\footnotetext[3]{Section de Mathématiques, Université de Genève, Rue du Conseil Général 7--9, 1205 Genève, Switzerland.}

\renewcommand*{\thefootnote}{\arabic{footnote}}
\vspace{0.25cm}

 \begin{abstract}  We consider a very general Rosenzweig-Porter-type model, $H=H_0+\lambda W$,
 where $H_0$ is an arbitrary Hermitian matrix and $W$ is a standard Wigner matrix. 
 We precisely trace the localization properties of the eigenvectors 
 and the eigenstate thermalisation hypothesis (ETH) as  the coupling
 constant $\lambda$ interpolates between the trivial $\lambda=0$ case and the fully mean field regime
 of large $\lambda$. Our results hold uniformly in $H_0$ and $\lambda$, substantially generalising all previous local laws
on deformed Wigner matrices even in the mean field regime. Our proof precisely captures the deterministic approximation to the resolvent which exhibits a strongly inhomogeneous structure. As a byproduct, we conclude the emergence of a mobility edge and study the phenomenon of re-entrant localization.
\end{abstract}


\section{Introduction} \label{sec:intro}

Inspired by E. Wigner's revolutionary vision from the 1950’s \cite{wigner1993characteristic}, random matrices have been extensively used
to model complex quantum systems. Many interesting phenomena are universal, i.e.~the exact details of the system are irrelevant, hence they can be accurately modeled 
by relatively simple random Hamiltonians as long as some basic properties, such as symmetries,
are respected. Many mathematical works in random matrix theory (RMT) in the last decades
focused on the simplest class of mean field models, the original $N\times N$ Wigner matrix $H=H^*$
with its independent, identically distributed (\emph{i.i.d.}) matrix elements being the primary example.
The Wigner matrix  is completely mean-field in the sense that all sites in the underlying configuration space of
$N$ points are connected to any other sites  with statistically identical strength.
 The most important results in RMT, such as concentration of resolvents (\emph{local laws}) \cite{erdHos2009local}, 
 eigenvalue rigidity \cite{erdHos2009local, erdHos2012rigidity}, Wigner-Dyson-Mehta universality of local eigenvalue statistics \cite{erdHos2011universality, tao2011random}, eigenfunction delocalization \cite{erdHos2009local, benigni2022optimal}, and eigenstate thermalisation \cite{cipolloni2021eigenstate} (also called quantum unique ergodicity),
  have all been proven first for the Wigner ensemble. Later, they have been systematically
  extended to a broader class of models by relaxing the i.i.d.~assumption or by allowing some non-trivial
  spatial structure in the underlying configuration space. In particular, for such models
  the density of states can substantially deviate from the original Wigner semicircle law and 
  can exhibit rich patterns such as spectral bands, internal spectral edges and cusps.
  However,   most of these extensions 
  still have a mean-field character, assuming, essentially, that the expectation of $H_{ij}$ 
  is bounded and its variances are comparable for any index pair $(i,j)$.

In the current paper we 
study a celebrated  non-mean-field random matrix model, originally 
initiated by Rosenzweig and Porter \cite{rosenzweig1960repulsion}. Consider a deterministic $N\times N$ Hermitian matrix $H_0$
(\emph{unperturbed Hamiltonian}) and add to it a random mean-field perturbation, i.e.~consider
\begin{equation}
\label{H}
H\equiv H_\lambda : = H_0 + \lambda W
\end{equation}
where $W$ is a standard Wigner matrix\footnote{A standard Wigner matrix has centred entries with $\E |W_{ij}|^2 =1/N$ ensuring that $\| W \| = 2+ o(1)$ with 
very high probability.} and $\lambda>0$ is a (typically small) coupling parameter. 
In the original RP model, $H_0$ was diagonal and $W$ was a Gaussian Wigner matrix (GUE or GOE)\footnote{Note that if $W$ is GUE/GOE,
then $H_0$ can always be diagonalized without changing the distribution of $W$. \label{diag}}. We will consider general $H_0$ and general distribution for the matrix elements of  $W$.
The nontrivial spatial structure is encoded in $H_0$ and it can practically be anything,
so a priori no general statement can be expected for $\lambda=0$.  It turns out that the
added  mean field noise, even with small $\lambda$,  allows us to draw less model-dependent conclusions;
the noise makes the model more universal.

Random matrices of the form~\eqref{H}, under the name of \emph{deformed Wigner matrices}, 
 have been extensively studied in the mean field regime \cite{knowles2014outliers, lee2015edge, lee2016bulk, cipolloni2023gaussian}
basically assuming that $\|H_0\|\lesssim \lambda$. With sufficient amount of randomness at hand in this case, $H$ 
fits the class of mean field RMT models and similar  results
as for the pure Wigner case hold naturally. In particular, local laws, strong concentration of the 
resolvent $G(z)=(H-z)^{-1}$, $z\in \C\setminus \R$,  around its deterministic approximation $M=M_\lambda(z)$
have been proven with optimal errors. In our case $M$ is given as the solution of a simple implicit matrix equation, the \emph{Matrix Dyson Equation} (MDE), see also \eqref{eq:MDE} below:
$$
    M = \frac{1}{H_0-z-\lambda^2\langle M\rangle}, \qquad \langle M\rangle: = \frac{1}{N}\Tr M.
$$

 A key assumption in these studies is that $\| M_\lambda(z)\|$ is bounded uniformly in $z$, 
even if $\eta=|\Im z|$ is very small, which usually requires  some important but restrictive technical conditions
on the regularity of the spectrum of $H_0$. Note that such strong boundedness cannot hold for $\lambda=0$:
if $\Re z$ equals an eigenvalue $\mu_i$ of $H_0$, then the norm of $G(z)=M(z)$  blows up as $1/\eta$.
Moreover, in the eigenbasis of $H_0$, for small $\eta$, only a few diagonal matrix elements 
$$
   \Im M(z)_{ii} = \frac{\eta}{(\mu_i-\Re z)^2+\eta^2}
$$
are large (where $\mu_i\approx\Re z$), while the rest are much smaller.
This \emph{strongly inhomogeneous} structure
is in contrast with 
mean field situations where typically $\Im M(z)_{ii}$ for different indices $i$ are the same
(Wigner case) or at least comparable.
 As the randomness is turned on, $\lambda>0$, this inhomogeneity 
mildens; the typical energy  scale of the Cauchy-like profile of $i\mapsto  \Im M(z)_{ii}$ broadens
from $\eta$ to  $\nu: =\eta+ \lambda^2 \rho$, where $\rho:= \pi^{-1}\langle \Im M\rangle$ is the local eigenvalue
density,  see Fig.~\ref{fig:profile}. This profile was first identified by Benigni in \cite{benigni2020eigenvectors} for diagonal $H_0$.

 \begin{figure}[h]
\includegraphics{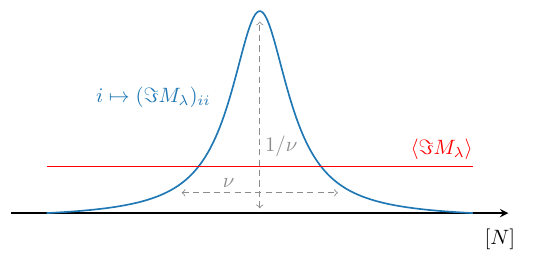}
\caption{ Typical shape of the profile $i \mapsto (\Im M_\lambda)_{ii}$ for a diagonal deterministic $H_0$, using the notation $\nu := \eta + \lambda^2 \rho$. In particular, $(\Im M_\lambda)_{ii}$ is very different from the average $\langle \Im M_\lambda \rangle$ as customary in mean-field models. }
\label{fig:profile}
 \end{figure}
Since understanding the  resolvent $G$  and its deterministic approximation $M$
 is a starting point of most current techniques in RMT,
the structure of $M$  has a profound effect on
the entire theory. The inhomogeneity discussed above is its most important 
feature; the uniform boundedness of $\|M \|$, a key previous assumption, is wrong
and it needs to be replaced by a more subtle analysis. In particular, the norm 
in which $M$ or $\Im M$ are estimated, matters a lot.
We call this regime
\emph{strongly inhomogeneous}.

Note that although  strong inhomogeneity naturally occurs for small randomness, $\lambda\ll \| H_0\|$, 
it may also be present in the mean field models where $\| H_0\|\lesssim \lambda$
if not enough spectral regularity on $H_0$ is assumed. Our current study covers this regime as well
since we make no assumption on $H_0$. In particular, we prove local laws
without an effective \emph{shape analysis}
that classified the possible singular behavior of the density\footnote{Besides the bulk of the spectrum, 
where the density is away from zero, 
there may be only two possible singularities of $\rho$: square root behaviour
at \emph{regular edges} and a cubic root profile at \emph{cusps}. } in the mean field models, see~\cite{shape}.

We highlight two main results: the single resolvent  local law, Theorem~\ref{theo:main1}, 
both in averaged and isotropic sense,
and its various natural consequences on eigenvector (de)localization,
and the eigenstate thermalisation (ETH), Theorem~\ref{theo:ETH},
which is based on a specific two-resolvent averaged local law with a \emph{regularized observable}. 
General multi-resolvent local laws for products of any number of resolvents, as well
as their isotropic versions could also be
proven without any conceptual novelty but lacking good applications we do not pursue them here.
Once the local laws for the strongly inhomogeneous 
models are understood, many other results follow with more or less standard and well-known 
procedures. We mention some of them in Section~\ref{sec:conseq}, focusing
on the physically more interesting ones such as the mobility edge and re-entrant localization \cite{ghosh2025reentrant}. We also compare our theorems
with previous results on $H_\lambda$. The closest and most relevant to us is the work of Benigni \cite{benigni2020eigenvectors} (see also \cite{bourgade2017eigenvectorsparse}), which partially motivated some questions in this manuscript. We also mention an earlier work from von Soosten and Warzel  \cite{von2019non} that studies closely similar questions. In these works
 a non-trivial profile of the eigenvectors, analogous to Fig.~\ref{fig:profile}, was
identified. Our results are much  more general in several  key aspects, most importantly
they hold for any $H_0$  while \cite{benigni2020eigenvectors, bourgade2017eigenvectorsparse, von2019non} assumed regular spectrum on $H_0$.
For our local laws we do not 
assume\footnote{Apart from a completely irrelevant technical one, see~\eqref{eq:kappa} later.} any  lower bound 
on $\lambda$, 
while the analogous
lower bound $\lambda\ge \lambda(H_0)$ in \cite{benigni2020eigenvectors, bourgade2017eigenvectorsparse, von2019non}  was $H_0$-dependent.
We also cover the entire spectrum, especially the more delicate singular regimes of the density,  while \cite{benigni2020eigenvectors, bourgade2017eigenvectorsparse, von2019non} focused on the  bulk. 
We also mention a related work \cite{von2018phase} where the very small $\lambda$ regime
was treated perturbatively to (de)localization phase transition.

For the proofs, we use the \emph{zigzag strategy}; an inductive application of a tandem of two flows on the space of matrices
 that allow us to successively  reduce $\eta=\Im z$, 
from a global scale to the optimal local one. The first flow, the \emph{zig-step}, consists
of dynamically adding a Gaussian component to $H$ by running an Ornstein-Uhlenbeck flow $H_t$.
It is complemented by a flow of the spectral parameter $z_t$ along 
a first order ODE, called the \emph{characteristic equation} in such a way that 
the time derivative of $G_t(z_t) =(H_t-z_t)^{-1}$ has a major cancellation, observed first by Pastur \cite{pastur1972spectrum} in 1972 and revived by von Soosten and Warzel in \cite{von2019non, von2019random}. In the second step, the \emph{zag-step},
this Gaussian component is removed by another flow using \emph{Green function comparison} (GFT) techniques. 
While various versions of both steps appeared in the literature long time ago (see, e.g., \cite{adhikari2020dyson, bourgade2021extreme, huang2019rigidity, lee2015edge} for the zig step
and \cite{erdHos2012bulk, knowles2013eigenvector, tao2011random, erdHos2017dynamical} for the original GFT arguments),
their successive tandem has been used to prove a local law only recently in \cite{cipolloni2024mesoscopic} (the 
name \emph{zigzag} first appeared in \cite{cipolloni2024out}). 
In the last years, the zigzag strategy superseded previous methods to prove local laws in mean field models, 
see, e.g., \cite{adhikari2023local, bourgade2022liouville, cipolloni2023eigenstate, landon2024single, landon2022almost, stone2025random, erdHos2025cusp, cipolloni2025maximum, osman2025bulk}.  Its pivotal mechanism, the cancellation in the zig step, is especially 
powerful in the spectral regimes with small density  (edges and cusps) where the earlier
static methods ran into an inherent instability problem. 
We mention that these  dynamical methods have also been essential in the recent spectacular progress in 
the mathematical theory of \emph{random band matrices} (RBM)  by proving most of the expected  analogues of earlier mean field RMT results \cite{dubova2024quantum, dubova2025delocalization2, dubova2025delocalization3, erdHos2025zigzag} (see also \cite{chen2022random, cipolloni2024dynamical, drogin2025localization, fyodorov1991scaling, erdHos2011quantum, erdHos2013delocalization, liu2023edge, peled2019wegner,schenker2009eigenvector, shcherbina2021universality, MR3623245, shcherbina2018universality, shcherbina2016transfer, sodin2010spectral, yau2025delocalization} for previous and related results in RBM, a prominent random matrix model which is not mean-field).

The main methodological message of this paper is that the zigzag strategy is very effective
in the current strongly inhomogeneous setting as well. The inhomogeneous profile of $M$
would introduce similar instabilities in the conventional static approach as the small density.
In fact, here the situation is more complex since we even do not have an effective shape analysis.

The second novelty is a streamlined version of the zigzag strategy that fully circumvents \emph{isotropic}
local laws and operates exclusively with averaged (tracial) quantities. In previous zigzag procedures for multi-resolvent local laws (see, e.g., \cite{cipolloni2023eigenstate, cipolloni2024eigenvector, nonHermdecay}),
the zag step, involving a higher order cumulant expansion, naturally required controlling
matrix elements of resolvents chain even if the final quantity of interest was tracial.
This meant that a separate zigzag analysis was needed first for matrix elements 
(often called \emph{isotropic quantities} in this context) before the zigzag for averaged quantities.
While the key mechanisms for these two zigzag procedures are very similar, the actual
estimates for the isotropic part would be longer and more complicated.  We found a way to 
estimate all isotropic quantities in higher cumulant expansions by a single tracial quantity.

\subsection*{Notations} 
For $k \in \N$ we set $[k] := \{1, ... , k\}$, and $\langle A \rangle := d^{-1} \mathrm{Tr}A$, $d \in \N$,
for the normalised trace of a $d \times d$-matrix $A$.
For positive quantities $f, g$ we write $f \lesssim g$ resp.~$f \gtrsim g$ and mean that $f \le C g$ resp.~$f \ge c g$ for some $N$-independent constants $c, C > 0$ that may depend only on the basic control parameters $C_p$, see
\eqref{cp}   in Assumption~\ref{ass:entries} below and on $\kappa_0$ in~\eqref{eq:kappa}. 
Moreover, we will also write $f \sim g$ in case that $f \lesssim g$ and $g \lesssim f$.

We denote vectors by bold-faced lower case Roman letters $\boldsymbol{x}, \boldsymbol{y} \in \C^{N}$, for some $N \in \N$, and define 
\begin{equation*}
	\langle \boldsymbol{x}, \boldsymbol{y} \rangle := \sum_i \bar{x}_i y_i\,, 
	\qquad A_{\boldsymbol{x} \boldsymbol{y}} := \langle \boldsymbol{x}, A \boldsymbol{y} \rangle\,.
\end{equation*}
Matrix entries are indexed by lower case Roman letters $a, b, c , ... ,i,j,k,... $ from the beginning or the middle of the alphabet and unrestricted sums over those are always understood to be over $\{ 1 , ... , N\}$. 

We will use the concept  \emph{'with very high probability'},  meaning that for any fixed $D > 0$, the probability of an $N$-dependent event is bigger than $1 - N^{-D}$ for all $N \ge N_0(D)$. Also, we will use the convention that $\xi > 0$ denotes an arbitrarily small positive exponent, independent of $N$.
Moreover, we introduce the common notion of \emph{stochastic domination} (see, e.g., \cite{loc_sc_gen}): For two families
\begin{equation*}
	X = \left(X^{(N)}(u) \mid N \in \N, u \in U^{(N)}\right) \quad \text{and} \quad Y = \left(Y^{(N)}(u) \mid N \in \N, u \in U^{(N)}\right)
\end{equation*}
of non-negative random variables indexed by $N$, and possibly an additional parameter $u$
from a parameter space $U^{(N)}$, we say that $X$ is stochastically dominated by $Y$, if for all $\epsilon, D >0$ we have 
\begin{equation*}
	\sup_{u \in U^{(N)}} \P \left[X^{(N)}(u) > N^\epsilon Y^{(N)}(u)\right] \le N^{-D}
\end{equation*}
for large enough $N \ge N_0(\epsilon, D)$. In this case we write $X \prec Y$. If for some complex family of random variables we have $\vert X \vert \prec Y$, we also write $X = O_\prec(Y)$. 

\subsection*{Acknowledgments} G.C.~and L.E.~are grateful for the hospitality of GSSI, where part of this work has been carried out. The research of G. C. is supported by the Italian Ministry of University and Research (MUR) - Fondo Italiano per la Scienza (FIS3) - 2024 Call, project UBLOCO, CUP F53C25000940001, and also partially supported by the MUR Excellence Department Project MatMod@TOV awarded to the Department of Mathematics, University of Rome Tor Vergata, CUP E83C23000330006. Additionally, G.C. thanks INdAM (Istituto Nazionale di Alta Matematica “Francesco Severi”) and the group GNFM. L.E. and J.H. acknowledge partial financial support by the ERC Advanced Grant ``RMTBeyond'' No. 101020331. Additionally, J.H. is partially supported by the ERC Consolidator Grant ``ProbQuant'' (jointly with the Swiss State Secretariat for Education, Research and Innovation).
\section{Main results} \label{sec:mainres}
We consider  deformed Wigner matrices of the form
\begin{equation}
\label{eq:smalldef}
H_\lambda = H_0 + \lambda W \quad \text{for} \quad \lambda>0.
\end{equation}
Here, $H_0 \in \C^{N \times N}$ is an arbitrary self-adjoint deterministic matrix, which we will call \emph{deformation}, and we  make a technical assumption  that
\begin{equation} \label{eq:kappa}
\lambda^{-1} \Vert H_0 \Vert \lesssim N^{\kappa_0} \quad  \text{for some} \quad \kappa_0 \in [ 0, \infty) 
\end{equation}
which we consider fixed throughout the entire paper. When $H_0$ is diagonal,  \eqref{eq:smalldef} is the Rosenzweig-Porter (RP) model \cite{rosenzweig1960repulsion}, which has attracted lots of attention as a simple toy model to study the presence of a non-ergodic delocalized states. The model \eqref{eq:smalldef} can thus be considered as a generalization of the RP model to arbitrary deformations $H_0$. Some interesting physical feature appearing only for non-diagonal $H_0$ are discussed in Section~\ref{subsubsec:reentrloc} below. 
Moreover, $W$  in~\eqref{eq:smalldef} is an $N \times N$ Wigner matrix, i.e.~a self-adjoint random matrix
$W = W^*$, with independent matrix elements $\{ w_{ab}\; : \; a\le b\}$ on and above the diagonal.
The diagonal and off-diagonal single entry distributions are given by 
$w_{aa} \stackrel{\dd}{=} N^{-1/2} \chi_{\mathrm{d}}$ and $w_{ab}  \stackrel{\dd}{=} N^{-1/2} \chi_{\mathrm{od}}$ for $a< b$, where on the random variables $\chi_{\mathrm{d}}, \chi_{\mathrm{od}}$ we impose the following assumption. 
\begin{assumption}[Moment assumption]
	\label{ass:entries}
The off-diagonal distribution $\chi_{\mathrm{od}}$ is a real or complex centered random variable, $\E \chi_{\mathrm{od}} = 0$, with $\E |\chi_{\mathrm{od}}|^2 = 1$, and for $\sigma := \E \chi_{\mathrm{od}}^2$ we have $|\sigma|< 1$. The diagonal distribution  $\chi_{\mathrm{d}}$ is a real centered random variable, $\E\chi_{\mathrm{d}}=0$.  Furthermore, we assume the existence of high moments, i.e.~for any $p\in \N$ there exists $C_p$ such that
	\begin{equation}\label{cp}
	\E \big[|\chi_{\mathrm{d}}|^p+|\chi_{\mathrm{od}}|^p\big]\le C_p.
	\end{equation}
\end{assumption}
For $z\in\C\setminus \R$, the resolvent of $H_\lambda$ is denoted by $G_\lambda(z) := (H_\lambda - z)^{-1}$.
For random matrix models with a substantial random component, the resolvent becomes deterministic  as $N$ increases 
and its approximation is obtained as the unique  solution of the \emph{Matrix Dyson equation (MDE)} \cite{stab_MDE, MR2376207}.
For deformed Wigner matrices, the MDE is particularly simple. 
Concretely, let $M_\lambda = M_\lambda(z)$ for $z \in \C\setminus \R$ be the unique solution to the MDE 
\begin{equation} \label{eq:MDE}
- \frac{1}{M_\lambda(z)} = z - H_0 + \lambda^2 \langle M_\lambda(z) \rangle \quad \text{under the condition} \quad \Im M_\lambda(z) \, \Im z > 0 \,. 
 \end{equation}
 Finally, we denote $\rho_\lambda(z) := \pi^{-1} |\langle \Im M_\lambda (z) \rangle|$ and $\eta := |\Im z|$. 
 
Starting from \eqref{eq:MDE}, by simple standard computations, which we skip, we obtain the following bounds on $M_\lambda$. 
 \begin{lemma}[$M$-bounds] \label{lem:Mbounds}
The unique solution to \eqref{eq:MDE} admits the following bounds:
\begin{equation} 
\label{eq:boundM}
\langle |M_\lambda(z)|^2 \rangle \lesssim \frac{\rho_\lambda(z)}{\eta+\lambda^{2} \rho_\lambda(z)} \qquad \text{and} \qquad \Vert M_\lambda(z) \Vert \lesssim \frac{1}{\eta+\lambda^{2} \rho_\lambda(z)}.
\end{equation}
 \end{lemma}
 In addition, we have the bound
 \begin{equation}
 \label{eq:boundrho}
 \rho_\lambda\lesssim \lambda^{-1},
 \end{equation}
 as a consequence of a simple Schwarz inequality together with the first bound in \eqref{eq:boundM}. 
 These results hold uniformly in $H_0$ and $\lambda$ even though previously they were mostly considered
 in the typical mean field regime, i.e.~when
  $\kappa_0=0$ in~\eqref{eq:kappa}, thus \nc $\| H_0\|\lesssim \lambda$. In a dimensionless formulation,
 after trivially scaling 
 out $\lambda$, one may think about the mean field models as $\| H_0\|\lesssim 1$,   $\lambda\sim 1$. 
  In this case, the density $\rho_\lambda$ is bounded and $\| M_\lambda(z)\|$ is bounded
 in the \emph{bulk of the spectrum}, i.e.~where $\rho_\lambda(\Re z)\ge c>0$. In our strongly
 inhomogeneous regime these conventional bounds are replaced by~\eqref{eq:boundM}--\eqref{eq:boundrho}.  \nc

\subsection{Single resolvent local law} Our first result is a local law for the resolvent of~$H_\lambda$. Notably, we cover the entire strongly inhomogeneous regime as well, i.e.~do not need to assume any lower bound on $\lambda > 0$, apart from the  irrelevant technical \eqref{eq:kappa}.  
However, any  local law can only be effective  in the regime where $\eta$ is larger than the local eigenvalue
spacing, i.e., $N\eta\rho_\lambda(z)\gg 1$, which implicitly carries a constraint on $\lambda$ depending on $\eta$.

\begin{theorem}[Single resolvent local law]
\label{theo:main1}
Fix any small $\epsilon>0$. Let $W$ be a Wigner matrix satisfying Assumption~\ref{ass:entries}, let $H_0=H_0^*\in \C^{N\times N}$ be an arbitrary  Hermitian matrix, let $H_\lambda=H_0+\lambda W$ be as in \eqref{eq:smalldef}, and $M_\lambda$ be defined in \eqref{eq:MDE}. Then, for $\eta:=|\Im z|$, $\rho_\lambda(z):=\pi^{-1}|\langle \Im M_\lambda(z)\rangle|$, and 
\begin{equation} \label{eq:Gammadef}
\Gamma_\lambda(z):=\frac{\Im M_\lambda(z)}{\langle \Im M_\lambda(z)\rangle} \,, 
\end{equation}
 we have the averaged and isotropic local laws:
\begin{align}
 \label{eq:avLL}
\big|\langle \big(G_\lambda(z)-M_\lambda(z)\big)B\rangle\big|&\prec \frac{\Vert B \Vert_{\rm hs}}{N \eta} \, \sqrt{\frac{\langle \Gamma_\lambda(z) B B^* \rangle}{\langle BB^*\rangle}}, \\
 \label{eq:isoLL}
\big|\langle {\bm x},\big(G_\lambda(z)-M_\lambda(z)\big){\bm y}\rangle\big|&\prec \left(\sqrt{\frac{\rho_\lambda(z)}{N\eta}}\lVert {\bm x}\rVert\lVert {\bm y}\rVert\right)\sqrt{\frac{(\Gamma_\lambda(z))_{{\bm x}{\bm x}}(\Gamma_\lambda(z))_{{\bm y}{\bm y}}}{\lVert {\bm x}\rVert^2\lVert {\bm y}\rVert^2}}, 
\end{align}
uniformly in spectral parameters $N\eta\rho_\lambda(z)\ge N^\epsilon$ with $|z|\le \lambda N^{\kappa_0 + 1/\epsilon}$, deterministic matrices $B\in\C^{N\times N}$, and deterministic vectors ${\bm x}, {\bm y}\in\C^{N}$.
\end{theorem}

The bounds in \eqref{eq:avLL}--\eqref{eq:isoLL} consist of two factors: the first one reflects the
correct size in a mean-field model, while the second one,  depending on the spatial structure of $H_0$, expresses
the nontrivial deviation from the mean-field bound. As we will explain below, in
the previously studied  mean-field models
we have $\Gamma\sim 1$, hence \nc the second factor is $\sim 1$, and so the errors $1/(N\eta)$ and $\sqrt{\rho/(N\eta)}$, for the averaged and isotropic laws, respectively, recover known local laws \cite{slow_corr, alt2020correlated, erdHos2025cusp}.
Additionally, we point out that the \eqref{eq:avLL} can be simplified to
\[
\big|\langle \big(G_\lambda(z)-M_\lambda(z)\big)B\rangle\big|\prec
 \frac{\lVert B\rVert}{N\eta},
\]
by a simple trace inequality.

As a corollary of our single resolvent local law we show a delocalization bound for the eigenvectors of~$H_\lambda$. The standard proof of this corollary is postponed to the beginning of Section~\ref{sec:conseq}.
\begin{corollary}[Delocalization of eigenvectors]
\label{cor:deloc}
Under the assumptions of Theorem~\ref{theo:main1}, let $\lambda_i, {\bm u}_i$ be the eigenvalues and the corresponding orthonormalized eigenvectors of $H_\lambda$.
For any index $i\in [N]$ we have the eigenvector profile bound
\begin{equation}
\label{eq:evdeloc}
\big|\langle {\bm u}_i,{\bm x}\rangle\big|\prec\sqrt{\frac{[\Gamma_\lambda(\lambda_i+\ii \eta_{\mathfrak{f}}(\lambda_i))]_{{\bm x}{\bm x}}}{N}},
\end{equation}
uniformly in deterministic vectors ${\bm x}\in\C^N$, where $\eta_{\mathfrak{f}}(E)$ is the typical eigenvalue spacing and is defined implicitly via\footnote{This relation uniquely defines $\eta_{\mathfrak{f}}(E)$ 
since the map $\eta\mapsto \eta\rho_\lambda(E+\ii\eta)$ is strictly monotone for any fixed $E\in\R$.} $N\eta_{\mathfrak{f}}(E)\rho_\lambda(E+\ii\eta_{\mathfrak{f}}(E))=1$.
\end{corollary}
 Corollary \ref{cor:deloc} yields an \emph{upper} bound on the overlap of $\bm u_i$ with any deterministic vector $\bm x$, thus providing \emph{delocalization} information on $\bm u_i$. However, since $\sum_{j} [\Gamma_\lambda(z)]_{\bm x_j \bm x_j} = N$ for any orthonormal basis $\{\bm x_j\}_{j \in [N] }$, 
 and in many applications $[\Gamma_\lambda(z)]_{\bm x_j \bm x_j}$ is very small for most indices $j$, \nc
 \eqref{eq:evdeloc} can readily be used to infer \emph{lower} bounds on the overlap of $\bm u_i$ with certain deterministic directions $\bm x$, thus providing \emph{localization} information on $\bm u_i$ as well. We point out that all these information on eigenvector (de)localization is captured by the operator $\Gamma_\lambda$. Hence, extracting any explicit information is highly model (i.e.~$H_0$) dependent and, in particular, requires (approximately) solving the MDE \eqref{eq:MDE}. In Section \ref{sec:mobedge} below we carry out this task for a particular model exhibiting a \emph{mobility edge}. 

\medskip 

We now compare Theorem~\ref{theo:main1} and Corollary~\ref{cor:deloc} with previous results. 
 Local laws and eigenvector delocalization, as well as eigenvalue rigidity results\footnote{Standard rigidity results assert that $|\lambda_i-\gamma_i|\prec\eta_{\mathfrak{f}}(\gamma_i)$, where $\gamma_i$ is the $i$-th quantile of $\rho_\lambda(x)$
 defined in~\eqref{rholambda}.},  have been proven in great generality, even for certain correlated matrices \cite{alt2020correlated, slow_corr}, but under some strong assumptions. For our model~\eqref{eq:smalldef}, 
the standard setup in the literature
imposes two conditions\footnote{They are typically formulated by scaling out $\lambda$ in~\eqref{eq:smalldef}; alternatively, by setting $\lambda=1$.}: (a) $\| H_0\|\lesssim \lambda$ and (b) $\| M_\lambda(z)\|\lesssim 1/\lambda$ uniformly in $z\in \C_+$.
Note that the boundedness of $\lVert M_\lambda \rVert$ is a highly non-trivial assumption
that can be guaranteed only by assuming certain structure, e.g., piecewise H\"older continuity on $H_0$ (see, e.g., \cite[Section 9]{shape}). Assumptions (a) and (b) are necessary to give a fairly precise description of the solution 
of the MDE \eqref{eq:MDE} and its {\it self-consistent density of states} 
\begin{equation}\label{rholambda}
\rho_\lambda(x):=\pi^{-1}\langle \Im M_\lambda(x+\ii 0^+)\rangle,
\end{equation}
 which, in turn, are used as inputs to prove local laws. 
A key result in random matrix theory, often called \emph{shape analysis} \cite{shape}, is that under (a) and (b) we have:
\begin{itemize}

\item $\Gamma\sim \mathbf{1}$, i.e.~there is a constant $c>0$ such that $c\mathbf{1}\le \Gamma\le c^{-1}\mathbf{1}$ in sense of quadratic forms.

\item $\rho_\lambda(x)$ may admit only three qualitatively different behaviors. In the \emph{bulk} regime $\rho_\lambda$ is smooth, in the \emph{edge} regime $\rho_\lambda$ has a square root singularity, and in the \emph{cusp} regime $\rho_\lambda$ vanishes as a cubic root. No other singularities can occur.

\end{itemize}

The strong conditions (a) and (b) have already been relaxed 
previous to our current work in some special cases.
Namely, it was assumed that $H_0$ was diagonal (not necessarily bounded) 
 and it had a \emph{regular spectrum}, in the sense that locally around a certain $E_0\in\R$ the empirical measure of the eigenvalues of $H_0$, smoothed on a certain scale, was (i) either bounded from above and below (bulk regime) or (ii)  it had a square root behavior (edge regime). These conditions control the density $\rho_\lambda$, but
   do not necessarily imply that $\lVert M_\lambda(z)\rVert$ is bounded.  
   Under these assumptions, the results in Theorem~\ref{theo:main1} and Corollary~\ref{cor:deloc} have been proven
    in the bulk regime \cite{bourgade2017eigenvector, huang2019rigidity, landon2017convergence, landon2019fixed, von2019non} and an analog of \eqref{eq:avLL} in the edge regime as well \cite{adhikari2020dyson, aggarwal2024edge, landon2017edge}.  However,
     since these results were typically used as \emph{a priori}
    bounds for the analysis of the Dyson Brownian motion, the random matrix $W$ was assumed to have
    Gaussian entries.

   Our Theorem~\ref{theo:main1} and Corollary~\ref{cor:deloc} do not impose any condition either
    on $H_0$ or on the behavior of the density $\rho_\lambda$ or $M_\lambda$;  in particular,  both $\lVert M_\lambda \rVert$ and
$\rho_\lambda$ may become large in certain regimes, see \eqref{eq:boundM}--\eqref{eq:boundrho} for the optimal bounds.
The inhomogeneity in the spectral and spatial variables is reflected in the second factors in~\eqref{eq:avLL}--\eqref{eq:isoLL}.
     In fact, in this generality, the shape analysis in \cite{shape} does not necessarily hold, i.e.~the density may exhibit
     patterns other than the standard  bulk/edge/cusp regimes.
     Furthermore, 
       our results hold uniformly in the entire spectrum of $H_\lambda$. Finally,  we consider $W$ in \eqref{eq:smalldef} to be a general Wigner matrix satisfying Assumption~\ref{ass:entries} and not necessarily with Gaussian entries.  \nc

\subsection{Eigenstate Thermalization Hypothesis} 
We now state our second main result. To do so,  for every $i \in [N]$ and arbitrarily small  but fixed $\xi > 0$, we introduce the shorthand notation
\begin{equation} \label{eq:gammahat}
M_i := M_\lambda(\widehat{\lambda}_i) 
 \quad \text{with} \quad \widehat{\lambda}_i := \lambda_i + \ii N^\xi \eta_{\mathfrak{f}}(\lambda_i),
\end{equation}
 where $\lambda_1\le \lambda_2\le \ldots \le\lambda_N$ are the eigenvalues of $H_\lambda$. 
Moreover, for every pair of indices $i,j \in [N]$,
 we define the $(i,j)$-\emph{regularization} of an observable $A \in \C^{N \times N}$ as
\begin{equation} \label{eq:regobsij}
\mathring{A}^{i,j} := A - \theta_{ij}\frac{\langle M_i A M_j^* \rangle}{\langle M_i M_j^* \rangle} \mathbf{1} \quad \text{with} \quad \theta_{ij} := \theta(\lambda^2 \Re \langle M_i M_j^* \rangle - 1/2),
\end{equation}
where we recall that $\theta(x) = (1 + \sgn(x))/2$ is the Heaviside function\footnote{The cutoff function
$\theta_{ij}$ expresses the fact that regularization is a relevant concept only in a certain
regime of the two spectral parameters. It allows for clean, unified final formulas but the proofs 
will be different for the  $\theta_{ij}=0$ and $\theta_{ij}=1$ regimes. The reader may 
first focus on the more interesting $\theta_{ij}=1$ regime.}.
Finally, denoting $\Gamma_i :=\Im M_i /\langle \Im M_i \rangle$ and abbreviating $\mathring{A} = \mathring{A}^{i,j}$, we introduce the control parameter\footnote{We point out that this notation is different compared to \cite[Eq. (5.5)]{cipolloni2024out} where we used $\mathfrak{s}_2$ instead of $\mathfrak{s}_2^2$ to denote an analogous quantity.} $\s_2(i,j;A)$ via
\begin{equation} \label{eq:s2ijdef}
	\begin{split}
		\big(\s_2(i,j;A)\big)^2:=\langle\Gamma_i \mathring{A} \Gamma_j \mathring{A}^*\rangle + \theta_{ij} \, \lambda^2\bigg| \frac{\langle M_i \mathring{A} \Gamma_j \rangle \langle \Gamma_i \mathring{A}^* M_j \rangle }{1 - \lambda^2 \langle M_i M_j\rangle}\bigg|  + (1 - \theta_{ij}) \, \lambda^4 \max_{*} |\langle M_i^{(*)} \mathring{A} M^{(*)}_{j} \rangle|^2 \langle \Gamma_i \, \Gamma_j \rangle ,
	\end{split}
\end{equation}
where $\max_*$ is understood to be the maximum over all four  possible constellations of adjoints taken in $|\langle M_i^{(*)} \mathring{A} M^{(*)}_{j} \rangle|^2$. 
This quantity will determine the bound on quadratic forms of eigenvectors of $H_\lambda$ and deterministic matrices $A\in \C^{N \times N}$, which is the content of the following result on the \emph{Eigenstate Thermalization Hypothesis}. Note that no  lower bound on $\lambda$ is assumed (apart from \eqref{eq:kappa});
the natural deterioration of the bound in the  small $\lambda$ regime is expressed 
through  the $\lambda$-dependence of the control quantity $\s_2$.

\begin{theorem}[Eigenstate Thermalization Hypothesis]
\label{theo:ETH}
Let $W$ be a Wigner matrix satisfying Assumption~\ref{ass:entries}, let $H_0=H_0^*\in\C^{N\times N}$ be an arbitrary deterministic matrix, and let $H_\lambda=H_0+\lambda W$ be as in \eqref{eq:smalldef}. Let ${\bm u}_i$ be the $\ell^2$-normalized eigenvectors of $H_\lambda$.  Denoting $\Gamma_i :=\Im M_i /\langle \Im M_i \rangle$ 
for any index $i\in [N]$, 
and recalling the notations above, we have
\begin{equation}
 \label{eq:ETH}
\left| \big\langle \bm u_i, A \bm u_j \big\rangle - \delta_{ij} { \theta_{ij}}\langle\Gamma_i A\rangle\right| \prec \frac{\s_2(i,j;A)}{\sqrt{N}},
\end{equation}
uniformly in deterministic matrices $A\in \C^{N\times N}$ and in index pairs $(i,j)$.   
\end{theorem}
The proof of Theorem \ref{theo:ETH} is given in Section \ref{sec:proofmain} below. A key input for its proof will be a local law bound for the products of two resolvents, see Theorem~\ref{thm:main2G} later.
\begin{remark}
We have several comments on Theorem \ref{theo:ETH}, that are easy to verify.
\begin{itemize}
	\item[(a)] \textnormal{[Eigenvector delocalization]} For $i=j$ and $A = \ket{\bm x} \bra{\bm x}$ being a rank one projection, we recover eigenvector delocalization as formulated in \eqref{eq:evdeloc}. 
\item[(b)] \textnormal{[Wigner matrices]} In case of $H_0 = 0$ and $\lambda = 1$, we recover the previously proven ETH for Wigner matrices \cite{cipolloni2023eigenstate} (see also \cite{cipolloni2021eigenstate, cipolloni2022rank}):
\begin{equation*}
\max_{i,j \in [N]}\left| \big\langle \bm u_i, A \bm u_j \big\rangle - \delta_{ij} \langle A\rangle\right| \prec \frac{\langle |A - \langle A \rangle|^2\rangle^{1/2}}{\sqrt{N}},
\end{equation*}
since we have that $\delta_{ij} \theta_{ij} = \delta_{ij}$, $\Gamma_i = \mathbf{1}$ and $\s_2(i,j;A) = \langle |A - \langle A \rangle|^2\rangle^{1/2}$. 
\item[(c)] \textnormal{[Deformed Wigner matrices]} Let $\lambda = 1$ and consider $H_0$ which ensures the deterministic approximation $M$ to be bounded in operator norm, $\Vert M \Vert \lesssim 1$.\footnote{This is the case, e.g., if $H_0$ has Hölder-$1/2$ regularly spaced eigenvalues; see \cite{shape} and the discussion below Corollary~\ref{cor:deloc}.} Then the $\lambda_i$'s in the rhs. of \eqref{eq:ETH} 
 (hidden in the definition of $M_i$ \eqref{eq:gammahat})
 can be replaced by the corresponding quantiles, $\gamma_i$ of the self-consistent density
 (see the discussion below Corollary~\ref{cor:deloc}), and so we have
\begin{equation} \label{eq:ETHex}
\left| \big\langle \bm u_i, A \bm u_j \big\rangle - \delta_{ij} \langle\Gamma_i A\rangle\right| \prec \frac{\big\langle \big|\mathring{A}^{i,j}\big|^2\big\rangle^{1/2}}{\sqrt{N}} \times  \frac{1}{|1 - \langle M_i M_j \rangle|^{1/2}} \,. 
\end{equation}
   Since for indices $i,j$ in the bulk regime $|1 - \langle M_i M_j \rangle|\sim 1$, our result extends and improves upon \cite[Theorem 2.7]{cipolloni2023gaussian} (see also \cite[Theorem 3.4]{bao2021equipartition}), which only considered the bulk of the spectrum and estimated the observable in operator norm. 
Note that outside of the bulk spectrum there is a deterioration\footnote{We point out that the $N^{-1/2}$-bound can be obtained everywhere in the spectrum if we project $A$ in a specific direction~\cite{corrETH}.} of the usual $N^{-1/2}$ ETH-bound (i.e.~\eqref{eq:ETHex} without the $|1 - \langle M_i M_j \rangle|^{-1/2}$ factor): For example, for $i=j$ we have the deterioration factor
$|1 - \langle M_i M_i \rangle|^{-1/2}\sim N^{1/6}$ 
 in the edge regime  and  $|1 - \langle M_i M_i \rangle|^{-1/2}\sim N^{1/4}$  in the cusp regime. 
   In~\cite{corrETH}, this deterioration is fully analyzed and its optimality (in the cusp regime) is also shown. At regular edges, a sophisticated study in \cite{corrETH}, heavily relying on shape analysis, allows to recover the $N^{-1/2}$ ETH-bound. A similar improvement could, in principle, be investigated here too, under further conditions on the deformation $H_0$. However, as our main focus is the completely general situation without any shape analysis, this is beyond the scope of the current paper.   
\item[(d)] \textnormal{[Going from $\mathring{A}$ to $A$]} Our definition of $\s_2(i,j;A)$ involves only $\mathring{A}$ on the rhs.~of \eqref{eq:s2ijdef}. We have the following bound in terms of $A$ only:
\begin{align}
	(\mathfrak{s}_2(i,j;{A}))^2
	\lesssim \left( \langle \Gamma_i A \Gamma_j A^* \rangle + \lambda^4 \max_* \big| \langle M_i^{(*)} A M_j^{(*)}\rangle \big|^2 \langle \Gamma_i \Gamma_j  \rangle \right) \times \left(1 + \frac{\theta_{ij}}{|1 - \lambda^2 \langle M_i M_j\rangle|}\right) \label{eq:Abound}
\end{align}
\item[(e)] \textnormal{[Mean field bound]} We also have the mean field upper bound
\begin{equation*}
	(\mathfrak{s}_2(i,j;{A}))^2 \lesssim \frac{\Vert A \Vert^2}{\lambda^2 \rho_i \rho_j} \left(1 + \frac{\theta_{ij}}{|1 - \lambda^2 \langle M_i M_j\rangle|}\right),
\end{equation*}
which is much simpler than~\eqref{eq:Abound} but in general is far from optimal. 
\end{itemize}
\end{remark}

\section{Consequences of local laws and further results}
\label{sec:conseq}

In this section we first present the proof of Corollary~\ref{cor:deloc} and then discuss several further consequences of the single and multi-resolvents local laws. More precisely, in Section~\ref{sec:mobedge} we discuss localization/delocalization properties of the eigenvectors of $H_\lambda$, and in particular the emergence of a mobility edge. In Section~\ref{subsubsec:reentrloc} we connect our local law in \eqref{eq:isoLL} and the bound in \eqref{eq:evdeloc} to the phenomenon of \emph{re-entrant localization} \cite{ghosh2025reentrant} for a generalization of the Rosenzweig-Porter model. In Section~\ref{sec:comments} we discover several further possible results that follow by using similar proof techniques, but that we do not pursue here for the sake of brevity. Finally, in Section~\ref{sec:non-Herm} we comment on related questions for non-Hermitian random matrices.

\begin{proof}[Proof of Corollary~\ref{cor:deloc}]
We note that in the spectral regime $N\eta\rho_\lambda(z)\ge 1$ we have
\begin{equation}
\label{eq:gridargu}
\big|\langle {\bm x},\big(G_\lambda(z)-G_\lambda(w)\big){\bm y}\big|\le \frac{N^2}{\lambda^2} |z-w| \lVert {\bm x}\rVert \lVert {\bm y}\rVert,
\end{equation}
where we used  $\| G_\lambda(z)\|\le 1/\eta$,  and that $\rho_\lambda\lesssim \lambda^{-1}$. As a consequence of the Lipschitz continuity of the resolvent, using a standard grid argument, we obtain that \eqref{eq:isoLL} holds simultaneously for all $z$'s, satisfying the constraint below \eqref{eq:isoLL}, with very high probability. We thus have
\begin{equation}
\label{eq:boundelnew}
\big|\langle {\bm u}_i,{\bm x}\rangle\big|^2\le N^\xi \eta_\mathfrak{f}(\lambda_i) \langle {\bm x},\Im G_\lambda(\lambda_i+N^\xi\ii\eta_\mathfrak{f}(\lambda_i)){\bm x}\rangle\le 2N^\xi \frac{[\Gamma_\lambda(\lambda_i+\ii N^\xi\eta_{\mathfrak{f}}(\lambda_i))]_{{\bm x}{\bm x}}}{N},
\end{equation}
with very high probability. Here in the first step we used the spectral representation of $\Im G$ and 
in the \nc last step we used \eqref{eq:isoLL} and that $N\eta_{\mathfrak{f}}(\lambda_i)\rho_\lambda(\lambda_i+\ii\eta_{\mathfrak{f}}(\lambda_i))=1$. Additionally, by monotonicity of $\rho_\lambda/\eta$ and $\eta\rho_\lambda$, we can replace $\Gamma_\lambda(\lambda_i+\ii N^\xi\eta_{\mathfrak{f}}(\lambda_i))$ in the rhs. of \eqref{eq:boundelnew} with $N^{2\xi}\Gamma_\lambda(\lambda_i+\ii \eta_{\mathfrak{f}}(\lambda_i))$.
\end{proof}

\subsection{Mobility edge for $H_\lambda$}
\label{sec:mobedge}
 We now discuss the bound \eqref{eq:evdeloc} in a specific case, illustrating the emergence of a
  threshold between localized and delocalized states, called the \emph{mobility edge}.   Localized states, 
  by definition, are essentially supported on  finitely many ($N$-independent) sites. In contrary,
  the support of the delocalized states grows with $N$.
  If the support is comparable with $N$ we call the state 
  \emph{fully delocalized}, otherwise we talk about \emph{non-ergodic delocalization} 
 (see the discussion below \eqref{eq:Lasympt} for a precise definition). 
  
  We now illustrate these concepts
  in a concrete example. 
  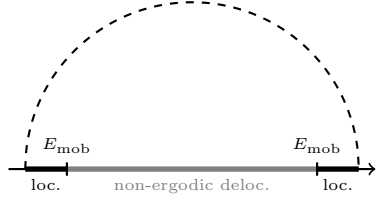
\begin{figure}[h]
  	\begin{tikzpicture}[scale=1.1]
  		
  		\def\R{2}
  		\draw[thick, dashed, black]
  		plot[domain=-2:2,samples=400]
  		({\x},{sqrt(\R*\R - \x*\x)});
  		
  		\draw[thick,->, black] (-2.2,0) -- (2.2,0);
  		
  		\draw[line width=2pt,black] (-2,0) -- (-1.5,0);
  		\draw[line width=2pt,gray]  (-1.5,0) -- (1.5,0);
  		\draw[line width=2pt,black] (1.5,0) -- (2,0);
  		
  		\foreach \x in {-1.5,1.5} {
  			\draw[thick,black] (\x,0.08) -- (\x,-0.08);
  		}
  		
  		\node[above] at (-1.5,0.08) {\tiny \color{black}$E_{\mathrm{mob}}$};
  		\node[above] at (1.5,0.08) {\tiny \color{black}$E_{\mathrm{mob}}$};
  		\node[below] at (0,0) {\color{gray} \tiny non-ergodic deloc.};
  		\node[below] at (-1.75,0) {\color{black} \tiny loc.};
  		\node[below] at (1.75,0) {\color{black} \tiny loc.};
  	\end{tikzpicture}
  	\caption{Mobility edge: For $N^{-1/2} \ll \lambda \ll N^{-1/6}$, localized and non-ergodic delocalized states coexist with threshold satisfying $\big| |E_{\rm mob}| - 2\big| \sim \frac{1}{N \lambda^2}$. The semicircular density of states is 
	shown by a dashed curve.}
  	\label{fig:mobedge}
  \end{figure}
  For this purpose, we take $H_0$ to be a diagonal matrix with its eigenvalues being sampled by a Wigner matrix,
  then fixed and put on the diagonal in increasing order. For convenience, to keep the spectrum contained in $[-2,2]$ for the random matrix as well, we rescale $H_0$ by a factor $\sqrt{1 - \lambda^2}$ for $\lambda \in [0,1]$, i.e., we consider\footnote{We warn the reader that in this subsection, all quantities no longer carry their dimensionality.}
\begin{equation}
H_\lambda = \sqrt{1 - \lambda^2} H_0 + \lambda W \,. 
\end{equation}

Since the eigenvalues of $H_0$ are well controlled with very high probability by the standard rigidity result,  
and we can rely on the above mentioned shape analysis for $H_\lambda$,  we can cast our eigenvector bounds in terms of the deterministic \emph{quantiles} $(\gamma_i)_{i \in [N]}$ of the semicircular density instead of the random eigenvalues. Indeed, recall the implicit definition of the semicircular $\gamma_i$'s, 
\begin{equation}
	\label{eq:defquant}
	\int_{-\infty}^{\gamma_i} \rho_{\rm sc}(x)\,  \dd x=\frac{i}{N} \,, \qquad \rho_{\rm sc}(x) := \frac{\sqrt{[4-x^2]_+}}{2 \pi} \,. 
\end{equation}
Then, by explicit deterministic computations, for the $i$-th entry of the normalized eigenvector $ {\bm u}_{i_0}$
of $H_\lambda$ we have
\begin{equation}
	\label{eq:del}
	\big|\langle {\bm u}_{i_0},{\bm e}_i\rangle\big|^2\prec \frac{\omega_\lambda(i_0)}{\big(f_\lambda(i,i_0)\big)^2 + \big(\omega_\lambda(i_0)\big)^2} \quad \text{with} \quad \omega_\lambda(i_0) := 1 + N \lambda^2 \rho_{i_0} \,.
\end{equation}
Here we abbreviated $\rho_{i_0}:=\rho_{\rm sc}(\gamma_{i_0}+\ii \eta_{\mathfrak{f}}(\gamma_{i_0}))$ and denoted the harmonic extension of the semicircular density to the complex upper half plane by the same symbol. Moreover, the function $f_\lambda(i, i_0)$ is given by
\begin{equation} \label{eq:fdef}
	f_\lambda(i,i_0) = N \rho_{i_0} \left[\left(1 - \tfrac{\lambda^2}{2}\right) (\gamma_{i_0} - \gamma_{i}) + \left(  1 - \tfrac{\lambda^2}{2} - \sqrt{1 - \lambda^2}\right)\gamma_i  \right] \,. 
\end{equation}
 In short, for small $\lambda$, the profile $i\mapsto \big|\langle {\bm u}_{i_0},{\bm e}_i\rangle\big|^2$
exhibits an almost Cauchy-like quadratic decay in $\gamma_{i_0} - \gamma_{i}$ and note that
$|\gamma_{i_0} - \gamma_{i}| \sim |i-i_0|/N$ in the bulk. \nc

To quantify the localization/delocalization of eigenvectors in the standard basis, for (small) fixed $\epsilon > 0$, we introduce the \emph{$\epsilon$-essential support} of $\bm u_{i_0}$ as
\begin{equation*}
L_\lambda(i_0) = L_\lambda^{(\epsilon)}(i_0) := \inf \bigg\{ |\mathcal{J}| : \sum_{j \in \mathcal{J}}	\big|\langle {\bm u}_{i_0},{\bm e}_j\rangle\big|^2  \ge 1 - \epsilon  \bigg\} \,. 
\end{equation*}
Then, by simple explicit computations involving the definitions \eqref{eq:defquant}, \eqref{eq:del}, and \eqref{eq:fdef}, we find that 
\begin{equation} \label{eq:Lasympt}
L_\lambda(i_0) \sim 1 + N \lambda^2 (4 - |\gamma_{i_0}|^2) + N \lambda^6 |\gamma_{i_0}|^2 
\end{equation}
with implicit constants depending only on $\epsilon$. We call an eigenvector $\bm u_{i_0}$ (i) \emph{localized} if $L_\lambda(i_0) \sim 1$, (ii) \emph{non-ergodic delocalized} if $1 \ll L_\lambda(i_0) \ll N $, and (iii) \emph{delocalized} if $L_\lambda(i_0) \sim N$. Note that \eqref{eq:Lasympt} interpolates between the $N \lambda^2$ (bulk) and $N \lambda^6$ (edge) spatial spreading of the eigenvector $ {\bm u}_{i_0}$ 
as $\gamma_{i_0}$ approaches to the edges $\pm 2$. In fact, from \eqref{eq:Lasympt}, we can immediately read off the support of eigenvectors in the bulk and edge, summarized in Table~\ref{tab:1}. 

	\begin{table}[h]
	\centering 
	\begin{tabular}{c|c|c}
		&	\centering \textbf{{Edge}} $\rho_{i_0} \sim N^{-1/3}$& \textbf{{Bulk}} $\rho_{i_0} \sim 1$   \\ \hline
		{localized} &	$\lambda \ll N^{-1/6}$ &  $\lambda \ll N^{-1/2}$ \\ \hline 
		{non-ergodic deloc.}& 	$N^{-1/6} \ll \lambda \ll 1$ & $N^{-1/2} \ll \lambda \ll 1$   \\ \hline
		{full delocalized}& 	$\lambda \gtrsim 1$ & $\lambda \gtrsim 1$ 
	\end{tabular}
	\caption{Summary of eigenvector behavior.}
	\label{tab:1}
\end{table}

In particular, we infer a \emph{mobility edge}\footnote{By symmetry of the semicircular spectrum there are two mobility edges at $\pm E_{\rm mob}$.}, $E_{\rm mob}$, illustrated in Figure \ref{fig:mobedge}, between the localized and non-ergodic delocalized phase in the regime $N^{-1/2} \ll \lambda \ll \lambda^{-1/6}$ satisfying 
\begin{equation*}
2 - |E_{\rm mob}| \sim \frac{1}{N \lambda^2} \,. 
\end{equation*}

Finally, in Figure \ref{fig:simulate} below we illustrate the behavior of edge and bulk eigenvectors by a simulation.  Note that the theoretical upper bounds (red curves) properly describe the shape 
of the simulated blue curves up to a constant factor that our bound~\eqref{eq:del} does not track. 

\begin{figure}[ht]
	\centering
	\includegraphics[width=0.6\linewidth]{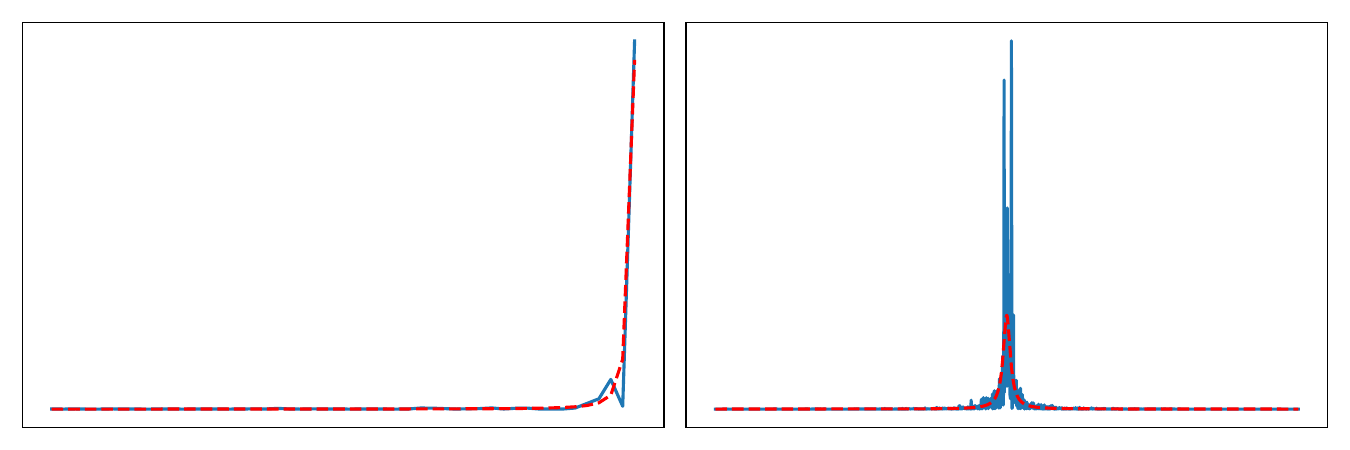}
	\vspace{2mm} 
	\includegraphics[width=0.6\linewidth]{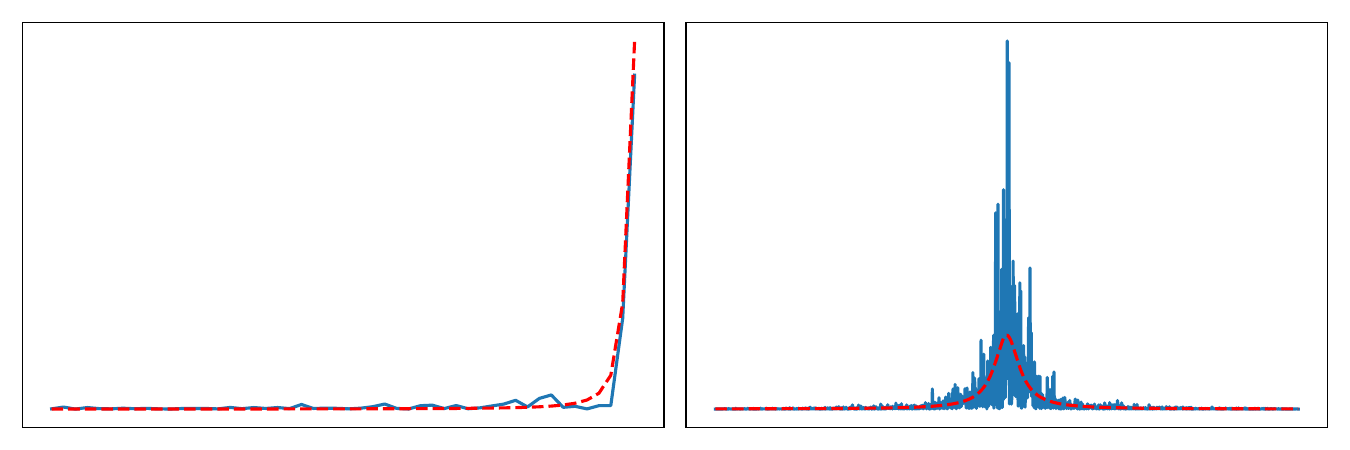}
	\vspace{2mm} 
	\includegraphics[width=0.6\linewidth]{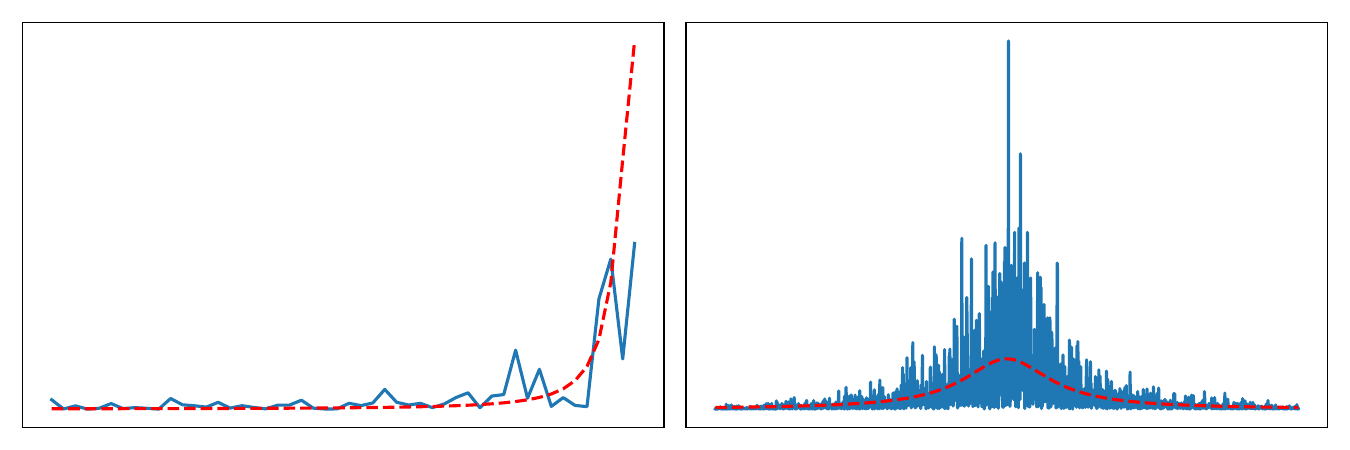}
	\caption{Simulation (blue) of the eigenvector overlap curve $i \mapsto |\langle \bm u_{i_0}, \bm e_i\rangle|^2$ for $i_0 =N$ (edge eigenvector, left column) and $i_0 = N/2$ (bulk eigenvector, right column) for $N=2000$ and three different $\lambda$-values: (i) $\lambda = N^{-1/4}$ in the first row, (ii) $\lambda = N^{-1/6}$ in the second row, (iii)  $\lambda = N^{-1/12}$ in the third row. The red curves sketch the upper bound on the rhs.~of \eqref{eq:del}. For illustrational purposes, the figures in the left column show only the indices $\{1950, ..., 2000\}$, while the figures in the right column cover the entire range of indices $\{1,...,2000\}$.  }
	\label{fig:simulate}
\end{figure}

\subsection{Re-entrant localization}
\label{subsubsec:reentrloc}  Following the very recent physics paper~\cite{ghosh2025reentrant}, 
consider the special $N\times N$ Hermitian matrix $H$ with entries
\begin{equation}
\label{eq:tridiagh}
H_{mn}:=h_n \delta_{m,n}+N^k(\delta_{m,n+1}+\delta_{m,n-1})+R_{mn} N^{-\gamma/2}.
\end{equation}
Here the real random variables $h_n$ are centered and have a joint distribution such that the following holds. The ordered statistics $h_{n_1}< h_{n_2}<\dots<h_{n_N}$ are such that the level spacings $s_i=h_{n_i+1}-h_{n_i}$, $i\in [N-1]$, 
are i.i.d.~Pareto distributed, i.e.~they are distributed according to the density
\[
p(s)=\frac{d\delta_{\mathrm{typ}}^d}{s^{d+1}}\theta(s-\delta_{\mathrm{typ}}),
\]
where $\theta(\cdot)$ is the Heaviside function. The parameter $d>0$ can be regarded as the fractal dimension of the disorder, while the cut-off parameter $\delta_{\mathrm{typ}}$ corresponds the typical level spacing of the $h_n$'s, and we choose it to be $\delta_{\mathrm{typ}}\sim N^{-1/d}$.  Additionally, $R=R^*$ is a
random matrix with independent entries (up to symmetry), \nc the $R_{mn}$'s are 
independent of the $h_n$'s, centered, normalized to $\E |R_{mn}|^2=1$. The two parameters in~\eqref{eq:tridiagh} satisfy $k\in \R, \gamma>0$. In particular, $H$ is a Laplacian with onsite (diagonal) i.i.d. potential plus a (small) Wigner matrix. In \cite{ghosh2025reentrant} it was observed that the orthonormal eigenvectors ${\bm u}_i$ of $H$ exhibit an interesting behavior, as $k,\gamma$ vary, called “re-entrant localization". In particular, by considering the inverse participation ratio,
\begin{equation}
\label{eq:IPR}
\mathrm{IPR}_q:=\sum_{a=1}^N |{\bm u}_i(a)|^{2q},
\end{equation}
for $q=2$, they (non-rigorously) showed the surprising effect that, as $k$ grows,  the system initially tends to localization from a non-ergodic delocalized phase, but then re-enters the fully ergodic phase.   In other words, contrary to basic intuition, in a certain regime the increasing
effect of the Laplacian enhances localization. This is because delocalization in the model~\eqref{eq:tridiagh}
for negligible Laplacian (very negative $k$) is of a different type than the approximate plane waves
exhibit for the Laplacian itself (very positive $k$) and the transition between these two types
of states is not monotone.

The Hamiltonian \eqref{eq:tridiagh} can be seen to be a particular case of the model \eqref{eq:smalldef}:
\begin{equation}
H=H_\lambda:=H_0+\lambda W, \qquad (H_0)_{mn}:=h_n \delta_{m,n}+N^k(\delta_{m,n+1}+\delta_{m,n-1}), \qquad \lambda:=N^{(1-\gamma)/2}.
\end{equation}
Combining our results \eqref{eq:isoLL}--\eqref{eq:evdeloc} with the Dyson Brownian motion (DBM) analysis\footnote{We point out that the DBM shows that moments of $|{\bm u}_i(a)|^2$ are well approximated by their local spectral average, i.e.~by the resolvent, which can in turn be well approximated by $\Gamma_\lambda$, as a consequence of \eqref{eq:isoLL}.} of \cite{benigni2020eigenvectors}, for bulk eigenvectors, we can write
\begin{equation}
\mathrm{IPR}_2=\big(1+o(1)\big)\sum_{a=1}^N\big(\big[(\Gamma_\lambda(z_i)\big]_{aa}\big)^2, \qquad\quad z_i:=\gamma_i+\ii N^\xi\gamma_i,
\end{equation}
for an arbitrary small $\xi>0$. Writing the explicit form of $\Gamma_\lambda(z_i)$, in terms of the eigenvalues and eigenvectors of $H_0$, we can rigorously deduce the phase diagram in \cite[Figure 1]{ghosh2025reentrant}. We point out that from our analysis 
a similar result for $\mathrm{IPR}_q$ for any integer $q\ge 2$ also follows, but we omit the details.

\subsection{Further results}
\label{sec:comments}
We present some additional results that can easily be obtained by minor modifications of the techniques presented in this paper.

\begin{itemize}

\item \textbf{No outliers:} Under the assumption of a \emph{shape analysis},  as explained below Corollary~\ref{cor:deloc},
 one can readily get an improvement of \eqref{eq:avLL} outside of the support of $\rho_\lambda$, where basically $\eta$ is replaced with the distance of $z$ to the support of the self consistent density $\rho_\lambda$. 
  In particular, one can show that there are no \emph{outliers}, i.e.~eigenvalues far away from this support, 
   and that the number of eigenvalues in each connected component of the support of $\rho_\lambda$ is a deterministic integer with very high probability (see, e.g., \cite[Eq. (3.18)]{campbell2024spectral} and \cite[Theorem 2.8]{erdHos2025cusp} for analogous results).

\medskip

\item \textbf{Toy model for quantum chaos:} A fundamental input for Theorem~\ref{theo:ETH} is the bound \eqref{eq:2G2} below for the product of two resolvents. By techniques analogous to those developed in this work we can prove not only a bound but
a local law for products of resolvents of arbitrary length, i.e.~we can identify its deterministic leading order 
behavior. \nc
 Once these local laws are achieved, we can easily obtain 
 physically relevant statements related to multi-resolvent local laws. For example, one can study the long time
asymptotic for thermalization of several observables \cite{cipolloni2021thermalisation}, out-of-time correlators (OTOC) \cite{cipolloni2024out}, Loschmidt echo \cite{erdHos2025loschmidt}, Prethermalization \cite{erdHos2025prethermalization},  the Spectral Form Factor  (SFF) \cite{cipolloni2023spectral}, and show the emergence of the Equipartition principle for these models  \cite{bao2021equipartition, cipolloni2023gaussian}. Now we can study all these quantities for the model $H_\lambda$, earlier only the case $\lambda=1$, $\| H_0\|\lesssim 1$ 
 was accessible.

Just as a concrete example, consider a diagonal $H_0 = \mathrm{diag}(\mu_1, ... , \mu_N)$ which we assume to have a smooth limiting density of states $\rho_0$.  Assume that $\mu_{i_0}$ lies in the \emph{bulk}, $\rho_0(\mu_{i_0}) > c$ with some small $c>0$, and denote $\alpha := 2 \pi \rho_0(\mu_{i_0})$. Then, for $\mathcal{J} \subset [N] $ with $|\mathcal{J}| \ge c N$ and (mildly) $N$-dependent $\lambda$, we have that 
\begin{equation*}
\,\,\,\,\,\,\,\,\,\,\,\,\,\,\,\,\,\,\,\,\,\,\,\frac{1}{|\mathcal{J}|} \sum_{j \in \mathcal{J}} \big| \big(\ee^{\ii t H_{\lambda}}\big)_{i_0 j} \big|^2 \approx \big(1 - \ee^{- \alpha \lambda^2 t}\big) \, \frac{\sum_{j \in \mathcal{J}} \left((\mu_j - \mu_{i_0})^2 + (\alpha \lambda^2)^2\right)^{-1}}{|\mathcal{J}| \sum_{j \in [N]} \left((\mu_j - \mu_{i_0})^2 + (\alpha \lambda^2)^2\right)^{-1}} + \ee^{- \alpha \lambda^2 t} \frac{\mathbf{1}(i_0 \in \mathcal{J})}{|\mathcal{J}|}, 
\end{equation*}
which shows quantum diffusion on kinetic time scales $t \sim 1/(\alpha \lambda^2)$. 
This follows similarly to the proof of \emph{prethermalization} in \cite{erdHos2025prethermalization} using an analog of Theorem~\ref{theo:main1} for products of two resolvents (or, alternatively, a variant of \eqref{eq:2G2A} in Theorem \ref{thm:2GLL} below without $\Im G$'s). 
\end{itemize}

\medskip

Finally, we mention other results that can be potentially achieved using similar techniques, but for which we omit the details for brevity.

\begin{itemize}

\item \textbf{Universal statistics of eigenvalues/eigenvectors:} In this work we focused on local laws for resolvents, however, using these bounds as an input one can prove universality of the distribution of eigenvalues and eigenvectors, if the shape analysis as above is assumed. For eigenvectors, one can prove the Gaussianity of $|{\bm u}_i(x)|^2$ and $\langle {\bm u}_i, A{\bm u}_j\rangle$, for bulk indices,  following the proof of \cite{benigni2020eigenvectors, bourgade2017eigenvectorsparse, marcinek2022high} (see also the very recent work \cite{benigni2026quantitative}) and of \cite{benigni2024fluctuations, cipolloni2023gaussian, cipolloni2022normal} respectively. For eigenvalues, in the bulk and at the edge of the spectrum universality follows analogously to \cite{landon2017convergence, landon2017edge, landon2019fixed} (or adapting the more recent \cite{bourgade2021extreme}), while for the cusp regime one can likely combine the technique of \cite{bourgade2021extreme} with the semicircular flow analysis of \cite[Section 4]{cipolloni2019cusp} to prove universality in this regime (and all the intermediate ones) as well.

\medskip

\item \textbf{Diverging rank perturbations:} The work \cite{peche2006largest} considered matrices of the form $GUE+H_0$, where $H_0$ is a deterministic deformation of diverging rank $N^\alpha$, for $\alpha\in (0,1)$. In this setting it was shown that 
\emph{“mini-bulks”}, i.e.~bumps of the limiting density of $GUE+H_0$ containing $\ll N$ eigenvalues, can emerge outside of the 
main bulk spectrum. In \cite{peche2006largest} it is also shown that locally one still see the emergence Wigner-Dyson statistics \cite{mehta1967random} in the bulk and Tracy-Widom statistics \cite{tracy1994level} at the edge
of the mini-bulk.  For $W+H_0$, when $W$ is a general Wigner matrix, the only known result in this direction is in \cite{huang2018mesoscopic}, where a weak form of rigidity is proven. We believe that the techniques developed in our work can be used to prove rigidity, delocalization/ETH bounds, and eigenvalues/eigenvectors universality in these mini-bulks. We also believe that with our local law as an input one could also study the emergence of the Pearcey \cite{tracy2006pearcey} process in the possible cusp regimes of the spectrum.
\end{itemize}

\subsection{Non-Hermitian Rosenzweig-Porter model}
\label{sec:non-Herm}

In this work we focused on Hermitian random matrices. However, it is very natural to also consider non-Hermitian matrices of the form
\begin{equation}
\label{eq:nonhermRP}
X_\lambda:=X_0+\lambda X,
\end{equation}
where $X_0$ is a general (possibly non-Hermitian) deterministic matrix (\emph{deformation})  and $X$ is a matrix with i.i.d. entries (e.g.~no symmetry assumption on $X$), satisfying a moment assumption similar to Assumption~\ref{ass:entries}. It is natural to ask localization/delocalization properties of the left/right eigenvectors of $X_\lambda$. For example, in the case when $X_0$ is diagonal (with i.i.d. complex standard Gaussian entries) and when $\lambda\ll 1$, the model \eqref{eq:nonhermRP} can be thought as a non-Hermitian version of the Rosenzweig-Porter model \cite{de2022non}. In \cite{de2022non}, it was observed that unlike in the Hermitian case, where for $\lambda:=N^{-\gamma}\ll 1$ the eigenvectors exhibit an interesting non-ergodic delocalized phase (see Section~\ref{sec:mobedge}), in the non-Hermitian case there is a sharp transition (on polynomial scales in $N$) from localization for $\gamma>0$ to a fully delocalized phase for $\gamma=0$. Surprisingly, in the recent work \cite{fyodorov2025kac} (see Section 2.3.1 therein), Fyodorov discovered that also in the non-Hermitian case non-ergodic delocalized states appear for $\lambda$ close to $1$ on logarithmic scales. More precisely, by computing the $\mathrm{IPR}_2$ (see \eqref{eq:IPR}) for either right or left eigenvectors and choosing $\lambda$ as $\lambda=[(\mu/2)\log N]^{-1/2}$, for an $N$-independent $\mu>0$, non-ergodic delocalized states indeed emerge in the non-Hermitian setting as well.

Similarly to the Hermitian case, by the methods developed in this work, we can study the non-Hermitian Rosenzweig-Porter model as well as its generalization when $X_0$ is not necessarily diagonal and not even normal. This can be achieved by considering the Hermitized matrix
\[
H_\lambda^z:=\left(\begin{matrix}
0 & X_\lambda-z \\
(X_\lambda-z)^* & 0
\end{matrix}\right)
\]
and studying its resolvent via an analog of the Zigzag strategy presented below\footnote{Studying this Hermitized matrix will also prove the localization/delocalization transition for the eigenvectors of the recently introduced Wishart-Rosenzweig-Porter model \cite{delapalme2025wishart}.}. This also gives a control on the singular vectors. Both the resolvent and the singular vectors are controlled uniformly in $z$ and with very high probability. We can thus use the idea of \cite[Proof of Corollary 2.4]{alt2021spectral} and \cite[Section 3.3.3]{cipolloni2024optimal} to conclude the same bounds for the actual left/right eigenvectors of $X_\lambda$. Moreover, we believe that the results \cite{nonHermdecay, cipolloni2024optimal, cipolloni2025non, cipolloni2023rightmost, cipolloni2025optimal, erdHos2024wegner}, obtained for $\lambda=1$ (and in some cases $X_0=0$), can be extended to the model \eqref{eq:nonhermRP}. In particular, by considering $X_0$ to be Hermitian and $\lambda\ll 1$, we can also study the eigenvalues/eigenvectors behavior in the case of the elliptic ensemble in the regime of weak non-Hermiticity (see \cite{byun2025progress} for a recent review discussing several results in this model). We do not present a detailed 
study as no new key inputs are needed and a similar strategy to the one used below applies.

\section{Two resolvent bound: Proof of Theorem \ref{theo:ETH}} \label{sec:proofmain}
The goal of this section is to give a proof of our ETH result in Theorem \ref{theo:ETH} based on a two-resolvent bound formulated in Theorem \ref{thm:main2G} below. We begin by defining the general regularization of observables, similarly to the special case introduced around \eqref{eq:regobsij}--\eqref{eq:s2ijdef} above. 
\begin{definition}[Regular observables and regularization]
	\label{def:regobs}
	Let $A \in \C^{N \times N}$ and $z_1, z_2 \in \C\setminus \R$ and denote $z_i^\pm := \Re z_i \pm \ii |\Im z_i|$. We say that $A$ is \emph{regular w.r.t.~$(z_1, z_2)$} (or, for short, \emph{$(z_1, z_2)$-regular}), if and only if $A$ is equal to its \emph{regularization} $\mathring{A}^{z_1, z_2}$ defined as 
	\begin{equation}
	\label{eq:regobs}
		\mathring{A}^{z_1, z_2} := A - \theta(z_1, z_2) \frac{\langle M(z_1^+) A M(z_2^-) \rangle }{\langle M(z_1^+)  M(z_2^-) \rangle} \mathbf{1} \,. 
	\end{equation}
	Here, with $\theta(x) = (1 + \sgn(x))/2 $ being the Heaviside function, we introduced the shorthand notation
	
\begin{equation}
	\label{eq:defcutoff}
\theta(z_1, z_2) := \theta\big( \lambda^2 \Re \langle M(z_1^+) M(z_2^-) \rangle - 1/2 \big) \,. 
	\end{equation}
\end{definition}
We point out that the definition of regularity in Definition~\ref{def:regobs} slightly differs from the one in \cite[Definition 2.2]{cipolloni2024eigenvector} (see also \cite[Definition 3.1]{cipolloni2024optimal}). In particular, the cut-off function in \cite[Eq. (2.9)]{cipolloni2024eigenvector} (see also \cite[Eq. (3.6)]{cipolloni2024optimal}) monitors the distance among the spectral parameters, rather than the size of $\langle M(z_1^+) M(z_2^-)\rangle$ as in \eqref{eq:defcutoff}. Here we made this different choice as it substantially simplifies the proof.
Similarly to the definition of regularity in \cite{cipolloni2024optimal}, we point out that $(z_1^{(*)}, z_2^{(*)})$-regularity is exactly the same for any constellation of adjoints taken. 
Moreover, abbreviating $\mathring{A} = \mathring{A}^{z_1, z_2}$, $M_i := M(z_i^+)$ and $\Gamma_i := \Im M_i /\langle \Im M_i \rangle $, we define
\begin{equation} \label{eq:defs2}
	\begin{split}
		\big(\s_2(z_1, z_2;A)\big)^2:=\langle\Gamma_1 \mathring{A} \Gamma_2 \mathring{A}^*\rangle &+ \theta(z_1, z_2) \, \lambda^2\bigg| \frac{\langle M(z_1^+) \mathring{A} \Gamma_2 \rangle \langle \Gamma_1 \mathring{A}^* M( z_2^+ ) \rangle }{1 - \lambda^2 \langle M(z_1^+) M( z_2^+)\rangle}\bigg|   \\[1mm]
		&+ (1 - \theta(z_1, z_2)) \, \lambda^4 \max_{*} |\langle M_1^{(*)} \mathring{A} M^{(*)}_{2} \rangle|^2 \langle \Gamma_1 \, \Gamma_2 \rangle ,
	\end{split}
\end{equation}
Then, our two resolvent bound takes the following form.

\begin{theorem}[Two resolvent bound] \label{thm:main2G} 
Fix $\epsilon > 0$.  Let $W$ be a Wigner matrix satisfying Assumption~\ref{ass:entries}, $H_0 \in \C^{N \times N}$ an arbitrary deterministic self adjoint matrix, and let $H_\lambda$ be as in \eqref{eq:smalldef}.  For $i=1,2$ take $z_i \in \C \setminus \R$ and denote $\eta_i := |\Im z_i|$, $\rho_i := \rho_\lambda(z_i)$ and $\ell:=\min_i(\eta_i\rho_i)$. 
	\begin{equation} \label{eq:2G2}
		\left| \big\langle \Im G(z_1) A \Im G(z_2) A^* 
		\big\rangle \right| \prec \rho_1 \rho_2\big( \s_2(z_1, z_2; A)\big)^2
	\end{equation}
uniformly in spectral parameters satisfying $N\ell\ge N^\epsilon$ and $\max_i |z_i| \le \lambda N^{\kappa_0 + 1/\epsilon}$, and deterministic $(z_1, z_2)$-regular matrices $A\in\C^{N\times N}$ (see Definition \ref{def:regobs}).
\end{theorem}

Here we presented only a bound for the product of two (imaginary parts of) resolvents and two regular observables, since this suffices to prove our ETH result \eqref{eq:ETH}. However, within its proof we will consider also different products of two resolvents (not necessarily only their imaginary parts) and we will prove a stronger local law, i.e.~bounds on the fluctuations as well (see Theorem~\ref{thm:2GLL} below). We also point out that while in this work we focus on local laws for averaged observables of products of two resolvents, our proof is robust allowing us to also prove local laws in an isotropic sense, for arbitrary long products of resolvents, and for arbitrary products containing resolvents and their imaginary parts. We omit this for the sake of brevity and clarity of the presentation.

We are now ready to give the proof of Theorem \ref{theo:ETH}; the proof of Theorem \ref{thm:main2G} is postponed to Section~\ref{subsec:proof2G}. 
\begin{proof}[Proof of Theorem~\ref{theo:ETH}]
Fix an arbitrarily small $\xi>0$, and recall the notation $\widehat{\lambda}_i$ from \eqref{eq:gammahat}. Assume that $\mathring{A}$ is $(\widehat{\lambda}_i,\widehat{\lambda}_j)$-regular, which is now equivalent to it being equal to its $(i,j)$-regularization in the sense of \eqref{eq:regobsij}, but for random spectral parameters. Similarly to the grid argument around \eqref{eq:gridargu} we can ensure that the bound \eqref{eq:2G2} holds with very high probability simultaneously for all spectral parameters described below \eqref{eq:2G2}. Then, abbreviating $\rho_i = \pi^{-1} \langle \Im M(\widehat{\lambda}_i) \rangle$, by spectral decomposition and using the bound \eqref{eq:2G2},
together with the definition of $\eta_{\mathfrak{f}}$ below \eqref{eq:evdeloc}, we have 
\begin{equation*}
N\big|\langle {\bm u}_i, \mathring{A} {\bm u}_j\rangle\big|^2\lesssim (N^2 \eta_i \eta_j)\big\langle \Im G(\widehat{\lambda}_i) \mathring{A} \Im G(\widehat{\lambda}_j) \mathring{A}^* \big\rangle \lesssim N^{2 \xi } \frac{\big\langle \Im G(\widehat{\lambda}_i) \mathring{A} \Im G(\widehat{\lambda}_j) \mathring{A}^* \big\rangle }{\rho_i \rho_j }.
\end{equation*}
Additionally, by \eqref{eq:2G2} and abbreviating $\s_2 = \s_2(i,j;A) = \s_2(\widehat{\lambda}_i, \widehat{\lambda}_j; A)$,  we have
\[
\big\langle \Im G(\widehat{\lambda}_i) \mathring{A} \Im G(\widehat{\lambda}_j) \mathring{A}^* \big\rangle\lesssim \rho_i\rho_j (\s_2)^2.
\]
Combining these two bounds, and using that $M_i M_i^* /\langle M_i M_i^* \rangle = \Gamma_i$, we obtain
\begin{equation*}
\big|\langle {\bm u}_i, A {\bm u}_j\rangle - \delta_{ij} \theta_{ij} \langle \Gamma_i A \rangle\big|^2 = \big|\langle {\bm u}_i, \mathring{A} {\bm u}_j\rangle\big|^2\prec \frac{1}{N}(\s_2)^2. \qedhere
\end{equation*}
\end{proof}

\section{Local law proofs: Zigzag strategy} \label{sec:proofLL}

To prove the local laws in Theorems~\ref{theo:main1} and the bound in Theorem~\ref{thm:main2G} we follow the Zigzag strategy \cite{cipolloni2023eigenstate, cipolloni2024out, cipolloni2024eigenvector}, consisting of three steps:
\begin{itemize}

\item[\textbf{1.}] \textbf{Global law.} Proof of a local law for spectral parameters away from the spectrum of $H_\lambda$. 
\item[\textbf{2.}] \textbf{Characteristic flow.} Propagate the global law to a local law by considering the evolution of the Wigner component of $H_\lambda$ along the Ornstein-Uhlenbeck flow
\begin{equation}
	\label{eq:OU}
	\dd W_t = - \frac{W_t}{2} \dd t + \frac{\dd B_t}{\sqrt{N}} \quad \text{with initial condition} \quad W_0 = W,
\end{equation}
and the expectation matrix evolving according to $H_{0,t}:=e^{-t/2}H_0$. 
The evolution of spectral parameters of the resolvents from the global regime to the local regime is governed by the \emph{characteristic equation} (see, e.g., \cite[Eq. (5.3)]{cipolloni2024mesoscopic}):
\begin{equation}
	\label{eq:char}
	\partial_t z_t=-\frac{1}{2}z_t-\lambda^2\langle M_{\lambda,t}(z_t)\rangle,
\end{equation}
where $M_{\lambda,t}(z)$ is the solution of \eqref{eq:MDE} with $H_0$ being replaced by $H_{0,t}$.

\item[\textbf{3.}] \textbf{Green function comparison.} Remove the Gaussian component by a \emph{Green's function compari-
son (GFT)} argument.
\end{itemize}
To keep the presentation simpler, we shall henceforth assume $\sigma := \E \chi_{\mathrm{od}}^2 $ is real and and $\E \chi_{\mathrm{d}}^2 = 1 + \sigma$. The modifications to the general case are routine and hence left to the reader (see, e.g., \cite[Section~4.4]{cipolloni2023eigenstate} for some details concerning the second step). Now, in Section \ref{subsec:setup}, we start by setting up an  overarching inductive proof scheme. Afterwards, in Sections~\ref{subsec:proof1G}--\ref{subsec:proof2G} we provide the proofs of the single resolvent and two resolvent laws, respectively. 
\subsection{Setting up the inductive zigzag procedure} \label{subsec:setup} In this section, we setup the overarching framework of our inductive zigzag proof, analogously to \cite[Section 3]{erdHos2025cusp}. 
\subsubsection{Two random matrix flows}
Along the proof, we use two similar (but distinct) flows in the space of $N \times N$ matrices, the \emph{zig-flow} (a standard Ornstein-Uhlenbeck process, naturally scaled by $\lambda$), defined as
\begin{equation} \label{eq:zigflow}
\dd H_{\lambda, t} = - \frac{1}{2} H_{\lambda,t} \dd t + \lambda \frac{\dd \BM_t}{\sqrt{N}}\,, \quad t \ge 0
\end{equation}
and the \emph{zag-flow} (a modified Ornstein-Uhlenbeck process), distinguished by the superscript $t$, 
\begin{equation} \label{eq:zagflow}
\dd H_\lambda^t = - \frac{1}{2} \big(H_\lambda^t - \E H_\lambda^t  \big) \dd t +  \lambda \frac{\dd \BM_t}{\sqrt{N}}\,, \quad t \ge 0
\end{equation}
where $\BM_t$ denotes the real symmetric or complex Hermitian standard matrix valued Brownian motion, depending on the symmetry class of $H_\lambda$ (i.e.~whether $\chi_{\mathrm{od}}$ is a real or complex random variable). On the one hand, by taking the expectation in \eqref{eq:zigflow}, we see that the deterministic $H_{0,t}$ evolves as $H_{0,t} = \ee^{-t/2} H_{0,0}$. On the other hand, the zag flow \eqref{eq:zagflow} keeps the first two moments of $H_\lambda^t$ preserved. Therefore, in particular, the deterministic approximation $M$ is unchanged during the zag-flow. 

For any $t \ge 0$, we define the flow maps $\mathfrak{F}_{\rm zig}^t$ and $\mathfrak{F}_{\rm zag}^t$ on the space of probability distributions on $\C^{N \times N}$ as
\begin{alignat}{2}
\mathfrak{F}_{\rm zig}^t \big[H_\lambda\big] &:= H_{\lambda, t}\,, \qquad &&\text{where } H_{\lambda, t} \text{ solves \eqref{eq:zigflow} with initial condition } H_{\lambda, 0} = H_\lambda \\ 
\mathfrak{F}_{\rm zag}^t \big[H_\lambda\big] &:= H_{\lambda}^t\,, \qquad &&\text{where } H_{\lambda}^t \text{ solves \eqref{eq:zagflow} with initial condition } H_{\lambda}^0 = H_\lambda \,. 
\end{alignat}
The obvious (but key) relation between the two random matrix flows is captured in the following lemma. 
\begin{lemma}[Relation between the zig and zag flow] \label{lem:zigzagrel}
Fix $\lambda > 0$ and consider $H_\lambda = H_0 + \lambda W$ with a deterministic Hermitian $H_0 \in \C^{N \times N}$ and $W$ being a Wigner matrix satisfying Assumption \ref{ass:entries}. Then it holds that
\begin{equation*}
\mathfrak{F}_{\rm zig}^t\big[\ee^{t/2} H_0 + \lambda W\big] \stackrel{\dd}{=} \mathfrak{F}_{\rm zag}^t\big[H_\lambda\big] \qquad \text{for all} \quad t \ge 0 \,. 
\end{equation*}
\end{lemma}

\subsubsection{Time dependent spectral domains and characteristic flow}
The induction base of our proof is laid by a \emph{global law} for the resolvent $G_\lambda(z) := (H_\lambda - z)^{-1}$ of a random matrix $H_\lambda = H_0 + \lambda W$. That is, we show that for spectral parameters $z$ "far away" from the spectrum of $H_\lambda$, the resolvent concentrates around the unique solution $M_\lambda(z)$ of the MDE \eqref{eq:MDE}. More precisely, for any $\epsilon, c > 0$, the global law (see Proposition \ref{prop:global}) holds on the domain (recall $\kappa_0$ from \eqref{eq:kappa})
	\begin{equation} \label{eq:Dglobdef}
		\mathcal{D}^{\rm glob} \equiv \mathcal{D}^{\rm glob}(\epsilon, c):= \left\{ z \in \C \setminus \R  : |\Re z| \vee |\Im z| \le \lambda N^{\kappa_0 + 1/\epsilon }, \, \lambda^2 \rho_\lambda(z)/|\Im z| \le 1/c  \right\}
	\end{equation}
	where we recall the notation $\rho_\lambda(z) := \pi^{-1} |\langle \Im M_\lambda(z) \rangle|$. The global domain is characterized by the \emph{stability operator} 
	\begin{equation} \label{eq:stabop}
		\mathcal{B}[\cdot] \equiv \mathcal{B}_\lambda(z)[\cdot] := \mathbf{1} - \lambda^2 M_\lambda(z) \langle \cdot \rangle M_\lambda(z)
	\end{equation}
	having a bounded inverse (bound by $1/c$).  In most previous zigzag proofs of local laws of a random matrix $H$, the
\emph{global} regime meant $\eta:=|\Im z| \gtrsim 1$ (or $\mathrm{dist}(z, \mathrm{supp}\rho)\gtrsim 1$, whenever the edge regime is included), 
under the overall condition that $\lVert H \rVert$ is typically of order 1. In our situation this would correspond
to $\eta\ge \lVert H_\lambda\rVert \sim \lVert H_0\rVert + \lambda$. We stress that the condition in \eqref{eq:Dglobdef}, $\eta\gtrsim \lambda^2\rho_\lambda$ is different. In particular, $\eta$ could be much smaller than $ \lVert H_0\rVert + \lambda$. More importantly, 
the new condition \eqref{eq:Dglobdef} does not depend on $\lVert H_0\rVert$, which could be very large for its far away 
eigenvalues that do not influence the zigzag argument around $z$. In short, \eqref{eq:Dglobdef} based upon the
boundedness of the stability operator is the canonically correct concept of the global regime. 
The goal of the zigzag proof is to propagate the global law in $\mathcal{D}^{\rm glob}$ to a \emph{local law}, valid in the \emph{above the scale} regime, where $N |\Im z| \rho_\lambda(z)$ is large. For $\epsilon > 0$ we define it as
	\begin{equation} \label{eq:Dabvdef}
	\mathcal{D}^{\rm abv}(\epsilon) := \left\{ z \in \C \setminus \R  : |\Re z| \vee  |\Im z|\le \lambda N^{\kappa_0 + 1/\epsilon}, \, N |\Im z| \rho_\lambda(z) \ge N^\epsilon  \right\}
\end{equation}
and note that studying local laws in this regime is natural, since $N |\Im z| \rho_\lambda(z)$ is the typical number of eigenvalues in an interval of length $|\Im z|$ around the energy $\Re z$. 
	
To achieve the propagation from global to local laws, we consider the time-dependent Matrix Dyson Equation (MDE)
\begin{equation} \label{eq:timedepMDE}
- \frac{1}{M_{\lambda, t}(z)} = z - H_{0,t} + \lambda^2 \langle M_{\lambda, t}(z) \rangle \quad z \in \C\setminus \R \quad \text{under the constraint} \quad \Im z \cdot \Im M_{\lambda, t}(z) > 0
\end{equation}
where $H_{0,t}$ solves the ordinary differential equation
\begin{equation}
\dd H_{0,t} = - \frac{1}{2} H_{0,t} \dd t \quad \text{for} \quad t \ge 0 \quad \text{with terminal condition} \quad H_{0,T} = H_0 \quad \text{for some} \quad T > 0\,. 
\end{equation}
In the following we will fix the terminal time $T$ to be given by $$T = 1\,.$$ That is, the total running time of our flows will be given by the interval of times $t \in [0,1]$. Given $M_t$, we consider the \emph{characteristic flow} for the time dependent spectral parameters $z_t$ evolving according to
\begin{equation} \label{eq:characteristic}
\dd z_t = - \frac{z_t}{2} \dd t - \lambda^2 \langle M_{\lambda, t}(z_t) \rangle \dd t \,. 
\end{equation}
Moreover, we also define the (inverse) flow map $\varphi_{s,t} : \C\setminus \R \to \C\setminus \R$, uniquely defined as
\begin{equation} \label{eq:flowmap}
\varphi_{s,t}(z_t) := z_s \quad \text{with $z_s$ solving \eqref{eq:characteristic}} \,. 
\end{equation}
One can immediately check that along the characteristic curves the solution to the MDE satisfies
\begin{equation} \label{eq:Mevolve}
\dd M_{\lambda, t}(z_t) =  \frac{1}{2} M_{\lambda, t}(z_t) \dd t \,. 
\end{equation}
In the following, we denote $\rho_{\lambda, t}(z) := \pi^{-1} |\langle \Im M_{\lambda, t}(z) \rangle|$ for any $z \in \C\setminus \R$ and $t \in [0,1]$. 
For constants $c', C' > 0$, we define the time dependent version of the \emph{above the scale} domain as
	\begin{equation*}
		\begin{split}
\mathcal{D}_t^{\rm abv} = \mathcal{D}_t^{\rm abv}(\epsilon, c', C') := \Big\{ z \in \C  \setminus \R : \ &|\Re z| \vee |\Im z|\le \lambda N^{\kappa_0 + 1/\epsilon} \cdot \big(1+ C'(1 - t)\big),  \, N |\Im z| \rho_{\lambda, t}(z) \ge N^\epsilon , \\
&\big(\lambda^2 \rho_{\lambda, t}(z)/|\Im z|\big)^{-1} \ge c' \cdot \big( N^{-1+\epsilon} + (1 - t)\big) \Big\}
		\end{split}
	\end{equation*}
Then we have the following key lemma on the propagation from the global to the local scale. The proof is completely analogous to that of \cite[Lemma 3.5]{erdHos2025cusp} and hence omitted. 
\begin{lemma}[Relations between time-dependent domains] \label{lem:Dtrel}
Fix a small $\epsilon > 0$. Then there exist constants $c', C' \sim 1$ such that the time dependent domains $\big(\mathcal{D}_t^{\rm abv}(\epsilon, c', C')\big)_{t \in [0,1]}$ satisfy 
\begin{equation} \label{eq:domaininterpolation}
\varphi_{s,t}(\mathcal{D}_t^{\rm abv}(\epsilon, c', C')) \subset \mathcal{D}_s^{\rm abv}(\epsilon, c', C')
\end{equation}
 for all $0 \le s \le t \le 1$. Moreover, there exists a constant $c \sim 1$ such that, when replacing $\rho_\lambda \to \rho_{\lambda,0}$ in the definition of $\mathcal{D}^{\rm glob}(\epsilon, c)$ in \eqref{eq:Dglobdef} and $\rho_\lambda \to \rho_{\lambda, 1}$ in the definition of $\mathcal{D}^{\rm abv}(\epsilon)$ in \eqref{eq:Dabvdef}, it holds that 
\begin{equation} \label{eq:initialandfinaldomains}
\mathcal{D}_0^{\rm abv}(\epsilon, c', C') \subset \mathcal{D}^{\rm glob}(\epsilon, c) \quad \text{and} \quad \mathcal{D}^{\rm abv}(\epsilon) \subset \mathcal{D}_1^{\rm abv}(\epsilon, c', C') \,. 
\end{equation}
\end{lemma}
The parameters $\epsilon, c', C', c$ of the domains are henceforth considered fixed and shall often be omitted. 

\subsubsection{Preparing the inductive proof: Two sequences of random matrices}
Our zigzag proof will be conducted inductively. in order to do so, we fix an arbitrarily small step size exponent $\delta > 0$ and let $K = K(\delta, \epsilon)$ be the smallest integer such that $N^{-K \delta} \le N^{-1+\epsilon}$, where $\epsilon > 0$ is the \emph{above-the-scale exponent} from the definition of $\mathcal{D}^{\rm abv}$ in \eqref{eq:Dabvdef}. We now define a sequence of times $\{t_k\}_{k=0}^K$ as
\begin{equation} \label{eq:tkdef}
t_0 := 0\,, \quad t_k := 1 - N^{-k \delta} \,, \quad k \in \{1, ... , K-1\}\,, \quad t_K := 1 \,. 
\end{equation}
Let $\{\dift_k\}_{k=1}^K$ denote the consecutive difference sequence of the $t_k$'s, that is
\begin{equation} \label{eq:dtkdef}
\dift_k := t_k - t_{k-1} \quad \text{for} \quad k \in \{1, ... , K\} \,. 
\end{equation}
Given $H_\lambda$ as the target random matrix ensemble, we construct two sequences of random matrices $\{H_k\}_{k=0}^K$ and $\{H^k\}_{k=1}^K$ recursively by\footnote{The random matrix $H_{0}$ defined in \eqref{eq:twosequences} should not be confused with the deterministic matrix $H_0$. We will clarify possible ambiguities whenever needed by writing $H_{k=0}$ for the random matrix. }
\begin{equation} \label{eq:twosequences}
H_K := H_\lambda \,, \quad H^k := \mathfrak{F}_{\rm zag}^{\dift_k}\big[H_k\big]\,, \quad H_{k-1} := \big(\ee^{\dift_k/2} - 1\big) \E H_k + H_k \,. 
\end{equation}

\subsection{Proof of the single resolvent law} \label{subsec:proof1G}
We initiate the zigzag induction by the following global law. 
\begin{proposition}[Global law]\label{prop:global}
Fix $\lambda>0$ and consider $H_\lambda = H_0 + \lambda W$ as in \eqref{eq:smalldef} with $W$ being a Wigner matrix satisfying Assumption~\ref{ass:entries}. Then, denoting $\rho_\lambda(z) := \pi^{-1} |\langle \Im M_\lambda(z) \rangle|$ and $\eta := |\Im z|$, recalling the notation $\Gamma_\lambda(z) = \Im M_\lambda(z) /\langle \Im M_\lambda(z) \rangle$ from \eqref{eq:Gammadef},   we have
\begin{subequations} \label{eq:globallaw}
	\begin{equation} \label{eq:avGL}
		\left| \left\langle \big(G_\lambda (z) - M_\lambda(z)\big) B \right\rangle \right|  \prec \frac{\Vert B \Vert_{\rm hs}}{N \eta} \, \sqrt{\frac{\langle \Gamma_\lambda(z) B B^* \rangle}{\langle BB^* \rangle}}
	\end{equation}
	and
	\begin{equation} \label{eq:isoGL}
		\left| \left\langle \bm x, \big(G_\lambda (z) - M_\lambda(z)\big) \bm y \right\rangle \right|  \prec  \left(\sqrt{\frac{\rho_\lambda(z)}{N \eta}} \Vert \bm x \Vert \, \Vert \bm y \Vert\right) \ \sqrt{\frac{\big(\Gamma_\lambda(z)\big)_{\bm x \bm x} \, \big(\Gamma_\lambda(z)\big)_{\bm y \bm y}}{\Vert \bm x \Vert^2 \, \Vert \bm y \Vert^2 }} \,. 
	\end{equation}
\end{subequations}
uniformly in spectral parameters $z \in \mathcal{D}^{\rm glob}$, deterministic matrices $B \in \C^{N \times N}$ and deterministic vectors $\bm x, \bm y \in \C^N$. 
\end{proposition}

The zigzag induction is then conducted along the following two propositions. 
\begin{proposition}[Zig step]\label{prop:zigstep}
	Fix $k \in \{1, ... , K\}$ and denote 
	\begin{equation*}
G_{\lambda, t}(z) := \big(\mathfrak{F}_{\rm zig}^{t-t_{k-1}}\big[H_{k-1}\big] - z\big)^{-1}\,, \quad t_{k-1} \le t \le t_k \,. 
	\end{equation*}
	Assume that $G_{\lambda, t}$ satisfies the single resolvent local laws \eqref{eq:globallaw} uniformly in $z \in \mathcal{D}^{\rm abv}_t$, $B \in \C^{N \times N}$, and $\bm x, \bm y \in \C^N$ at time $t= t_{k-1}$. Then $G_{\lambda, t}$ satisfies the single resolvent local laws \eqref{eq:globallaw} uniformly in $z \in \mathcal{D}^{\rm abv}_t$, $B \in \C^{N \times N}$, and $\bm x, \bm y \in \C^N$ uniformly in all times $t\in [ t_{k-1}, t_k]$. 
\end{proposition}

\begin{proposition}[Zag step]\label{prop:zagstep}
	Fix $k \in \{1, ... , K\}$ and denote $s_k := \dift_k$. Let $H_k$ be the random matrix defined in \eqref{eq:twosequences}, and denote 
\begin{equation*}
	G_{\lambda}^s(z) := \big(\mathfrak{F}_{\rm zag}^{s}\big[H_{k}\big] - z\big)^{-1}\,, \quad 0 \le s \le s_k \,. 
\end{equation*}
Assume that $G_{\lambda}^s$ satisfies the single resolvent local laws \eqref{eq:globallaw} uniformly in $z \in \mathcal{D}^{\rm abv}_{t_k}$, $B \in \C^{N \times N}$, and $\bm x, \bm y \in \C^N$ at time $s= s_k$. Then $G_{\lambda}^s$ satisfies the single resolvent local laws \eqref{eq:globallaw} uniformly in $z \in \mathcal{D}^{\rm abv}_{t_k}$, $B \in \C^{N \times N}$, and $\bm x, \bm y \in \C^N$ uniformly in all times $s\in [0, s_k]$. 
\end{proposition}

We now put together the above cardinal zig and zag steps in the proof of Theorem \ref{theo:main1}. 

\begin{proof}[Proof of Theorem \ref{theo:main1}]
By means of Lemmas \ref{lem:zigzagrel} and \ref{lem:Dtrel}, the proof Theorem \ref{theo:main1} follows by a simple induction involving Propositions \ref{prop:global}, \ref{prop:zigstep}, and \ref{prop:zagstep}. 
\end{proof}

\subsection{Proof of the two resolvent bound} \label{subsec:proof2G}
The goal of this section is to give the proof of the two resolvent bound in Theorem \ref{thm:main2G}. Throughout the entire argument we will fix two spectral parameters $z_1, z_2 \in \mathcal{D}^{\rm abv}(\epsilon)$ and focus on the more delicate regime with 
\begin{equation*}
\theta(z_1, z_2) = 1 \,, 
\end{equation*}
where the regularization $\mathring{A}^{z_1, z_2}$, defined in \eqref{eq:regobs}, of an observable $A$ is non-trivial. The modifications needed for the complementary case $\theta(z_1, z_2) = 0$ are outlined in Appendix \ref{subsec:regime2}. 

In the following, for any set $\mathcal{D} \subset \C \setminus \R$, we denote the subset of such $(z_1, z_2) \in \mathcal{D}^2$ satisfying $\theta(z_1, z_2) = 1$ by
\begin{equation} \label{eq:theta=1domain}
(\mathcal{D})^{\times 2}_{\theta = 1} := \left\{(z_1, z_2) \in \mathcal{D}^{ 2} \, : \, \theta(z_1, z_2) = 1 \right\}
\end{equation}
\subsubsection{Two resolvent local laws: Proof of Theorem \ref{thm:main2G}}
The resolvent bound in Theorem~\ref{thm:main2G} will follow from a two-resolvent local law, i.e.~a concentration estimate for $\langle \Im G(z_1) A \Im G(z_2) A^* \rangle$, naturally accompanied by local laws for $\langle G(z_1^{(*)}) A G(z_2^{(*)}) \rangle$ for a $(z_1, z_2)$-regular observable $A$ (see Theorem~\ref{thm:2GLL} below). In order to formulate the local laws, we introduce their deterministic approximation. For $z_1, z_2 \in \C \setminus \R$ and $A \in \C^{N \times N}$ let
\begin{equation}
	\label{eq:m12}
	M_\lambda(z_1,A,z_2):=M_\lambda(z_1)AM_\lambda(z_2) + \lambda^2 \frac{\langle M_\lambda(z_1)AM_\lambda(z_2) \rangle }{1-\lambda^2\langle M_\lambda(z_1)M_\lambda(z_2)\rangle} M_\lambda(z_1)M_\lambda(z_2)\,, 
\end{equation}
where $M_\lambda$ is the unique solution to \eqref{eq:MDE}. 
Then we have the following bounds on the deterministic approximations, whose proofs are given in Appendix~\ref{subsec:Mbounds}. 
\begin{lemma}
	\label{lem:m12bounds}
Fix $(z_1, z_2) \in (\C\setminus\R)_{\theta = 1}^{\times 2}$  and denote $\eta_i := |\Im z_i|$ and $\rho_i := \pi^{-1} |\langle \Im M(z_i) \rangle|$ for $i =1,2$. Let $A = \mathring{A}^{z_1, z_2} \in \C^{N\times N}$ be $(z_1,z_2)$-regular as in Definition~\ref{def:regobs} (recall also the $z_i^\pm$ notation from there). Then it holds that, for $\s_2 = \s_2(z_1, z_2;A)$, 
	\begin{equation}
		\label{eq:Mbounds}
		\begin{split}
			\big|\langle M_\lambda(z_1^+,A,{z_2^+})\rangle\big|&\lesssim \sqrt{\frac{\rho_2}{\eta_1}}\sqrt{\rho_1\rho_2} \mathfrak{s}_2, \\
			\big|\langle M_\lambda(z_1^-,A,z_2^-)\rangle\big|&\lesssim \sqrt{\frac{\rho_1}{\eta_2}}\sqrt{\rho_1\rho_2} \mathfrak{s}_2, \\
			\big|\langle M_\lambda({z_1^-},A,{z_2^+})\rangle\big|&\lesssim \sqrt{\frac{\rho_1}{\eta_2}\vee\frac{\rho_2}{\eta_1}}\sqrt{\rho_1\rho_2} \mathfrak{s}_2, \\
		\max_* 	\big|\langle \widehat{M}_\lambda(z_1^{(*)},A,z_2^{(*)})\rangle\big|&\lesssim \sqrt{\frac{\rho_1}{\eta_2}\wedge\frac{\rho_2}{\eta_1}}\sqrt{\rho_1\rho_2} \mathfrak{s}_2, \\
		\max_*	\big|\langle \widehat{M}_\lambda(z_1^{(*)},A,z_2^{(*)})A^*\rangle\big|&\lesssim {\rho_1\rho_2} \big(\mathfrak{s}_2\big)^2, 
		\end{split}
	\end{equation}
	where $\max_*$ is understood to be the maximum over all possible constellations of adjoints taken, and we introduced the shorthand notation 
\begin{equation} \label{eq:hatMfull}
	\widehat{M}_\lambda(z_1,A,z_2):=-\frac{M_\lambda(z_1,A,z_2)+M_\lambda(\overline{z_1},A,\overline{z_2})-M_\lambda(z_1,A,\overline{z_2})-M_\lambda(\overline{z_1},A,z_2)}{4}.
\end{equation}
\end{lemma}
We have the following two-resolvent local laws, whose proofs are given in Section \ref{subsubsec:zigzagproof} below. 
\begin{theorem}[Two resolvent local laws] \label{thm:2GLL}
	Fix $\epsilon > 0$. Let $W$ be a Wigner matrix satisfying Assumption~\ref{ass:entries}, $H_0$ an arbitrary deterministic self adjoint matrix, and let $H_\lambda$ be as in \eqref{eq:smalldef}. Take $(z_1, z_2) \in \big(\mathcal{D}^{\rm abv}(\epsilon)\big)^{\times 2}_{\theta = 1}$ and denote $G_i:=(H_\lambda-z_i)^{-1}$, $M_i:=M_\lambda(z_i)$, $\rho_i:=\pi^{-1}|\langle \Im M_i\rangle|$, $\eta_i:=|\Im z_i|$, and $\ell := \min_i (\eta_i \rho_i)$. Then it holds that 
		\begin{equation} \label{eq:2G1A}
		\left| \big\langle G_1^{(*)} A G_2^{(*)} \big\rangle - \big\langle M_\lambda\big(z_1^{(*)},  A ,z_2^{(*)}\big) \big\rangle \right| \prec \sqrt{\frac{\rho_1 \rho_2}{N \eta_1 \eta_2}} \s_2(z_1, z_2; A)
	\end{equation}
	and 
	\begin{equation} \label{eq:2G2A}
		\left| \big\langle \Im G_1 A \Im G_2 A^* \big\rangle - \big\langle \widehat{M}_\lambda(z_1,  A ,z_2)A^* \big\rangle \right| \prec \frac{\rho_1 \rho_2}{\sqrt{N \ell}}\big( \s_2(z_1, z_2; A)\big)^2
	\end{equation}
	uniformly in $(z_1, z_2)$-regular matrices $A = \mathring{A}^{z_1, z_2} \in \C^{N \times N}$. 
\end{theorem}
We can now give the proof of Theorem \ref{thm:main2G}. 
\begin{proof}[Proof of Theorem \ref{thm:main2G}]
For spectral parameters $z_1, z_2$ satisfying $\theta(z_1, z_2) = 1$, Theorem~\ref{thm:main2G} immediately follows from \eqref{eq:2G2A} in Theorem~\ref{thm:2GLL}, the last $M$-bound in \eqref{eq:Mbounds}, and using that $N \ell \ge 1$. 

As mentioned above, for spectral parameters $z_1, z_2$ satisfying $\theta(z_1, z_2) = 0$ the modifications needed to the proof in the complementary regime $\theta(z_1, z_2) = 1$ will be discussed in Appendix \ref{subsec:regime2}. 
\end{proof}

In the remainder of this section, we will drop the subscript $\lambda$. 

\subsubsection{Zigzag induction for two resolvents: Proof of Theorem \ref{thm:2GLL}} \label{subsubsec:zigzagproof}
Throughout our zigzag induction, we fix a target random matrix $H_K = H_\lambda = H_0 + \lambda W$ (cf.~\eqref{eq:twosequences}) and target spectral parameters $z_1, z_2$ with $(z_1, z_2) \in \big(\mathcal{D}^{\rm abv}(\epsilon)\big)^{\times 2}_{\theta = 1}$. Moreover, we consider a fixed deterministic matrix $A = \mathring{A} \in \C^{N \times N}$, which we assume to be $(z_1, z_2)$-regular (cf.~Definition \ref{def:regobs}; the $M$ being the deterministic approximation to the resolvent of $H_K$).  Recall that in Section \ref{subsec:setup} we have fixed the terminal time to be given by $T=1$. So, having fixed $z_1, z_2$ and $A$, for $t \in [0,1]$, we define a time-dependent version of the basic control parameter $\s_2$ as 
\begin{equation} \label{eq:s2time}
	\big(\s_2(t)\big)^2 = 	\big(\s_2(z_1, z_2;A; t)\big)^2:=\langle\Gamma_1 \mathring{A} \Gamma_2 \mathring{A}^*\rangle + \lambda^2\bigg| \frac{\langle M(z_1^+) \mathring{A} \Gamma_2 \rangle \langle \Gamma_1 \mathring{A}^* M( z_2^+ ) \rangle }{1 - \ee^{t-1}\lambda^2 \langle M(z_1^+) M( z_2^+)\rangle}\bigg| 
\end{equation}
which satisfies the condition $\s_2(t=1) = \s_2(z_1, z_2;A)$ with the time-independent $\s_2(z_1, z_2;A)$ defined in \eqref{eq:defs2}. Finally, for $i \in [2]$, we denote by $z_{i,t}$ the solution to \eqref{eq:characteristic} with terminal condition $z_{i,1} = z_i$. Correspondingly, we will use the shorthand notations $M_{i,t}:= M_\lambda(z_{i,t})$, $M_{1^{(*)}2^{(*)},t}^A:=M_{\lambda,t}(z_{1,t}^{(*)},A,z_{2,t}^{(*)})$, $\widehat{M}_{12,t}^A:=\widehat{M}_{\lambda,t}(z_{1,t}, z_{2,t}; A)$, $\rho_{i,t}:=\pi^{-1}|\langle \Im M_{i,t}\rangle|$, $\eta_{i,t}:=|\Im z_{i,t}|$, and $\ell_t := \min_i (\eta_{i,t} \rho_{i,t})$. Note that, by definition of the regularization, if $A$ is $(z_1, z_2)$-regular for the target spectral parameters, then it is also regular w.r.t.~$(z_{1,t}, z_{2,t})$ because $M_{i,t} = \ee^{t/2} M_{i,0}$ (recall \eqref{eq:Mevolve}). 

We initiate the zigzag induction by the following global laws. 
\begin{proposition}[Global law for two resolvents] \label{prop:global2}
Using the above notations and conventions, we have that the resolvents
\begin{equation*}
G_{i,0} := \big(H_{k=0} - z_{i,0}\big)^{-1} \quad i \in [2]\,,
\end{equation*}
satisfy the local laws
	\begin{equation} \label{eq:2G1Aglobal}
		\left| \big\langle G_{1,0}^{(*)} A G_{2,0}^{(*)} \big\rangle - \big\langle {M}_{1^{(*)}2^{(*)},0}^A \big\rangle \right| \prec \sqrt{\frac{\rho_{1,0} \rho_{2,0}}{N \eta_{1,0} \eta_{2,0}}} \s_2(0)
	\end{equation}
	and 
	\begin{equation} \label{eq:2G2Aglobal}
		\left| \big\langle \Im G_{1,0} A \Im G_{2,0} A^* \big\rangle - \big\langle \widehat{M}_{12,0}^A \big\rangle \right| \prec \frac{\rho_{1,0} \rho_{2,0}}{\sqrt{N \ell_0}}\big( \s_2(0)\big)^2 
	\end{equation}
		uniformly in $(z_1, z_2)$-regular matrices $A = \mathring{A}^{z_1, z_2} \in \C^{N \times N}$. 
\end{proposition}

The zigzag induction is then conducted along the following two propositions.   We point out that the bound in \eqref{eq:2G1Aglobal} is not optimal. In fact, the optimal bound would have an additional small factor $1/\sqrt{N \ell}$. The method presented here could also achieve this optimal bound, however, we do not pursue this direction to keep the presentation shorter.
\begin{proposition}[Zig step for two resolvents]
\label{prop:zigstep2}
	Fix $k \in \{1, ... , K\}$ and denote 
\begin{equation*}
	G_{i,t} := \big(\mathfrak{F}_{\rm zig}^{t-t_{k-1}}\big[H_{k-1}\big] - z_{i,t}\big)^{-1}\,, \quad t_{k-1} \le t \le t_k \,, \quad i \in [2]\,. 
\end{equation*}
Assume that the $G_{i, t}$'s satisfy the two resolvent local laws 
	\begin{equation} \label{eq:2G1Azig}
	\left| \big\langle G_{1,t}^{(*)} A G_{2,t}^{(*)} \big\rangle - \big\langle {M}_{1^{(*)}2^{(*)},t}^A \big\rangle \right| \prec 	\sqrt{\frac{\rho_{1,t} \rho_{2,t}}{N \eta_{1,t} \eta_{2,t}}} \s_2(t)
\end{equation}
and 
\begin{equation} \label{eq:2G2Azig}
	\left| \big\langle \Im G_{1,t} A \Im G_{2,t} A^* \big\rangle - \big\langle \widehat{M}_{12,t}^A \big\rangle \right| \prec \frac{\rho_{1,t} \rho_{2,t}}{\sqrt{N \ell_t}}\big( \s_2(t)\big)^2 
\end{equation}
	uniformly in $(z_1, z_2)$-regular matrices $A = \mathring{A}^{z_1, z_2} \in \C^{N \times N}$ at time $t= t_{k-1}$. Then the $G_{i, t}$'s satisfy the two resolvent local laws \eqref{eq:2G1Azig}--\eqref{eq:2G2Azig} 	uniformly in $(z_1, z_2)$-regular matrices $A = \mathring{A}^{z_1, z_2} \in \C^{N \times N}$ uniformly in all times $t\in [ t_{k-1}, t_k]$. 
\end{proposition}

\begin{proposition}[Zag step for two resolvents] \label{prop:zagstep2}
	Fix $k \in \{1, ... , K\}$ and denote $s_k := \dift_k$. Let $H_k$ be the random matrix defined in \eqref{eq:twosequences}, and denote 
\begin{equation*}
	G_{i}^s := \big(\mathfrak{F}_{\rm zag}^{s}\big[H_{k}\big] - z_{i,t_k}\big)^{-1}\,, \quad 0 \le s \le s_k \,, \quad i \in [2]\,. 
\end{equation*}
Assume that the $G_{i}^s$'s satisfy the two resolvent local laws 
	\begin{equation} \label{eq:2G1Azag}
	\left| \big\langle (G_{1}^s)^{(*)} A (G_{2}^s)^{(*)} \big\rangle - \big\langle {M}_{1^{(*)}2^{(*)},t_k}^A \big\rangle \right| \prec \sqrt{\frac{\rho_{1,t_k} \rho_{2,t_k}}{N \eta_{1,t_k} \eta_{2,t_k}}} \s_2(t_k)
\end{equation}
and 
\begin{equation} \label{eq:2G2Azag}
	\left| \big\langle \Im G_{1}^s A \Im G_{2}^s A^* \big\rangle - \big\langle \widehat{M}_{12,t_k}^A \big\rangle \right| \prec \frac{\rho_{1,t_k} \rho_{2,t_k}}{\sqrt{N \ell_{t_k}}}\big( \s_2(t_k)\big)^2 
\end{equation}
	uniformly in $(z_1, z_2)$-regular matrices $A = \mathring{A}^{z_1, z_2} \in \C^{N \times N}$ at time $s= s_k$. Then the $G_{i}^s$'s satisfy the two resolvent local laws \eqref{eq:2G1Azag}--\eqref{eq:2G2Azag} 	uniformly in $(z_1, z_2)$-regular matrices $A = \mathring{A}^{z_1, z_2} \in \C^{N \times N}$ uniformly in all times $s\in [0, s_k]$. 
\end{proposition}

\begin{proof}[Proof of Theorem \ref{thm:2GLL}]
Similarly to the proof of Theorem \ref{theo:main1}, the proof of Theorem \ref{thm:2GLL} follows by a simple induction involving Propositions \ref{prop:global2}, \ref{prop:zigstep2}, and \ref{prop:zagstep2}. 
\end{proof}

\section{Single resolvent law: Proof of Propositions \ref{prop:global}, \ref{prop:zigstep}, and \ref{prop:zagstep}}
The proofs of the three propositions are given in Sections \ref{subsec:zigstep}, \ref{subsec:zagstep}, and \ref{subsec:global} below. Throughout the entire section, we will very often drop the subscript $\lambda$ from $G$, $M$'s, and $\rho$'s. 
\subsection{Characteristic flow: Proof of Proposition \ref{prop:zigstep}} \label{subsec:zigstep} 
We conduct the proof in the complex Hermitian case with $\sigma = 0$, the obvious modifications in the real symmetric case\footnote{For a detailed treatment, we refer to \cite[Section 4]{cipolloni2023eigenstate}.} are left to the reader. Also, note that it suffices to prove the statement only for fixed $z$ and $t$, since uniformity follows by a standard simple grid argument. 

In the following, we often drop the index $k$ as it remains fixed and introduce $t_{\rm init} := t_{k-1}$ and $t_{\rm fin} := t_{k}$. For $t \in [t_{\rm init}, t_{\rm final}]$, we denote $G_t:=\big(\mathfrak{F}_{\rm zig}^{t - t_{\rm init}}[H_{k-1}] - z_t\big)^{-1}$ for $z_t := \varphi_{t, t_{\rm final}}(z)$ with $\varphi$ defined in \eqref{eq:flowmap}, and $M_t:= M_{\lambda, t}(z_t)$. 

By It\^{o}'s formula and $\partial_t M_t=M_t/2$, we obtain
\begin{equation}
	\label{eq:flows}
	\begin{split}
		\dd \langle (G_t-M_t)B\rangle&=\dd N_{\mathrm{av},t}+\frac{1}{2}\langle (G_t-M_t)B\rangle\dd t+\lambda^2\langle G_t-M_t\rangle\langle G_t^2B\rangle\dd t, \\
		\dd (G_t-M_t)_{{\bm x}{\bm y}}&=\dd N_{\mathrm{iso},t}+\frac{1}{2}(G_t-M_t)_{{\bm x}{\bm y}}\dd t+\lambda^2\langle G_t-M_t\rangle (G_t^2)_{{\bm x}{\bm y}}\dd t,
	\end{split}
\end{equation}
where
\[
\dd N_{\mathrm{av},t}=\frac{\lambda }{\sqrt{N}}\sum_{a,b=1}^N\partial_{ab}\langle G_t\rangle \,\d (\dd B_t)_{ab}, \qquad\quad \dd N_{\mathrm{iso},t}=\frac{\lambda}{\sqrt{N}}\sum_{a,b=1}^N\partial_{ab}(G_t)_{{\bm x}{\bm y}} \,\d (\dd B_t)_{ab}
\]
and $\partial_{ab}$ denotes the partial derivative with respect to the matrix entry $(\mathfrak{F}_{\rm zig}^{t - t_{\rm init}}[H_{k-1}])_{ab}$. 
In the following we first consider the case $B=I$, then we explain the minor changes needed in the general case $B\ne I$. For an arbitrarily small $\xi > 0$ (in particular, $\xi \le \epsilon/10$ with $\epsilon$ in the condition $N \eta_t \rho_t \ge N^{\epsilon}$), we define the stopping time 
\begin{equation}
	\label{eq:defstoptime}
	\begin{split}
	\tau:=\inf\Bigg\{t \in [t_{\rm init}, t_{\rm final}] : \  &\big| \langle G_t-M_t\rangle\big|\ge \frac{N^{2\xi}}{N\eta_t}\sqrt{\frac{\langle \Im M_t BB^*\rangle}{\rho_t}}, \\ &\sup_{{\bm v},{\bm w}\in \{{\bm x},{\bm y}\}}\big|(G_t-M_t)_{{\bm v}{\bm w}}\big|\ge  N^{2\xi}\sqrt{\frac{(\Im M_t)_{{\bm v}{\bm v}}(\Im M_t)_{{\bm w}{\bm w}}}{N\eta_t \rho_t }}\Bigg\}
	\end{split}
\end{equation}
where we used the notation $\rho_t:= \pi^{-1}|\langle \Im M_t\rangle|$.

We start estimating the quadratic variation of the stochastic terms:
\begin{equation*}
	\begin{split}
		\mathrm{QV}\big[N_{\mathrm{av},t}\big]&\lesssim \int_{t_{\rm init}}^{t\wedge \tau } \frac{\lambda^2}{N^2\eta_s^2} \langle \Im G_s\Im G_s\rangle\, \dd s \lesssim  \int_{t_{\rm init}}^{t\wedge \tau} \frac{\lambda^2\rho_s}{N^2\eta_s^3}\, \dd s\lesssim \frac{1}{N^2\eta_{t\wedge \tau}^2}, \\
		\mathrm{QV}\big[N_{\mathrm{iso},t}\big]&\lesssim  \int_{t_{\rm init}}^{t \wedge \tau} \frac{\lambda^2}{N\eta_s^2} (\Im G_s)_{{\bm x}{\bm x}}(\Im G_s)_{{\bm y}{\bm y}}\, \dd s \\
		&\lesssim   \int_{t_{\rm init}}^{t\wedge \tau} \frac{\lambda^2}{N\eta_s^2} |(\Im M_s)_{{\bm x}{\bm x}}||(\Im M_s)_{{\bm y}{\bm y}}|\, \dd s\lesssim \frac{\rho_{t\wedge \tau}}{N\eta_{t\wedge \tau}}\cdot\frac{|(\Im M_{t\wedge \tau})_{{\bm x}{\bm x}}||(\Im M_{t\wedge \tau })_{{\bm y}{\bm y}}|}{\rho_{t\wedge \tau}^2}.
	\end{split}
\end{equation*}
We point out that here we used that $\Im M_s\sim \Im M_t$, for any $s, t$, by the line above \eqref{eq:flows} and that for $\alpha>1$ we have
\[
\int_0^t \frac{\lambda^2 \rho_s}{\eta_s^\alpha}\, \dd s\lesssim \frac{1}{\eta_t^{\alpha-1}}.
\]
Hence, by an application of the path-wise Burkholder-Davis-Gundy (BDG) inequality (see Appendix B.6, Eq.~(18) in \cite{shorack2009empirical}) we have
\begin{equation*}
\sup_{t_{\rm init}\le s\le t}\big| N_{\mathrm{av},s}\big|\lesssim \frac{N^\xi}{N\eta_{t\wedge \tau}}, \qquad\quad \sup_{t_{\rm init}\le s\le t}\big| N_{\mathrm{iso},s}\big|\lesssim N^\xi \sqrt{\frac{\rho_{t\wedge \tau}}{N\eta_{t\wedge \tau}}}\sqrt{\frac{(\Im M_{t\wedge \tau})_{{\bm v}{\bm v}}(\Im M_{t\wedge \tau})_{{\bm w}{\bm w}}}{\rho_{t \wedge \tau}^2}}.
\end{equation*}
Next, using the estimates
\begin{equation*}
	\big|\langle G_t^2\rangle\big|\le \frac{\rho_t}{\eta_t}, \qquad\quad \big|(G_t^2)_{{\bm x}{\bm y}}\big|\le \frac{\rho_t}{\eta_t} \sqrt{\frac{(\Im M_t)_{{\bm x}{\bm x}}(\Im M_t)_{{\bm y}{\bm y}}}{\rho_t^2}}.
\end{equation*}
for $t_{\rm init} \le t\le\tau$, we obtain
\begin{equation*}
	\begin{split}
		\big|\langle (G_{t\wedge \tau }-M_{t \wedge \tau})\rangle\big|&\le\big|\langle (G_{t_{\rm init}}-M_{t_{\rm init}})\rangle\big|+ \int_{t_{\rm init}}^{t\wedge \tau}\big|\langle G_s-M_s\rangle\big| \left(\frac{1}{2}+\frac{\lambda^2\rho_s}{\eta_s}\right)\dd s+\mathcal{O}\left(\frac{N^\xi}{N\eta_{t\wedge \tau}}\right), \\
		\big|(G_{t \wedge \tau }-M_{t\wedge \tau})_{{\bm x}{\bm y}}\big|&\le\big|(G_{t_{\rm init}}-M_{t_{\rm init}})_{{\bm x}{\bm y}}\big|+\mathcal{O}\left(\sqrt{\frac{\rho_{t\wedge \tau}}{N\eta_{t\wedge \tau}}}\sqrt{\frac{(\Im M_{t \wedge \tau})_{{\bm x}{\bm x}}(\Im M_{t\wedge \tau})_{{\bm y}{\bm y}}}{\rho_{t\wedge \tau}^2}}\right).
	\end{split}
\end{equation*}
Finally, using that 
\[
\exp\left(\int_s^t \frac{\lambda^2\rho_r}{\eta_r}\,\dd r\right)\lesssim \frac{\eta_s}{\eta_t},
\]
we conclude
\begin{equation}
\label{eq:boundfinsg}
	\begin{split}
		\big|\langle (G_{t\wedge \tau}-M_{t\wedge \tau})\rangle\big|&\le \frac{\eta_{t_{\rm init}}}{\eta_{t\wedge \tau}}\big|\langle (G_{t_{\rm init}}-M_{t_{\rm init}})\rangle\big|+\mathcal{O}\left(\frac{N^\xi\log N}{N\eta_{t\wedge \tau}}\right), \\
		\big|(G_t-M_t)_{{\bm x}{\bm y}}\big|&\le\big|(G_{t_{\rm init}}-M_{t_{\rm init}})_{{\bm x}{\bm y}}\big|+\mathcal{O}\left(\sqrt{\frac{\rho_{t\wedge \tau}}{N\eta_{t\wedge \tau}}}\sqrt{\frac{(\Im M_{t \wedge \tau})_{{\bm x}{\bm x}}(\Im M_{t\wedge \tau})_{{\bm y}{\bm y}}}{\rho_{t\wedge \tau}^2}}\right).
	\end{split}
\end{equation}
Finally, using the assumption in Proposition \ref{prop:zigstep} to estimate the initial condition in \eqref{eq:boundfinsg}, we obtain
\[
\big|\langle (G_t-M_t)\rangle\big|\le \frac{N^\xi\log N}{N\eta_t}, \qquad\quad \big|(G_t-M_t)_{{\bm x}{\bm y}}\big|\le\sqrt{\frac{\rho_t}{N\eta_t}}\sqrt{\frac{(\Im M_t)_{{\bm x}{\bm x}}(\Im M_t)_{{\bm y}{\bm y}}}{\rho_t^2}}
\]
for $t_{\rm init} \le t\le \tau$, and we thus have $\tau=t_{\rm final}$, concluding the proof for $B=I$.

For general $B\in\C^{N\times N}$ the analysis of the first equation in \eqref{eq:flows} is simpler since for its last term we do not need to use a Gronwall inequality but rather just plug in the result we just obtained for $B=I$. More precisely, we can bound this term by (here we consider ${t_{\rm init}} \le t\le \tau$)
\begin{equation*}
\lambda^2 \int_{t_{\rm init}}^{t} \frac{N^\xi}{N\eta_s^2}\langle \Im G_s\rangle^{1/2}\langle \Im G_s BB^*\rangle^{1/2}\,\dd s\lesssim \lambda^2 \int_{t_{\rm init}}^{t} \frac{N^\xi \rho_s}{N\eta_s^2}\langle \Im M_s BB^*\rangle^{1/2}\,\dd s\lesssim \frac{N^\xi}{N\eta_t}\sqrt{\frac{\langle \Im M_{t} BB^*\rangle}{\rho_{t }}}
\end{equation*}
Similarly, we estimate the quadratic variation of the stochastic term in first equation in \eqref{eq:flows} by
\begin{equation*}
\int_{t_{\rm init}}^{t } \frac{1}{N^2\eta_s^2}\langle \Im G_s B\Im G_s B^*\rangle\, \dd s\lesssim \int_{t_{\rm init}}^{t} \frac{1}{N^2\eta_s^3}\langle \Im G_s BB^*\rangle\, \dd s\lesssim \frac{1}{N^2\eta_t^2} \frac{\langle \Im M_t BB^*\rangle}{\rho_t}.
\end{equation*}
Proceeding as in \eqref{eq:boundfinsg} (and the formula below it) we conclude the proof also for general $B\in\C^{N\times N}$. \qed

\subsection{Green function comparison: Proof of Proposition \ref{prop:zagstep}} \label{subsec:zagstep} 
The goal of this section is to prove Proposition \ref{prop:zagstep} and thereby conclude the argument for the \emph{zag} step of our inductive zigzag proof. Since throughout the argument, the time $t_k$ defined in \eqref{eq:tkdef} remains fixed, we shall henceforth drop the subscript $t_k$. Additionally, to further condense the notation, we abbreviate $\mathcal{D} := \mathcal{D}_{t_k}^{\rm abv}$ and $s_{\rm final} := s_k$. 

Our proof will be conducted in an iterative way along vertical truncations of $\mathcal{D}$ given by
\begin{equation}
	\mathcal{D}_\gamma := \{ z  \in \mathcal{D} : |\Im z| \ge \lambda N^{-1+ \gamma} \} \quad \text{for} \quad  \gamma \in (0, \infty)\,. 
\end{equation}
This is the content of the following proposition, whose proof is given in Section \ref{subsubsec:Gronwall} below. 
\begin{proposition}[Zag bootstrap] \label{prop:zagiterate}
	Fix $\gamma_0 \in (0,\infty)$ and assume that the resolvent $G^s = G^s_\lambda$ from Proposition \ref{prop:zagstep} satisfies the bounds  
	\begin{equation} \label{eq:GFTinput}
		|(G^s(z))_{\bm u \bm v}| \lesssim |M(z)_{\bm u \bm v}| + \sqrt{ (\Im M(z))_{\bm u \bm u} (\Im M(z))_{\bm v \bm v}}\,, \quad  |(\Im G^s(z))_{\bm u \bm v}| \lesssim  \sqrt{ (\Im M(z))_{\bm u \bm u} (\Im M(z))_{\bm v \bm v}} 
	\end{equation}
	with very high probability, uniformly in $z \in \mathcal{D}_{\gamma_0}$, deterministic vectors $\bm u, \bm v \in \C^{N}$, and $s \in [0, \sfin]$. 
	
	Fix an arbitrarily small step size exponent $\delta > 0$ and let $\gamma_1 \ge \gamma_0 - \delta$. Assume that the resolvent $G^s$ satisfies the local laws \eqref{eq:globallaw} uniformly in $z \in \mathcal{D}_{\gamma_1}$, $B \in \C^{N \times N}$, and $\bm x, \bm y \in \C^N$ at time $s = \sfin$. Then the resolvent $G^s$ satisfies the local laws \eqref{eq:globallaw} uniformly in $z \in \mathcal{D}_{\gamma_1}$ $B \in \C^{N \times N}$, and $\bm x, \bm y \in \C^N$, uniformly for all times $s \in [0, \sfin]$.  
\end{proposition}
Armed with Proposition \ref{prop:zagiterate}, we can readily conclude the proof of Proposition \ref{prop:zagstep}. 
\begin{proof}[Proof of Proposition \ref{prop:zagstep}] 
	Analogously to \cite[Proof of Proposition 3.7]{erdHos2025cusp}, the argument goes via induction in $\gamma(k) := \kappa_0 + 1/\epsilon - k \delta$ by iteratively applying Proposition \ref{prop:zagiterate}. For the base case, we clearly have the validity of the bounds \eqref{eq:GFTinput} for $\gamma(0) = \kappa_0 + 1/\epsilon$ using trivial elementary bounds on $G$ and $M$.\footnote{In fact, it even holds deterministically.} After $K := \lceil (\kappa_0 + 1/\epsilon - \epsilon)/\delta \rceil $ steps, we have that $\mathcal{D}_{\gamma(K)} = \mathcal{D}$, and hence proven Proposition \ref{prop:zagiterate}. 
\end{proof}

\subsubsection{Gronwall estimates: Proof of Proposition \ref{prop:zagiterate}} \label{subsubsec:Gronwall}
The goal of this section is to give the proof of Proposition \ref{prop:zagiterate}. The argument is based on two Gronwall estimates, formulated in Propositions \ref{prop:zagisoGron}--\ref{prop:zagavGron} below, which, in turn, are proven in the following Section \ref{subsubsec:cumexp}. We note that the isotropic part of Proposition \ref{prop:zagiterate} will be proven in an entirely self-consistent way, based solely on the \emph{isotropic Gronwall estimate} in Proposition \ref{prop:zagisoGron} below. This serves as an input for the \emph{average Gronwall estimate} in Proposition \ref{prop:zagavGron}. 

The basic tool for establishing the Gronwall estimates is a well known \emph{cumulant expansion}, which we recall within the proof in Section \ref{subsubsec:cumexp}. A second key input in the arguments is the following monotonicity estimate on resolvents (Lemma \ref{lem:monotone}), whose proof is identical to the one of \cite[Lemma~5.3]{erdHos2025cusp} and hence omitted. For simplicity of the notation, we shall henceforth suppose that our spectral parameter $z = E + \ii \eta \in \HH$ is in the upper half plane and identify any domain $\mathcal{D}$ with its restriction $\mathcal{D} \cap \HH$. 

\begin{lemma}[Monotonicity estimate] \label{lem:monotone}
	Fix a constant $\gamma_0 \in (0,\infty)$ and assume the very-high-probability-bounds \eqref{eq:GFTinput} hold uniformly in $z \in \mathcal{D}_{\gamma_0}$, $s \in [0, \sfin]$, and deterministic vectors $\bm u, \bm v \in \C^N$. 
	
	Then, for any fixed $\gamma_1 \ge \gamma_0 - \delta$, we have that
	\begin{equation}
\begin{split}
	|(G^s(E + \ii \eta_1 ))_{\bm u \bm v}| &\lesssim \frac{\eta_0}{\eta_1}\left(|M(E+ \ii \eta_0 )_{\bm u \bm v}| + \sqrt{ (\Im M(E + \ii \eta_0))_{\bm u \bm u} (\Im M(E + \ii \eta_0))_{\bm v \bm v}}\right) \\
	  |(\Im G^s(E + \ii \eta_1))_{\bm u \bm v}| &\lesssim  \frac{\eta_0}{\eta_1}\sqrt{ (\Im M(E+ \ii \eta_0))_{\bm u \bm u} (\Im M(E + \ii \eta_0))_{\bm v \bm v}} 
\end{split}
	\end{equation}
	with very high probability, uniformly in $z = E + \ii \eta_1 \in \mathcal{D}_{\gamma_1}$ for any $\eta_0 \ge \lambda N^{-1+\gamma_0 } \vee \eta_1$, time $s \in [0, \sfin]$, and deterministic vectors $\bm u, \bm v \in \C^N$. 
\end{lemma}

We are now prepared to formulate our isotropic and average Gronwall estimates and remind the reader that within the \emph{zag} step, the deterministic approximation $M$ is time-independent. 
\begin{proposition}[Isotropic Gronwall estimate] \label{prop:zagisoGron}
	Assume the conditions of Proposition \ref{prop:zagiterate}. Fix $\bm x, \bm y \in \C^N$, $z = E + \ii \eta_1 \in \mathcal{D}_{\gamma_1}$ and $\eta_0 \ge \lambda N^{-1+\gamma_0} \vee \eta_1$ such that $\eta_0 /\eta_1 \le N^\delta$. For $s \in [0,\sfin]$, define
	\begin{equation}
		S_s := (G^s(E + \ii \eta_1) - M(E + \ii \eta_1))_{\bm x \bm y} \,. 
	\end{equation}
	Then, for any large (even) $p\in \N$ it holds that
	\begin{equation} \label{eq:isoGron}
		\left| \frac{\dd }{\dd s} \E |S_s|^p \right| \lesssim  \left(1 +  N^{-2 \delta} \lambda^2 {\frac{\rho(E+ \ii \eta_0)}{\eta_0}} \right)\, \left[\E |S_s|^p + (N^{10\delta}\Psi(\eta_1; \bm x, \bm y))^p\right]
	\end{equation}
	uniformly in $s \in [0,\sfin]$, $\bm x, \bm y \in \C^N$, and $z \in \mathcal{D}_{\gamma_1}$. Here, for $\eta \in [\eta_1, \eta_0]$,  we denoted
	\begin{equation}
		\Psi(\eta; \bm x, \bm y) := \sqrt{\frac{\rho(E+ \ii \eta)}{N \eta}} \ \sqrt{\frac{\big(\Im M(E + \ii \eta)\big)_{\bm x \bm x} \, \big(\Im M(E+ \ii \eta)\big)_{\bm y \bm y}}{ \rho(E + \ii \eta)^2}} \,. 
	\end{equation}
\end{proposition}
Armed with Proposition \ref{prop:zagavGron}, we can now apply Gronwall's lemma: Uniformly in $s \in [0, \sfin]$, from \eqref{eq:isoGron}, we find that 
\begin{equation} \label{eq:integrateisoGron}
\begin{split}
\E |S_s|^p &\lesssim \exp\left[ \bigg(1 + N^{-2 \delta} \lambda^2 \frac{\rho(E+ \ii \eta_0)}{\eta_0}\bigg) (\sfin - s)\right] \, \left[\E |S_{\sfin}|^p + (N^{10\delta}\Psi(\eta_1; \bm x, \bm y))^p\right] \\
&\lesssim \exp(N^{-\delta}) \left[\E |S_{\sfin}|^p + (N^{10\delta}\Psi(\eta_1; \bm x, \bm y))^p\right] \lesssim \left[\E |S_{\sfin}|^p + (N^{10\delta}\Psi(\eta_1; \bm x, \bm y))^p\right]
\end{split}
\end{equation}
where we used that $\lambda^2 \rho(E+ \ii \eta_0)/\eta_0 \lesssim N^{k \delta}$ by \eqref{eq:tkdef} and $\sfin \lesssim N^{-(k-1)\delta}$ by \eqref{eq:dtkdef}. 

To estimate $\E |S_{\sfin}|^p$, we recall that $G^s$ satisfies the isotropic local law in \eqref{eq:globallaw} uniformly in $\mathcal{D}_{\gamma_1}$ and $\bm x, \bm y \in \C^N$ at $s = \sfin$. Hence, because $p$ and $\delta$ were arbitrary, we deduce that 
\begin{equation} \label{eq:isolaw}
\big|\big(G^s(E + \ii \eta_1) - M(E + \ii \eta_1)\big)_{\bm x \bm y}\big| \prec \Psi(\eta_1; \bm x, \bm y)
\end{equation}
uniformly in $z = E + \ii \eta_1 \in \mathcal{D}_{\gamma_1}$, $\bm x, \bm y \in \C^N$ and times $s \in [0, \sfin]$, which concludes the proof of the isotropic part of Proposition \ref{prop:zagiterate}. 

We are hence left with proving the average part of Proposition \ref{prop:zagiterate}. 
\begin{proposition}[Average Gronwall estimate] \label{prop:zagavGron}
	Assume the conditions of Proposition \ref{prop:zagiterate}. Fix $B \in \C^{N\times N}$, $z = E + \ii \eta_1 \in \mathcal{D}_{\gamma_1}$ and $\eta_0 \ge \lambda N^{-1+\gamma_0} \vee \eta_1$ such that $\eta_0 /\eta_1 \le N^\delta$. For $s \in [0,\sfin]$, define
	\begin{equation}
		R_s := \langle (G^s(E + \ii \eta_1) - M(E + \ii \eta_1))B \rangle \,. 
	\end{equation}
	Then, for any large (even) $p\in \N$ it holds that
	\begin{equation} \label{eq:avGron}
		\left| \frac{\dd }{\dd s} \E |R_s|^p \right| \lesssim \left(1 +   N^{-2 \delta} \lambda^2 {\frac{\rho(E+ \ii \eta_0)}{\eta_0}}\right) \, \left[\E |R_s|^p + \big(N^{20\delta} \Phi(\eta_1; B)\big)^p\right]
	\end{equation}
	uniformly in $s \in [0,\sfin]$, $B \in \C^{N \times N}$, and $z \in \mathcal{D}_{\gamma_1}$. Here, for $\eta \in [\eta_1, \eta_0]$,  we denoted
	\begin{equation}
		\Phi(\eta; B) := {\frac{1}{N \eta}} \ \sqrt{\frac{\langle \Im M(E + \ii \eta) BB^* \rangle}{ \rho(E + \ii \eta)}} \,. 
	\end{equation}
\end{proposition}

Completely analogously to \eqref{eq:integrateisoGron}, armed with \eqref{eq:avGron}, we can apply Gronwall's lemma, to find that 
\begin{equation} \label{eq:integrateavGron}
	\begin{split}
		\E |R_s|^p &\lesssim \exp\left[ \bigg(1 + N^{-2 \delta} \lambda^2 \frac{\rho(E+ \ii \eta_0)}{\eta_0}\bigg) (\sfin - s)\right] \, \left[\E |R_{\sfin}|^p + (N^{20\delta}\Phi(\eta_1; B))^p\right] \\
		&\lesssim \exp(N^{-\delta}) \left[\E |R_{\sfin}|^p + (N^{20\delta}\Psi(\eta_1; \bm x, \bm y))^p\right] \lesssim \left[\E |S_{\sfin}|^p + (N^{20\delta}\Phi(\eta_1; B))^p\right]
	\end{split}
\end{equation}
again using that $\lambda^2 \rho(E+ \ii \eta_0)/\eta_0 \lesssim N^{k \delta}$ by \eqref{eq:tkdef} and $\sfin \lesssim N^{-(k-1)\delta}$ by \eqref{eq:dtkdef}. 

To estimate $\E |R_{\sfin}|^p$, we recall that $G^s$ satisfies the average local law in \eqref{eq:globallaw} uniformly in $\mathcal{D}_{\gamma_1}$ and $B\in \C^{N\times N}$ at $s = \sfin$. Hence, because $p$ and $\delta$ were arbitrary, we deduce that 
\begin{equation} \label{eq:avlaw}
	\big|\big\langle \big(G^s(E + \ii \eta_1) - M(E + \ii \eta_1)\big)B \big\rangle\big| \prec \Phi(\eta_1; B)
\end{equation}
uniformly in $z = E + \ii \eta_1 \in \mathcal{D}_{\gamma_1}$, $B \in \C^{N \times N}$ and times $s \in [0, \sfin]$, which concludes the proof of the average part of Proposition \ref{prop:zagiterate} and thus the entire proof of Proposition \ref{prop:zagiterate}. \qed
\subsubsection{Cumulant expansion: Proof of the Gronwall estimates in Propositions \ref{prop:zagisoGron}--\ref{prop:zagavGron}}
In this section, we give the proofs of Propositions \ref{prop:zagisoGron} and \ref{prop:zagavGron}. 
\subsubsection*{Isotropic case: Proof of Proposition \ref{prop:zagisoGron}} \label{subsubsec:cumexp}
We perform a standard cumulant expansion (see, e.g., \cite[Proposition 3.2]{slow_corr}) to find that 
\begin{equation} \label{eq:cumexbasiciso}
	\left| \frac{\dd }{\dd s} \E |S_s|^p \right| \lesssim \sum_{k=3}^L \frac{\lambda^k}{N^{k/2}}\sum_{l=0}^k \left| \sum_{a,b} \E\big[ \partial_{ab}^l \partial_{ba}^{k-l} |S_s|^p\big]\right| + \mathcal{O}\left( \big(\Psi(\eta_1; \bm x, \bm y)\big)^p  \right)
\end{equation}
for some $L = \mathcal{O}(1)$ large enough. More precisely, $L$ will be chosen as follows: As every derivative acting on $S_s$ creates a new resolvent, we can estimate the resulting expression by means of Lemma \ref{lem:monotone}, and additionally using Lemma \ref{lem:Mbounds} and monotonicity of $\eta \mapsto \Im M(E + \ii \eta)/\eta $ (decreasing) and $\eta \mapsto \eta \rho(E + \ii \eta)$ (increasing) as
\begin{equation*}
	\begin{split}
&\frac{\lambda^{L+1}}{N^{(L+1)/2}} N^2 N^{(\kappa_0 + 1/\epsilon )p}N^{(p+L)\delta}\frac{ \left(\sqrt{(\Im M(E + \ii \eta_0))_{\bm x \bm x} (\Im M(E + \ii \eta_0))_{\bm y \bm y}}\right)^p}{(\eta_0 \lambda^2 \rho(E + \ii \eta_0))^{(L+1)/2}} \\
\lesssim \, &N^2 N^{(\kappa_0 + 1/\epsilon) p} N^{(2p+L) \delta} \frac{1}{(N \eta_0 \rho(E + \ii \eta_0))^{(L+1-p)/2}} \big(\Psi(\eta_1; \bm x, \bm y)\big)^p \\ \lesssim \, &N^{2 + (\kappa_0 + 1/\epsilon +\epsilon/2 +2 \delta) p} N^{-L(\epsilon/2 - \delta)} \big(\Psi(\eta_1; \bm x, \bm y)\big)^p \lesssim \big(\Psi(\eta_1; \bm x, \bm y)\big)^p
	\end{split}
\end{equation*}
where the last estimate holds when choosing 
\begin{equation} \label{eq:Ldef}
L := \left\lceil \big( 2 + (\kappa_0 + 1/\epsilon +\epsilon/2 +2 \delta) p \big)/ (\epsilon/2 - \delta) \right\rceil \,. 
\end{equation}
We start by discussion terms of order $k=3$ in \eqref{eq:cumexbasiciso} and henceforth drop the sub-/super-script $s$ for ease of notation. We also ignore taking adjoints and complex conjugations, since they are irrelevant for our estimates. To further simplify the notation, we follow the convention that every resolvent is understood having $E+ \ii \eta_1$ as an argument, i.e.~$G \equiv G(E+ \ii \eta_1)$, while the deterministic approximation and the density is understood with argument $E + \ii \eta_0$, i.e.~$\rho \equiv \rho(E+ \ii \eta_0)$ and $M \equiv M(E + \ii \eta_0)$. The three derivatives can have hit up to three factors of $S_s$, producing terms of the forms (i) $(\partial^3 S) |S|^{p-1}$, (ii) $(\partial S) (\partial^2 S)|S|^{p-2}$, and (iii) $(\partial S)^3 |S|^{p-3}$. 

As the first exemplary term in case (i), we estimate (neglecting a $|S|^{p-1}$ factor)
\begin{equation} \label{eq:case1ex1}
	\begin{split}
		&\lambda^3 N^{-3/2} \left|  \sum_{a,b}  G_{\bm x a} G_{bb} G_{aa} G_{b \bm y}  \right|  \\
		\lesssim &\lambda^3 N^{-3/2} N^{2 \delta} \sum_{a,b} \left| G_{\bm x a} (|M|)_{aa} (|M|)_{bb} G_{b \bm y} \right| \\
		\lesssim  &\lambda^3 N^{-3/2} N^{2 \delta}  \sqrt{\sum_a |G_{\bm x a}|^2} \sum_{a,b} \big|M_{ab}\big|^2 \, \sqrt{\sum_b |G_{b \bm y}|^2} \\
		\lesssim &  \lambda^3 N^{-1/2} N^{2 \delta}\langle |M|^2 \rangle \eta_1^{-1} \sqrt{(\Im G)_{\bm x \bm x} (\Im G)_{\bm y \bm y}} \\
		\lesssim & \lambda N^{-1/2} N^{4 \delta} \eta_1^{-1}\sqrt{(\Im M)_{\bm x \bm x} (\Im M(E + \ii \eta_0))_{\bm y \bm y}} 
		\lesssim  \lambda N^{8 \delta}\sqrt{\frac{\rho}{\eta_0}} \Psi(\eta_1; \bm x, \bm y) \,. 
	\end{split}
\end{equation}
In the first step, we used \eqref{eq:GFTinput} together with $(\Im M)_{\bm u \bm u} + |M_{\bm u \bm u}| \lesssim (|M|)_{\bm u \bm u}$ and Lemma \ref{lem:monotone} to bound $G_{bb}$ and $G_{aa}$. To go to the third line, we employed a Schwarz inequality and estimated 
\begin{equation} \label{eq:extendsum}
	\sum_a |(|M|)_{aa}|^2 \le \sum_{ab} |(|M|)_{ab}|^2 = \sum_{ab} |M_{ab}|^2 = N \langle |M|^2 \rangle
\end{equation}
by \emph{extending the summation}. 
In the following step, we used a Ward identity $\sum |G_{\bm xa}|^2 = (\Im G)_{\bm x \bm x}/\eta_1$. To go to the ultimate line, we employed \eqref{eq:GFTinput} together with Lemma \ref{lem:monotone} to bound
\begin{equation*}
	(\Im G)_{\bm x \bm x} \lesssim N^\delta (\Im M)_{\bm x \bm x}
\end{equation*}
uniformly in $\bm x \in \C^N$, and Lemma \ref{lem:Mbounds}. In the last step, we used that $\eta \mapsto \rho(E+ \ii \eta)/\eta$ and $\eta \mapsto (\Im M(E + \ii \eta))_{\bm u \bm u}/\eta$ (for any $\bm u \in \C^N$) are monotonically decreasing, and $\eta \mapsto \eta \rho(E + \ii \eta)$ is monotonically increasing. 

Next, also in case (i), we estimate (again neglecting a $|S|^{p-1}$ factor inside the expectation)
\begin{equation} \label{eq:case1ex2}
	\begin{split}
		&\lambda^3 N^{-3/2} \left|  \sum_{a,b}  G_{\bm x a} G_{ba} G_{bb} G_{a \bm y}  \right|  \\
		\lesssim &\lambda^2 N^{-3/2} N^{2 \delta} \sum_{a,b} \left| G_{\bm x a} \sqrt{(\Im M)_{bb}/\rho} \,  (|M|)_{bb} G_{a \bm y} \right| \\
		\lesssim  &\lambda^2 N^{-3/2} N^{2 \delta}  \lambda^{-1}\sqrt{\sum_a |G_{\bm x a}|^2} \,   \sqrt{\sum_a |G_{a \bm y}|^2} \left(\sum_{b} \big|(|M|)_{bb}\big|^2\right)^{1/2} \, \left(\sum_b \frac{(\Im M)_{bb}}{\rho}\right)^{1/2} \\
		\lesssim &  \lambda^2 N^{-1/2} N^{2 \delta}\langle |M|^2 \rangle^{1/2} \eta_1^{-1} \sqrt{(\Im G)_{\bm x \bm x} (\Im G)_{\bm y \bm y}} \\
		\lesssim & \lambda N^{-1/2} N^{4 \delta} \eta_1^{-1}\sqrt{(\Im M)_{\bm x \bm x} (\Im M)_{\bm y \bm y}} 
		\lesssim \lambda N^{8 \delta}\sqrt{\frac{\rho}{\eta_0}} \Psi(\eta_1; \bm x, \bm y) \,. 
	\end{split}
\end{equation}
similarly to \eqref{eq:case1ex1}. More precisely, in the first step, we used \eqref{eq:GFTinput} together with Lemma \ref{lem:Mbounds} and Lemma~\ref{lem:Mmax1} below, to estimate 
\begin{equation*}
	\max_a |G_{\bm x a}| \lesssim N^\delta \lambda^{-1} \sqrt{(\Im M)_{\bm x \bm x}/\rho} 
\end{equation*} 
uniformly in $\bm x \in \C^N$. 
To go to the third line, we then used a Schwarz inequality. In the following step, we used a Ward identity, \eqref{eq:extendsum}, and the fact that $\langle \Im M \rangle = \pi \rho$. To go to the ultimate line, we then used Lemma \ref{lem:Mbounds}, and \eqref{eq:GFTinput} together with Lemma \ref{lem:monotone}. 

\begin{lemma} \label{lem:Mmax1}
	Uniformly in $\bm x \in \C^N$ and $z \in \mathbb{H}$, it holds that 
	\begin{equation} \label{eq:Mmax1}
		\sup_{\Vert \bm y \Vert = 1} |M_{\bm x \bm y}(z)|  + \sup_{\Vert \bm y \Vert = 1} |M_{\bm y\bm x}(z)| \lesssim \lambda^{-1} \sqrt{(\Im M(z))_{\bm x \bm x}/\rho(z)}
	\end{equation}
\end{lemma}
\begin{proof}
	We only prove the bound for the first term in \eqref{eq:Mmax1}; the other one is analogous and hence omitted. We also omit the argument $z$ of $M$ and $\rho$. Then, since $\bm y \mapsto |M_{\bm x \bm y}|$ is continuous, it holds that hat $\sup_{\Vert \bm y \Vert = 1} |M_{\bm x \bm y}| = |M_{\bm x\bm y_*}|$ for some $\bm y_* \in \C^N$ of norm one. Hence, denoting the spectral decomposition of $H_0$ by $H_0 = \sum_j \mu_j \ket{\bm u_j} \bra{\bm u_j}$, we have
	\begin{equation*}
		\begin{split}
			|M_{\bm x\bm y_*}| &\le \sum_{j}\frac{|\langle \bm x, \bm u_j \rangle \langle \bm u_j, \bm y_{*}\rangle|}{|\mu_{j} - z - \lambda^2 \langle M_\lambda\rangle|} 
			\lesssim \sqrt{\sum_j\frac{|\langle \bm x, \bm u_{j}\rangle|^2}{|\mu_{j} - z - \lambda^2 \langle M\rangle|^2}} \sqrt{\sum_j |\langle \bm u_j, \bm y_{*}\rangle|^2} \\
			&\lesssim \lambda^{-1}\sqrt{\sum_a\frac{|\langle \bm x, \bm u_{j}\rangle|^2 (\eta + \lambda^2 \rho)}{|\mu_{j} - z - \lambda^2 \langle M\rangle|^2 \rho}} = \lambda^{-1} \sqrt{(\Im M)_{\bm x \bm x}/\rho} \,. \qedhere
		\end{split}
	\end{equation*}
\end{proof}

We now turn to case (ii) and continue neglecting the factors of $|S|$, that are not differentiated. Again, we consider two exemplary terms. For the first one, we estimate
\begin{equation} \label{eq:case2ex1}
	\begin{split}
		&\lambda^3 N^{-3/2} \left|  \sum_{a,b}  G_{\bm x a} G_{b\bm y} G_{\bm x a} G_{bb} G_{a \bm y}  \right|  \\
		\lesssim &\lambda^2 N^{-3/2} N^{2 \delta} \sum_{a,b} \left| G_{\bm x a} G_{b \bm y} \sqrt{(\Im M)_{\bm x \bm x}/\rho} \, (|M|)_{bb} G_{a \bm y} \right| \\
		\lesssim  &\lambda^2 N^{-3/2} N^{2 \delta} \sqrt{(\Im M)_{\bm x \bm x}/\rho}  \sqrt{\sum_a |G_{\bm x a}|^2} \,  \left(\sum_{a,b} \big|M_{ab}\big|^2\right)^{1/2} \, \sum_a |G_{a \bm y}|^2 \\
		\lesssim &  \lambda^2 N^{-1} N^{2 \delta}\langle |M|^2 \rangle^{1/2} \sqrt{(\Im M)_{\bm x \bm x}/\rho} \,  \eta_1^{-3/2}\sqrt{(\Im G)_{\bm x \bm x} } (\Im G)_{\bm y \bm y} \\
		\lesssim & \lambda N^{-1} N^{4 \delta} \sqrt{(\Im M)_{\bm x \bm x}/\rho} \,  \eta_1^{-3/2}\sqrt{(\Im M)_{\bm x \bm x} } (\Im M(E + \ii \eta_0))_{\bm y \bm y} 
		\lesssim   N^{8 \delta}  \, \lambda\sqrt{\frac{\rho}{\eta_0}} \big(\Psi(\eta_1; \bm x, \bm y)\big)^2 \,. 
	\end{split}
\end{equation}

For the second exemplary term, we estimate 
\begin{equation} \label{eq:case2ex2}
	\begin{split}
		&\lambda^3 N^{-3/2} \left|  \sum_{a,b}  \big[ G_{\bm x a} G_{b\bm y} G_{\bm x a} G_{ba} G_{b \bm y} \big] \right|  \\
		\lesssim &\lambda^2 N^{-3/2} N^{2 \delta} \sqrt{(\Im M)_{\bm y \bm y}/\rho} \, \sum_{a,b} \left| G_{\bm x a} G_{b \bm y} G_{\bm x a} \big(\max_a |G_{ba}|\big)  \right| \\
		\lesssim  &\lambda N^{-3/2} N^{2 \delta} \sqrt{(\Im M)_{\bm y \bm y}/\rho} \, {\sum_a |G_{\bm x a}|^2} \,  \sqrt{ \sum_b |G_{b \bm y}|^2} \sqrt{\sum_b \frac{(\Im M)_{bb}}{\rho}} \\
		\lesssim &  \lambda N^{-1} N^{2 \delta} \sqrt{(\Im M)_{\bm y \bm y}/\rho} \, \eta_1^{-3/2} {(\Im G)_{\bm x \bm x} } \sqrt{(\Im G)_{\bm y \bm y}} \\
		\lesssim & \lambda N^{-1} N^{4 \delta} \eta_1^{-3/2} (\rho(E + \ii \eta_0))^{-1/2}{(\Im M)_{\bm x \bm x} } (\Im M(E + \ii \eta_0))_{\bm y \bm y} 
		\lesssim   N^{8 \delta} \, \lambda\sqrt{\frac{\rho}{\eta_0}} \big(\Psi(\eta_1; \bm x, \bm y)\big)^2  \,. 
	\end{split}
\end{equation}	
Finally, we turn to case (iii), in which case we estimate 
\begin{equation} \label{eq:case3}
	\begin{split}
		&\lambda^3 N^{-3/2} \left|  \sum_{a,b}  G_{\bm x a} G_{b\bm y} G_{\bm x a} G_{b \bm y} G_{\bm x a} G_{b \bm y} \right|  \\
		\lesssim &\lambda N^{-3/2} N^{2 \delta} \sum_{a,b} \left| G_{\bm x a} G_{b \bm y} G_{\bm x a} G_{b \bm y} \sqrt{(\Im M)_{\bm x \bm x}(\Im M)_{\bm y \bm y}/ \rho^2}  \right| \\
		\lesssim  &\lambda N^{-3/2} N^{2 \delta}   \sqrt{(\Im M)_{\bm x \bm x} (\Im M)_{\bm y \bm y} /\rho^2} \,  {\sum_a |G_{\bm x a}|^2} \,   \sum_b |G_{b \bm y}|^2 \\
		\lesssim &  \lambda N^{-3/2} N^{2 \delta} \eta_1^{-2} \sqrt{(\Im M)_{\bm x \bm x} (\Im M)_{\bm y \bm y}/\rho^2} \,  {(\Im G)_{\bm x \bm x} } (\Im G)_{\bm y \bm y} \\
		\lesssim & \lambda N^{-3/2} N^{4 \delta} \eta_1^{-2} \sqrt{(\Im M)_{\bm x \bm x} (\Im M)_{\bm y \bm y}/\rho^2} \, {(\Im M)_{\bm x \bm x} } (\Im M)_{\bm y \bm y} 
		\lesssim   N^{8 \delta} \lambda   \, \sqrt{\frac{\rho}{\eta_0}} \big(\Psi(\eta_1; \bm x, \bm y)\big)^3  \,. 
	\end{split}
\end{equation}

We now turn to terms of order $k \ge 4$, where we aim to control
\begin{equation} \label{eq:higherorderISO}
	\lambda^k N^{-k/2} \sum_{l=0}^k\left| \sum_{a,b} \E  \left[  \partial_{ab}^l \partial_{ba}^{k-l} |S_s|^p \right] \right| \,. 
\end{equation}
In the following, we focus on $l=0$; the other summands in \eqref{eq:higherorderISO} can be treated analogously and are hence omitted. Now, distributing the $k$ derivatives on $n \in [k]$ factors of $S$, we find that for $k_i \in \N$ such that $\sum_{i=1}^n k_i = k$  and (w.l.o.g.) $k_1 \le k_2 \le ... \le k_n$, the $l=0$ term of \eqref{eq:higherorderISO} can be rewritten as (ignoring the "untouched" factor $|S|^{p-n}$ and the expectation)
\begin{equation} \label{eq:higherorderISO2}
	\lambda^k N^{-k/2} \sum_{a,b} \prod_{i=1}^{n} \left|G_{\bm x a_{i_1}} G_{b_{i_1}a_{i_2}} ... G_{b_{i_{k_i-1}}a_{i_{k_i}}} G_{b_{i_{k_i}} \bm y}\right|
\end{equation}
where $a_{i_j} = a$ and $b_{i_j} = b$. First, if $k_2 = 1$, since there are at least two factors of $S$ which are hit by a single derivative, we find that 
\begin{equation} \label{eq:higherorderISOcase1}
	\begin{split}
		\eqref{eq:higherorderISO2} \lesssim& N^{-k/2} \lambda^{k} \left(\frac{1}{\lambda^2 \rho}\right)^{k-n} \lambda^{-2(n-2)} N^{(k+n-4)\delta}\left(\frac{(\Im M)_{\bm x \bm x} (\Im M)_{\bm y \bm y}}{\rho^2}\right)^{\frac{n-2}{2}} \sum_{a,b} |G_{\bm x a} G_{b \bm y}|^2 \\
		\lesssim& N^{(k+n-2)\delta} \lambda^4 \frac{\rho^2}{\eta_1^2} \left(\frac{1}{\sqrt{N} \lambda \rho}\right)^k \rho^n \left(\frac{(\Im M)_{\bm x \bm x} (\Im M)_{\bm y \bm y}}{\rho^2}\right)^{\frac{n}{2}}  \\
		\lesssim & N^{(k+3n) \delta}  \lambda \sqrt{\frac{\rho}{\eta_0}}  \left(\frac{1}{\sqrt{N} \lambda \rho}\right)^{k-n} \left(\frac{\eta_0}{\rho \lambda^2}\right)^{\frac{n-3}{2}}\big(\Psi(\eta_1; \bm x, \bm y) \big)^n  
		\lesssim  N^{8 \delta} \lambda \sqrt{\frac{\rho}{\eta_0}} \big(\Psi(\eta_1; \bm x, \bm y) \big)^n
	\end{split}
\end{equation}
where we used Lemma \ref{lem:Mmax1} and Lemma \ref{lem:Mbounds} together with Lemma \ref{lem:monotone} and \eqref{eq:GFTinput} in the first step. Additionally, we employed a Ward identity. In the penultimate step, we again used monotonicity of $\eta \mapsto \eta \rho(E + \ii \eta)$. Finally, in the last step, we used $N \lambda^2 \rho^2 \ge N^{10\delta}$ and that $\lambda^2\rho/\eta_0 \ge N^{10\delta}$ together with $n \ge 3$ (as a consequence of $n_2 = 1$).

Second, if $n_2 > 1$, we only use the first factor to bound
\begin{equation}
	\begin{split}
		\eqref{eq:higherorderISO2} \lesssim &N^{-k/2} \lambda^{k}  \left(\frac{1}{\lambda^2 \rho}\right)^{k-n-1} \lambda^{-2(n-1)} N^{(n+k-2)\delta} \left(\frac{(\Im M)_{\bm x \bm x} (\Im M)_{\bm y \bm y}}{\rho^2}\right)^{\frac{n-1}{2}} \sum_{a,b} |G_{\bm x a} G_{b \bm y} G_{ab}| \\
		\lesssim & N^{(n+k)\delta}\left(\frac{1}{\sqrt{N} \lambda \rho}\right)^{k-1} \rho^n \left(\lambda^2 \frac{\rho}{\eta_0}\right)^{3/2}\left(\frac{(\Im M)_{\bm x \bm x} (\Im M)_{\bm y \bm y}}{\rho^2}\right)^{\frac{n}{2}}   \\
		\lesssim & N^{(k+3n)\delta} \lambda \sqrt{\frac{\rho}{\eta_0}} \left(\frac{1}{\sqrt{N} \lambda \rho}\right)^{k-n-2} \left(\frac{\eta_0}{\rho \lambda^2}\right)^{\frac{n-1}{2}} 
		\big(\Psi(\eta_1; \bm x, \bm y) \big)^n 
		\lesssim N^{8 \delta}\lambda \sqrt{\frac{\rho}{\eta_0}}\big(\Psi(\eta_1; \bm x, \bm y) \big)^n
	\end{split}
\end{equation}
similarly to \eqref{eq:higherorderISOcase1}. In the last step, we additionally used that, since $n_2 > 1$, it necessarily holds that $k \ge n+2$. 

By collecting all the above estimates and application of Young inequalities, this concludes the proof of Proposition \ref{prop:zagisoGron}. \qed

\subsubsection*{Average case: Proof of Proposition \ref{prop:zagavGron}}
Analogously to the proof of Proposition \ref{prop:zagisoGron}, we perform a standard cumulant expansion to find that 
\begin{equation} \label{eq:cumexbasicav}
	\left| \frac{\dd }{\dd s} \E |R_s|^p \right| \lesssim \sum_{k=3}^L \frac{\lambda^k}{N^{k/2}}\sum_{l=0}^k \left| \sum_{a,b} \E\big[ \partial_{ab}^l \partial_{ba}^{k-l} |R_s|^p\big]\right| + \mathcal{O}\left( \big(\Phi(\eta_1; B)\big)^p  \right)
\end{equation}
for some $L = \mathcal{O}(1)$ large enough, chosen similarly to \eqref{eq:Ldef}. 

Just as in the proof of Proposition \ref{prop:zagisoGron}, we start by discussion terms of order $k=3$ in \eqref{eq:cumexbasicav} and henceforth drop the sub-/super-script $s$ for ease of notation. We also ignore taking adjoints and complex conjugations, since they are irrelevant for our estimates. However, now, since the isotropic law is already established, all resolvents, $M$, and $\rho$'s are understood to have the argument $E+ \ii \eta_1$, unless differently stated explicitly. The three derivatives can have hit up to three factors of $R_s$, producing terms of the forms (i) $(\partial^3 R) |R|^{p-1}$, (ii) $(\partial R) (\partial^2 R)|R|^{p-2}$, and (iii) $(\partial R)^3 |R|^{p-3}$.

As the first exemplary term in case (i), we estimate (note that the additional $N^{-1}$ arises from the averaged trace), neglecting a $|R|^{p-1}$ factor, similarly to \eqref{eq:case1ex1}, 
\begin{equation} \label{eq:case1ex1av}
	\begin{split}
		&\lambda^3 N^{-5/2} \left| \sum_{a,b} (GBG)_{ab} G_{aa} G_{bb}\right|  \\
		\lesssim &\lambda^3 N^{-5/2} \sqrt{\sum_{a,b} |(GBG)_{ab}|^2} \sum_a |(|M|)_{aa}|^2 \\
		\lesssim & \lambda N^{-1} \langle GBG G^* B^* G^* \rangle^{1/2} \lesssim N^{-2\delta} \lambda \sqrt{\frac{\rho(E+ \ii \eta_0)}{\eta_0}} \big(N^{20\delta} \Phi(\eta_1; B)\big) \,. 
	\end{split}
\end{equation}
In the first step, we employed the already established isotropic law to bound $|G_{aa}| \lesssim (|M|)_{aa}$ (same for $b$) and used a Schwarz inequality. To go to the last line, we used \eqref{eq:extendsum} together with Lemma \ref{lem:Mbounds}, and additionally a Ward identity together with the bound $\langle \Im G B \Im G B^* \rangle \le \eta^{-1} \langle \Im G BB^*\rangle$ following from a simple norm bound on $\Im G$ and positivity of the remaining matrix in the trace. In the ultimate step, we employed a spectral decomposition of $BB^* = \sum_j |\sigma_j|^2 \ket{\bm v_j} \bra{\bm v_j}$ to bound
\begin{equation} \label{eq:ImGImM}
	\langle \Im G BB^* \rangle = \frac{1}{N} \sum_j |\sigma_j|^2 (\Im G)_{\bm v_j \bm v_j} \lesssim \frac{1}{N} \sum_j |\sigma_j|^2 (\Im M)_{\bm v_j \bm v_j} = \langle \Im M BB^* \rangle 
\end{equation}
where the estimate follows from the isotropic law \eqref{eq:isoLL}. Moreover, we used the fact that  $\eta \mapsto \eta \rho(E + \ii \eta)$ is monotonically increasing. 

As another term in case (i) we consider
\begin{equation} \label{eq:case1ex2avstart}
	\lambda^3 N^{-5/2} \left| \sum_{a,b} (GBG)_{aa} G_{bb} G_{ab}\right|
\end{equation}
for which we decompose $G_{ab} = M_{ab} + (G_{ab} - M_{ab})$ and apply the isotropic law \eqref{eq:isoLL} for the $(G-M)$-term. Indeed, for this part of \eqref{eq:case1ex2avstart}, bounding $|(G-M)_{ab}| \lesssim N^\xi \Psi(\eta_1; \bm e_a, \bm e_b)$ we find
\begin{equation} \label{eq:case1ex2av}
	\begin{split}
		&\lambda^3 N^{-5/2} \left| \sum_{a,b} (GBG)_{aa} G_{bb} (G-M)_{ab}\right|  \\
		\lesssim &\lambda^3 N^{-5/2} \sqrt{\frac{\rho}{N \eta_1}} \sqrt{\sum_{a} |(GBG)_{aa}|^2} \sqrt{\sum_b |(|M|)_{bb}|^2 } \sum_a \frac{(\Im M)_{aa}}{\rho}\\
		\lesssim & N^{-2\delta} \lambda^2 {\frac{\rho(E+ \ii \eta_0)}{\eta_0}} \big(N^{20\delta} \Phi(\eta_1; B)\big) \,. 
	\end{split}
\end{equation}
In the last step we extended the summations over $a,b$ which are reduced by a square root, used that $\langle \Im M \rangle \sim \rho$, and argued similarly to \eqref{eq:case1ex1av}. The term with $M_{ab}$ instead of $(G-M)_{ab}$ can be handled similarly. 

Next, we study one exemplary terms in case (ii) and estimate
\begin{equation} \label{eq:case2av}
	\begin{split}
		&\lambda^3 N^{-7/2} \left| \sum_{a,b} (GBG)_{ab} (GBG)_{aa} G_{bb}\right| \\
		\lesssim & \lambda^3 N^{-7/2} \sqrt{\sum_{a,b} |(GBG)_{ab}|^2} \sqrt{\sum_a |(GBG)_{aa}|^2} \sqrt{\sum_b |(|M|)_{bb}|^2} \\
		\lesssim & N^{-2 \delta} \lambda^2 \frac{\rho(E+ \ii \eta_0)}{\eta_0} \big(N^{20 \delta} \Phi(\eta_1; B)\big)^2
	\end{split}
\end{equation}
by a Schwarz inequality, extending the single $a$ and $b$ summations, using \eqref{eq:extendsum} together with Lemma \ref{lem:Mbounds}, and arguing similarly to \eqref{eq:case1ex1av}. 
All other possible constellations of indices $a,b$ can be handled similarly and are hence left to the reader.

For $k=3$, we finally turn to case (iii), where we estimate
\begin{equation} \label{eq:case3av}
	\begin{split}
		\lambda^3 N^{-9/2} \left| \sum_{a,b} \big((GBG)_{ab}\big)^3 \right| 
		\lesssim  \lambda^2 N^{-3} \eta_1^{-1} \sum_{a,b} \left|(GBG)_{ab} \right|^2 \,  \Phi(\eta_1; B) 
		\lesssim  N^{-2 \delta} \lambda^2 \frac{\rho}{\eta_0} \big(N^{20 \delta} \Phi(\eta_1; B)\big)^3
	\end{split}
\end{equation} 
where in the first step we employed that 
\begin{equation} \label{eq:GBGtriv}
	|(GBG)_{ab}| \le \sqrt{(GBB^* G^*)_{aa}} \sqrt{(GG^*)_{bb}} \lesssim \sqrt{N} \frac{\sqrt{\rho}}{\eta} \sqrt{\frac{\langle \Im M BB^* \rangle}{\rho}} \frac{1}{\lambda \sqrt{\rho}} = N^{3/2} \lambda^{-1} \Phi(\eta;B)
\end{equation}
as follows by using Ward identities, \eqref{eq:ImGImM}, and Lemma \ref{lem:Mbounds}. In the last step in \eqref{eq:case3av} we argue similarly to \eqref{eq:case1ex1av}, \eqref{eq:case1ex2av}, and \eqref{eq:case2av} above.

For the higher order terms with $k \ge 4$, we assume that $n  \in [k]$ factors of $R$ were differentiated and distinguish three cases, (i) $2 \le n \le k-2$, (ii) $n=1$, and (iii) $n \in \{k-1, k\}$. In case (i), we assume, for concreteness,  that the product of resolvent entries contains the term $\big((GBG)_{aa} G_{bb}\big)^2$; other index constellations can be handled similarly. Now, keeping the $\big((GBG)_{aa} G_{bb}\big)^2$-term, we find
\begin{equation} \label{eq:higherorderav}
	\begin{split}
		&\lambda^k N^{-k/2 - n} \lambda^{-n+2} N^{\frac{3(n-2)}{2}} \left(\frac{1}{\lambda^2 \rho}\right)^{k-n-2} \sum_{a,b} \big|(GBG)_{aa} G_{bb}\big|^2 \, \big(\Phi(\eta_1; B)\big)^{n-2} \\
		\lesssim & N^{-2\delta}\lambda^4 N^{-4} \left(\frac{1}{\sqrt{N} \lambda \rho}\right)^{k-n-2} N^4 \frac{\rho (E+ \ii \eta_0)}{\eta_0} \lambda^{-2} \, \big(N^{3 \xi}\Phi(\eta_1; B)\big)^{n} \\
		\lesssim &N^{-2 \delta} \lambda^2  \frac{\rho(E+ \ii \eta_0)}{\eta_0} \, \big(N^{20 \delta}\Phi(\eta_1; B)\big)^{n}
	\end{split}
\end{equation}
where in the beginning we estimated $n-2$ factors of $(GBG)$ trivially as in \eqref{eq:GBGtriv}, and $k-n-2$ factors of $G$ by $1/(\lambda^2 \rho)$, as a consequence of the isotropic law and Lemma \ref{lem:Mbounds}. To go to the second line, we then argued similarly as in \eqref{eq:case2av}. In the ultimate step we simply used that $N \lambda^2 \rho^2 \ge N^\delta$ and $n \le k-2$ by assumption. 

Case (ii) with $n=2$ can be handled analogously with the sole difference that there is only a single $(GBG)$ factor. If this term turns out to be off-diagonal, $(GBG)_{ab}$, then there necessarily also exists an off-diagonal $G$ term. If, in contrast, $(GBG)_{aa}$ is diagonal, then there exists a complementing diagonal $G_{bb}$. Keeping these constellations makes the summation over $a,b$ effective. 

Finally, in case (iii), we can argue similarly to above. However, now, by simple combinatorics, there exist at least two factors of $R$ which are differentiated exactly once and thus producing an off-diagonal $(GBG)_{ab}$. Estimating every other factor trivially (via \eqref{eq:GBGtriv} and as described below \eqref{eq:higherorderav}), the analog of \eqref{eq:higherorderav} becomes 
\begin{equation} \label{eq:higherorderav2}
	\begin{split}
		&\lambda^k N^{-k/2 - n} \lambda^{-n+2} N^{\frac{3(n-2)}{2}} \left(\frac{1}{\lambda^2 \rho}\right)^{k-n} \sum_{a,b} \big|(GBG)_{ab} \big|^2 \, \big(\Phi(\eta_1; B)\big)^{n-2} \\
		\lesssim & N^{-2\delta}\lambda^2 N^{-3} \left(\frac{1}{\sqrt{N} \lambda \rho}\right)^{k-n} N^3 \frac{\rho(E+ \ii \eta_0)}{\eta_0}  \, \big(N^{20 \delta}\Phi(\eta_1; B)\big)^{n} 
		\lesssim N^{-2 \delta} \lambda^2  \frac{\rho(E+ \ii \eta_0)}{\eta_0} \, \big(N^{20 \delta}\Phi(\eta_1; B)\big)^{n}
	\end{split}
\end{equation} 
Similarly to the proof of Proposition \ref{prop:zagisoGron}, by collecting the above bounds and application of various Young inequalities, this concludes the proof of Proposition \ref{prop:zagavGron}. 
\qed

\subsection{Global law: Proof of Proposition \ref{prop:global}} \label{subsec:global} 
In this section, we prove the global law from Proposition~\ref{prop:global}. Similarly to the \emph{zag} step, the proof is conducted iteratively on truncated global domains, 
\begin{equation}
\mathcal{D}_\gamma^{\rm glob} := \left\{ z = E + \ii \eta : z \in \mathcal{D}^{\rm glob}\,, \ \eta \ge N^{-1 + \gamma}\right\} \quad \text{for} \quad \gamma > 0 \,. 
\end{equation}
Once the global law is established in a domain $\mathcal{D}^{\rm glob}_{\gamma_0}$ for some $\gamma_0 > 0$, a simple monotonicity argument similarly to Lemma \ref{lem:monotone} together with the elementary bound from Lemma \ref{lem:Mmax1} shows that 
\begin{equation} \label{eq:globalinput}
\max_{\Vert \bm v \Vert = 1} |G_{\bm u \bm v}(z)| + \max_{\Vert \bm v \Vert = 1} |G_{\bm v \bm u}(z)| \lesssim N^\delta \lambda^{-1} \sqrt{\frac{(\Im M(z))_{\bm u \bm u}}{\rho(z)}} \quad \text{and} \quad (\Im G(z))_{\bm u \bm u} \lesssim N^{2\delta} (\Im M(z))_{\bm u \bm u}
\end{equation}
hold uniformly in $z \in \mathcal{D}_{\gamma_1}^{\rm glob}$ for any $\gamma_1 \ge \gamma_0 - \delta$ uniformly in deterministic vectors $\bm u \in \C^N$. Here, just as in Section \ref{subsec:zagstep}, $\delta > 0$ is an arbitrarily small step size exponent. 

The key step in the iteration is thus going from the resolvent bounds \eqref{eq:globalinput} to a bound on $G-M$. This is formalized in the following bootstrapping lemma.
\begin{lemma}[Bootstrapping: Gap in the values of $G-M$] \label{lem:bootstrap}
Fix a spectral parameter $z \in \mathcal{D}^{\rm glob}_{\gamma_1}$ with some $\gamma_1>0$ such that \eqref{eq:globalinput} holds on $\mathcal{D}_{\gamma_1}^{\rm glob}$. Then we have that 
\begin{equation} \label{eq:bootstrap}
|\langle G(z) - M(z) \rangle| \prec  N^{-10 \delta } \rho(z)  \quad \implies \quad \eqref{eq:globallaw} \quad \text{holds}
\end{equation}
\end{lemma}

Given Lemma \ref{lem:bootstrap}, we can easily conclude the proof of Proposition \ref{prop:global}: By trivial elementary bounds on $G$ and $M$, as indicated in the proof of Proposition \ref{prop:zagstep}, together with Lemma \ref{lem:Mmax1}, we have that \eqref{eq:globalinput} holds when choosing $\gamma_0 = \kappa_0 + 1/\epsilon$. Moreover, by a simple Neumann expansion $|\langle G(z) - M(z) \rangle| \prec  N^{-10 \delta } \rho(z)$ holds for $\eta = \lambda N^{\gamma_0}$ as well. This concludes the proof of Proposition \ref{prop:global} by simple iteration. \qed 

 \bigskip 

The isotropic \eqref{eq:isoGL} and average \eqref{eq:avGL} part of the implication in Lemma \ref{lem:bootstrap} will be proven in Sections \ref{subsec:isoGL} and \ref{subsec:avGL}, respectively. 
\subsubsection{Proof of the isotropic bootstrap} \label{subsec:isoGL}
 We begin with some preliminaries. First, we recall the \emph{stability operator} $\mathcal{B}$, defined in \eqref{eq:stabop}. We control the inverse of the stability operator with the aid of the following lemma, whose proof is given in Appendix \ref{app:techn}. 
\begin{lemma}[Inverse of the stability operator] \label{lem:invstab}
	Let $M_\lambda(z)$ be the solution to the MDE \eqref{eq:MDE} at spectral parameter $z \in \C$. Then it holds that 
	\begin{equation} \label{eq:invstab}
		\big(\mathcal{B}_\lambda(z)\big)^{-1}[X] = X + \frac{\lambda^2 M_\lambda(z)^2}{1 - \lambda^2 \langle M_\lambda(z)^2\rangle} \langle X \rangle \quad \text{for all} \quad X \in \C^{N \times N}\,, 
	\end{equation}
	where the denominator is lower bounded as
	\begin{equation} \label{eq:stabbound}
		\big|1 - \lambda^2 \langle M_\lambda(z)^2 \rangle \big| \gtrsim \big(\big(\lambda^{2} \rho_\lambda(z)/ |\Im z|\big)^{-1} \wedge 1\big)+ \lambda^2 \rho_\lambda(z)^2
	\end{equation}
\end{lemma}
In the following, we will only use the first term in the stability bound \eqref{eq:stabbound} in the form 
\begin{equation} \label{eq:betadef}
	\big|1 - \lambda^2 \langle M_\lambda(z)^2 \rangle \big| \gtrsim \big(\lambda^{2} \rho_\lambda(z)/ |\Im z|\big)^{-1} \wedge 1=: \beta_\lambda(z)  \equiv \beta
\end{equation}
and the fact that, in the global domain, by definition, $\beta \ge c\wedge 1  \gtrsim 1$.

Our proof is based on the following \emph{isotropic bootstrap lemma}. 
\begin{lemma}[Isotropic bootstrap] \label{lem:bootstrapiso}
Under the conditions of Lemma \ref{lem:bootstrap}, we have the following: Assume that, uniformly in $\bm x, \bm y \in \C^N$ it holds that 
\begin{equation} \label{eq:ansatz}
\Psi(z; \bm x, \bm y) := \s^{\rm iso} (z; \bm x, \bm y)^{-1} \big| \big(G(z)-M(z)\big)_{\bm x \bm y} \big| \prec  \psi \,,  \quad  \s^{\rm iso}(z; \bm x, \bm y) := \sqrt{\frac{\rho(z)}{N \eta}} \sqrt{\frac{(\Im M)_{\bm x \bm x} (\Im M)_{\bm y \bm y}}{\rho(z)^2}}
\end{equation}
for some deterministic control quantity $\psi \ge 1$. 
Then we have 
\begin{equation} \label{eq:globaloutcome}
\Psi(z; \bm x, \bm y) \prec 1+ N^{-\delta/2 } \psi 
\end{equation}
again uniformly in $\bm x, \bm y \in \C^N$. 
\end{lemma}
By iteration, Lemma \ref{lem:bootstrapiso} yields that 
\begin{equation*}
\Psi(z; \bm x, \bm y) \prec 1
\end{equation*}
uniformly in $\bm x, \bm y \in \C^N$. This finishes the proof of the isotropic part of Lemma \ref{lem:bootstrap} and we are hence left with proving Lemma \ref{lem:bootstrapiso}. 

\begin{proof}[Proof of Lemma \ref{lem:bootstrapiso}]
By definition of the stability operator \eqref{eq:stabop} and using the MDE \eqref{eq:MDE} together with the definition of the resolvent $G = (H-z)^{-1}$ (henceforth neglecting the spectral parameter, since it remains fixed), we have the identity 
\begin{equation} \label{eq:G-Mident}
	G-M = - \mathcal{B}^{-1}\big[ M \underline{\lambda WG} \big] + \lambda^2 \mathcal{B}^{-1}\big[ M(G-M) \langle G-M \rangle \big]
\end{equation}
where we introduced the \emph{self-energy renormalization} (underline) as
\begin{equation*}
	\underline{\lambda W G} := \lambda WG + \lambda^2 \langle G \rangle G
\end{equation*}
Then, using \eqref{eq:G-Mident}, for 
\begin{equation*}
S = S^{\bm x \bm y} := \big(G-M\big)_{\bm x \bm y}
\end{equation*}
it holds that, for any even (large) $p \in \N$, 
\begin{equation} \label{eq:isoGLstart}
\E \big[|S|^p\big] \le \left| \E \big[ \big(\mathcal{B}^{-1}[M \underline{\lambda WG}]\big)_{\bm x \bm y} \overline{S} |S|^{p-2} \big] \right| + \lambda^2 \left| \E \big[ \big(\mathcal{B}^{-1}[M \langle G-M \rangle (G-M)]\big)_{\bm x \bm y} \overline{S} |S|^{p-2} \big] \right| \,. 
\end{equation}
We begin by estimating the second term on the rhs.~of \eqref{eq:isoGLstart}. Ignoring the factor $\overline{S} |S|^{p-2}$ and the expectation, we estimate 
\begin{equation} \label{eq:quadiso}
\begin{split}
&\lambda^2  \left|\big(\mathcal{B}^{-1}[M \langle G-M \rangle (G-M)]\big)_{\bm x \bm y} \right| \\
\lesssim &\lambda^2 N^{-10 \delta } \rho \left(\left|\big(M (G-M)\big)_{\bm x \bm y} \right| + \left|\beta^{-1} \lambda^2 (M^2)_{\bm x \bm y} \langle M (G-M) \rangle \right|\right) \\
\lesssim & \lambda^2 N^{-8 \delta} \rho \left(\s^{\rm iso}(M^* \bm x, \bm y) + \beta^{-1} \lambda^2 \sqrt{(|M|^2)_{\bm x \bm x} (|M|^2)_{\bm y \bm y}} N^{-1} \sum_a\s^{\rm iso}(M^* \bm e_a, \bm e_a)\right) \psi  \\
\lesssim & N^{-8 \delta} (1 + \beta^{-1}) \s^{\rm iso}(\bm x, \bm y) \, \psi \lesssim N^{-8 \delta}\s^{\rm iso}(\bm x, \bm y) \, \psi  \,,
\end{split}
\end{equation}
all with very high probability. 
To go to the second line, we used Lemma \ref{lem:invstab} and the input bound \eqref{eq:bootstrap}. In the penultimate step, we then plugged in the ansatz \eqref{eq:ansatz} and used a Schwarz inequality. To go to the last line, we used
\begin{equation} \label{eq:Mtricks}
\big(\Im M\big)_{M^* \bm u M^* \bm u} \le \frac{(\Im M)_{\bm u \bm u}}{\lambda^4 \rho^2} \quad \text{and} \quad (|M|^2)_{\bm u \bm u} \lesssim \frac{(\Im M)_{\bm u \bm u}}{\lambda^2 \rho}
\end{equation}
uniformly in $\bm u \in \C^N$, and $N^{-1}\sum_a \s^{\rm iso}(\bm e_a, \bm e_a) \lesssim   \sqrt{\rho/(N \eta)}$. Hence, by a weighted Young inequality for handling the previously ignored $S$-factors, we find that 
\begin{equation} \label{eq:isoglobquadratic}
\lambda^2 \left| \E \big[ \big(\mathcal{B}^{-1}[M \langle G-M \rangle (G-M)]\big)_{\bm x \bm y} \overline{S} |S|^{p-2} \big] \right| \lesssim \big( N^{-5 \delta} \s^{\rm iso}(\bm x, \bm y) \psi \big)^p + N^{-\delta} \E[|S|^p] \,. 
\end{equation}

We now turn to estimating the first term in \eqref{eq:isoGLstart}, which, by means of Lemma \ref{lem:invstab}, we control as
 \begin{equation} \label{eq:isoglobstart}
 	\begin{split}
 	&	\left| \E \left[ \big(\mathcal{B}^{-1}[M \underline{\lambda WG}]\big)_{\bm x \bm y} \overline{S} |S|^{p-2} \right] \right| \\
 		& \qquad \lesssim \left| \E \left[\big(M \underline{\lambda WG}\big)_{\bm x \bm y} \overline{S} |S|^{p-2}\right] \right| + \beta^{-1} \lambda^2 |(M^2)_{\bm x \bm y}|\left|\E \left[\langle M \underline{\lambda WG}\rangle \overline{S} |S|^{p-2}\right] \right|
 	\end{split}
 \end{equation}
and discuss each of the two terms on the rhs.~of \eqref{eq:isoglobstart} separately. 

We begin first the first term and employ a cumulant expansion as in \eqref{eq:cumexbasiciso} to find that 
\begin{equation} \label{eq:isocumex}
\begin{split}
\left| \E \left[\big(M \underline{\lambda WG}\big)_{\bm x \bm y} \overline{S} |S|^{p-2}\right] \right| \lesssim& \frac{\lambda^2}{N} \bigg| \E \bigg[ \sum_{a,b} M_{\bm x a} G_{b \bm y} (\partial_{ab} + \partial_{ba})\{\overline{S} |S|^{p-2}\} \bigg] \bigg| \\
& + \sum_{k=3}^L \frac{\lambda^k}{N^{k/2}} \sum_{l = 0}^k \bigg| \sum_{a,b} \E\bigg[ M_{\bm x a} (\partial_{ab}^l \partial_{ba}^{k-l})\{ G_{b \bm y} \overline{S} |S|^{p-2}\} \bigg] \bigg| + |\Omega_L^{\bm x, \bm y}|
\end{split}
\end{equation}
where the last error term  can be controlled as $|\Omega_L^{\bm x, \bm y}| \lesssim \big(\s^{\rm iso}(\bm x, \bm y)\big)^p$ by choosing $L = \mathcal{O}(1)$ large enough, similarly to \eqref{eq:Ldef}. 

First we estimate the second order terms with $\lambda^2$. We focus on the contribution coming from $\partial_{ab}$; the contribution from $\partial_{ba}$ can be bounded completely analogously. Ignoring the difference between $S$ and $\overline{S}$ and dropping the non-differentiated factor $|S|^{p-2}$, we estimate
\begin{equation} \label{eq:global2ndiso}
\begin{split}
\frac{\lambda^2}{N} \bigg| \sum_{a,b} M_{\bm x a} G_{b \bm y} G_{\bm x a} G_{b \bm y} \bigg| &\lesssim \frac{\lambda^2}{N} \eta^{-3/2} \sqrt{(|M|^2)_{\bm x \bm x}} \sqrt{(\Im G)_{\bm x \bm x}} (\Im G)_{\bm y \bm y} \\
&\lesssim N^{5 \delta} \frac{\lambda^2}{N} \frac{1}{\sqrt{\lambda^2 \rho}} \frac{\rho^2}{\eta^{3/2}} \left(\sqrt{\frac{(\Im M)_{\bm x \bm x} (\Im M)_{\bm y\bm y}}{\rho^2}}\right)^2
\lesssim N^{5 \delta} \big(\s^{\rm iso}(\bm x, \bm y)\big)^2 \,. 
\end{split}
\end{equation}
In the first step, we used a Schwarz inequality and three times a Ward identity. To go to the second line, we used the second input bound from \eqref{eq:globalinput} as well as the second $M$-estimate from \eqref{eq:Mtricks}. In the ultimate step, we then use that $\beta^{-1} \le 1/c \lesssim 1$. 

Next, we move on to terms of third and higher order in the second line of \eqref{eq:isocumex}. As an exemplary term, we consider (again neglecting non-differentiated $S$-factors)
\begin{equation} \label{eq:global3rdiso}
\begin{split}
&\frac{\lambda^3}{N^{3/2}} \bigg| \sum_{a,b} M_{\bm x a} G_{b \bm y} G_{\bm x b} G_{a \bm y} G_{\bm x a} G_{b \bm y}\bigg| \\
\lesssim &N^{2\delta}\lambda N^{-3/2} \eta^{-3/2} \sqrt{(|M|^2)_{\bm x \bm x}} \sqrt{(\Im G)_{\bm x \bm x}} (\Im G)_{\bm y \bm y} \sqrt{\frac{(\Im M)_{\bm x \bm x} (\Im M)_{\bm y \bm y}}{\rho^2}} \\
\lesssim &N^{7 \delta} \lambda N^{-3/2} \frac{1}{\sqrt{\lambda^2 \rho}} \frac{\rho^{2}}{\eta^{3/2}} \left(\sqrt{\frac{(\Im M)_{\bm x \bm x} (\Im M)_{\bm y \bm y}}{\rho^2}}\right)^3\lesssim N^{7 \delta} \big(\s^{\rm iso}(\bm x , \bm y)\big)^3
\end{split}
\end{equation}
estimated similarly to the second order term considered above. Observe that the structure of the terms in 
\eqref{eq:global3rdiso} is identical to \eqref{eq:case3}, the only difference being that $G_{\bm x a}$ in \eqref{eq:case3} is replaced by $M_{\bm x a}$ in \eqref{eq:global3rdiso}. The improved estimate in \eqref{eq:global3rdiso} is then a consequence of using 
\begin{equation*}
\left(\sum_a |M_{\bm x a}|^2\right)^{1/2} \lesssim \lambda^{-1} \sqrt{\frac{(\Im M)_{\bm x \bm x}}{\rho}} \quad \text{instead of} \quad \left(\sum_a |G_{\bm x a}|^2\right)^{1/2} \lesssim N^{2 \delta} \sqrt{\frac{\rho}{\eta}} \sqrt{\frac{(\Im M)_{\bm x \bm x}}{\rho}}
\end{equation*}
which follows by a Ward identity and the second estimate in \eqref{eq:globalinput}. 

Therefore, arguing exactly as in the proof of Proposition \ref{prop:zagisoGron}, using a weighted Young inequality, we conclude that
\begin{equation} \label{eq:isoglobfirstterm}
\left| \E \left[\big(M \underline{\lambda WG}\big)_{\bm x \bm y} \overline{S} |S|^{p-2}\right] \right| \lesssim \big( N^{4 \delta} \s^{\rm iso}(\bm x, \bm y) \big)^p + N^{-\delta } \E [|S|^p]
\end{equation}

We now consider the second term on the rhs.~of \eqref{eq:isoglobstart} and two exemplary terms of order two (as in \eqref{eq:global2ndiso}) and order $k\ge 3$ (as in \eqref{eq:global3rdiso}) obtained via cumulant expansion as in \eqref{eq:isocumex}. As an analog of \eqref{eq:global2ndiso} we obtain (note that the additional $N^{-1}$ emerged from the normalized trace)
\begin{equation*}
\begin{split}
&\beta^{-1}\lambda^4 N^{-2} |(M^2)_{\bm x \bm y}| \, \bigg| \sum_{j,a,b} M_{ja} G_{bj} G_{\bm x a} G_{b \bm y} \bigg| \\
\lesssim & \lambda^4 N^{-2}  \frac{\rho}{\eta} \sqrt{(|M|^2)_{\bm x \bm x} (|M|^2)_{\bm y \bm y}} \sqrt{\sum_{a,b} |(GM)_{ba}|^2} \sqrt{(\Im G)_{\bm x \bm x}/\rho} \sqrt{(\Im G)_{\bm y \bm y}/\rho} \\
\lesssim &N^{4 \delta} \lambda \sqrt{\frac{\rho}{\eta}}  \left(\sqrt{\frac{\rho}{N \eta}}\sqrt{\frac{(\Im M)_{\bm x \bm x} (\Im M)_{\bm y \bm y}}{\rho^2}}\right)^2 \lesssim N^{4 \delta}  \big(\s^{\rm iso}(\bm x, \bm y)\big)^2
\end{split}
\end{equation*}
where we used that 
\begin{equation*}
\sqrt{\sum_{a,b} |(GM)_{ab}|^2} = \sqrt{N} \langle GMM^*G^*\rangle^{1/2} \le N \langle GG^* \rangle^{1/2} \langle MM^* \rangle^{1/2} \lesssim  N^{1+2 \delta} \lambda^{-1} \sqrt{\frac{\rho}{\eta}} \,. 
\end{equation*}
and the fact that $\beta = (\lambda^2 \rho/\eta)^{-1} \ge c \gtrsim 1$ in $\mathcal{D}^{\rm glob}$. 
Next, as an analog of \eqref{eq:global3rdiso}, we obtain
\begin{equation*}
\begin{split}
&\beta^{-1}\lambda^5 N^{-5/2} |(M^2)_{\bm x \bm y}| \, \bigg| \sum_{j,a,b} M_{ja} G_{bj} G_{\bm x a} G_{b \bm y} G_{\bm x a} G_{b \bm y}\bigg| \\
\lesssim &N^{4 \delta} \lambda^3 N^{-5/2} \frac{\rho}{\eta} \sqrt{(|M|^2)_{\bm x \bm x} (|M|^2)_{\bm y \bm y}} \sqrt{\sum_{a,b} |(GM)_{ba}|^2} \left(\sqrt{\frac{(\Im M)_{\bm x \bm x} (\Im M)_{\bm y \bm y}}{\rho^2}}\right)^2 \\
\lesssim &N^{6 \delta}  \left( \sqrt{\frac{\rho}{N \eta}} \sqrt{\frac{(\Im M)_{\bm x \bm x} (\Im M)_{\bm y \bm y}}{\rho^2}}\right)^3 \lesssim N^{6 \delta}  \big(\s^{\rm iso}(\bm x, \bm y)\big)^3 \,. 
\end{split}
\end{equation*}
Other terms of order three or higher can be treated analogously. Hence, by a weighted Young inequality, just as in \eqref{eq:isoglobfirstterm}, we obtain
\begin{equation*}
\beta^{-1} \lambda^2 |(M^2)_{\bm x \bm y}|\left|\E \left[\langle M \underline{\lambda WG}\rangle \overline{S} |S|^{p-2}\right] \right| \lesssim \big(N^{ 4 \delta  } \s^{\rm iso}(\bm x, \bm y)\big)^p + N^{-\delta} \E [|S|^p]
\end{equation*}
and thus, together with \eqref{eq:isoglobquadratic} and \eqref{eq:isoglobfirstterm}
\begin{equation*}
\E [|\Psi(\bm x, \bm y)|^p] \lesssim  N^{ 4 \delta p}  + N^{-\delta p} \psi^p 
\end{equation*}
 Because $p$ and $\delta$ are arbitrary, this immediately proves \eqref{eq:globaloutcome} (by standard relations between and high moment bounds and properties of stochastic domination) and we have thus concluded the proof of Lemma~\ref{lem:bootstrapiso}. 
\end{proof}

\subsubsection{Proof of the average average bootstrap and completion of the proof of Lemma \ref{lem:bootstrap}} \label{subsec:avGL}
Our proof of the average part of Lemma \ref{lem:bootstrap} is based on the following \emph{average bootstrap lemma}, relying, in particular on the already established isotropic law from Section \ref{subsec:isoGL}. 
\begin{lemma}[Average bootstrap] \label{lem:bootstrapav}
	Under the conditions of Lemma \ref{lem:bootstrap}, we have the following: Assume that, uniformly in $B \in \C^{N\times N}$ it holds that 
	\begin{equation} \label{eq:ansatzav}
	\Phi(z; B) := 	\s^{\rm av}(z; B)^{-1} \big| \big\langle (G(z)-M(z)) B\big\rangle \big| \prec \phi \quad \text{with} \quad \s^{\rm av}(z; B) := \frac{1}{N \eta} \sqrt{\frac{\langle \Im M(z) BB^* \rangle}{\rho(z)}}
	\end{equation}
	for some deterministic control quantity $\phi \ge 1$,  and the isotropic law \eqref{eq:isoGL} holds on $\mathcal{D}_{\gamma_1}^{\rm glob}$. 
	Then
	\begin{equation} \label{eq:globaloutcomeav}
		\Phi(z; B) \prec  1 + N^{-\delta/2} \phi \,, 
	\end{equation}
	again uniformly in $B \in \C^{N \times N}$. 
\end{lemma}
By iteration, Lemma \ref{lem:bootstrapav} yields that 
\begin{equation*}
	\Phi(z; B)\prec 1
\end{equation*}
uniformly in $B\in \C^{N \times N}$. This finishes the proof of the average part of Lemma \ref{lem:bootstrap} and hence, together with Section \ref{subsec:isoGL}, the entire proof of Lemma \ref{lem:bootstrap}.  \qed

\bigskip 

We are hence left with proving Lemma \ref{lem:bootstrapav}. 

\begin{proof}[Proof of Lemma \ref{lem:bootstrapav}]
The argument is similar to that given in Lemma \ref{lem:bootstrapiso}. Analogously to \eqref{eq:isoGLstart}, again dropping the spectral parameter as an argument, since it remains fixed, we find that with
\begin{equation*}
 R = R^B := \big\langle (G-M)B\big\rangle
\end{equation*}
it holds that
\begin{equation} \label{eq:globalavstart}
\E [|R|^p] \le \left| \E \big[ \langle M \underline{\lambda WG} \widetilde{B}\rangle \overline{R} |R|^{p-2} \big] \right| + \lambda^2  \left| \E \big[ \langle G-M\rangle \langle M(G-M) \widetilde{B} \rangle \overline{R} |R|^{p-2} \big] \right|
\end{equation}
for any even (large) $p \in \N$. Here, we introduced the shorthand notation
\begin{equation} \label{eq:Btilde}
\widetilde{B} := \big( (\mathcal{B}^{-1})^*[B^*]\big)^* = B + \lambda^2 \frac{\langle MBM \rangle}{1 - \lambda^2 \langle M^2\rangle} \,. 
\end{equation}

We begin with discussing the second term on the rhs.~of \eqref{eq:globalavstart}. Analogously to \eqref{eq:quadiso}, we obtain
\begin{equation*}
\begin{split}
&\lambda^2 \beta^{-1} \left|\langle G-M\rangle \langle M(G-M) \widetilde{B} \rangle \right|  \\
& \qquad \lesssim \lambda^2 N^{-10 \delta} \rho \left( \s^{\rm av}(BM) + \beta^{-1} \lambda^2 |\langle MBM\rangle| \s^{\rm av}(M)\right) \phi 
\lesssim N^{-10 \delta} \s^{\rm av}(B) \phi \,. 
\end{split}
\end{equation*}
Here, we used that 
\begin{equation} \label{eq:tildetonotilde}
\s^{\rm av}(BM) \lesssim \frac{1}{\lambda^2 \rho} \s^{\rm av}(B) \quad \text{and} \quad |\langle MBM\rangle| \s^{\rm av}(M) \lesssim \frac{1}{\lambda^4 \rho} \s^{\rm av}(B)
\end{equation}
as follows from $\Vert MM^* \Vert \lesssim 1/(\lambda^4 \rho^2)$ and elementary trace inequalities, and the fact that $\beta \gtrsim 1$. Then, by a weighted Young inequality, analogously to \eqref{eq:isoglobquadratic}, we obtain
\begin{equation} \label{eq:quadav}
\lambda^2  \left| \E \big[ \langle G-M\rangle \langle M(G-M) \widetilde{B} \rangle \overline{R} |R|^{p-2} \big] \right| \lesssim \big(N^{ 4 \delta}(1 + \beta^{-1}) \s^{\rm av}(B) \phi\big)^p + N^{-\delta p } \E [|R|^p]
\end{equation}

Next, we turn to the first term on the rhs.~of \eqref{eq:globalavstart}. By a cumulant expansion, just as in \eqref{eq:isocumex}, we find that
\begin{equation} \label{eq:avcumex}
\begin{split}
\left| \E \big[ \langle M \underline{\lambda WG} \widetilde{B}\rangle \overline{R} |R|^{p-2} \big] \right| \lesssim& \frac{\lambda^2}{N^2} \bigg| \E \bigg[ \sum_{j,a,b} M_{j a} (G\widetilde{B})_{b j} (\partial_{ab} + \partial_{ba})\{\overline{R} |R|^{p-2}\} \bigg] \bigg| \\
& + \sum_{k=3}^L \frac{\lambda^k}{N^{k/2+1}} \sum_{l = 0}^k \bigg| \sum_{j,a,b} \E\bigg[ M_{j a} (\partial_{ab}^l \partial_{ba}^{k-l})\{ (G\widetilde{B})_{b j} \overline{R} |R|^{p-2}\} \bigg] \bigg| + |\Omega_L^{B}|
\end{split}
\end{equation}
where the last error term can be controlled as $|\Omega^B_L| \lesssim \big(\s^{\rm av}(B)\big)^p$ for $L = \mathcal{O}(1)$ large enough. 

First, we bound the second order term in the first line of \eqref{eq:avcumex} and focus on the contribution coming from $\partial_{ab}$; the contribution coming from $\partial_{ba}$ can be bounded completely analogously. Similarly to \eqref{eq:global2ndiso}, in particular ignoring the difference between $\overline{R}$ and $R$ and dropping the non-differentiated factor $|R|^{p-2}$, we estimate (note the additional $N^{-1}$ originates from the \emph{average} trace)
\begin{equation} \label{eq:2ndordav}
	\begin{split}
\lambda^2 N^{-3} \bigg| \sum_{a,b} (G\widetilde{B}M)_{ab} (GBG)_{ab}\bigg| &\lesssim \lambda^2 N^{-2} \langle G\widetilde{B} MM^* \widetilde{B}^* G^*\rangle^{1/2} \langle GBGG^* B^* G^* \rangle^{1/2} \\
&\lesssim  \s^{\rm av}(\widetilde{B}) \s^{\rm av}(B) \lesssim  (1 + \beta^{-1}) \big(\s^{\rm av}(B)\big)^2 \lesssim \big(\s^{\rm av}(B)\big)^2 \,, 
	\end{split}
\end{equation}
where to go to the second line, we used that $\Vert MM^* \Vert \lesssim 1/(\lambda^4\rho^2)$, $\Vert GG^* \Vert \le \eta^{-2}$, two Ward identities, and
\begin{equation*}
\langle \Im G BB^* \rangle \lesssim  \langle \Im M BB^*\rangle
\end{equation*}
as follows from the already proven isotropic law and spectral decomposition for $BB^*$; see also \eqref{eq:ImGImM} below. In the ultimate step, we then used that $\s^{\rm av}(\widetilde{B}) \lesssim (1 + \beta^{-1}) \s^{\rm av}(B)$, which can be obtained similarly to \eqref{eq:tildetonotilde}. 

The higher order terms with $k \ge 3$ in the second line of \eqref{eq:avcumex} then have the same structure as the terms arising in \eqref{eq:case1ex1av}, \eqref{eq:case1ex2av}, \eqref{eq:case2av}, \eqref{eq:case3av}, \eqref{eq:higherorderav}, or \eqref{eq:higherorderav2}, with the exception that one $GBG$-factor becomes a $G\widetilde{B}M$; cf.~\eqref{eq:2ndordav}. We can hence follow these estimates, but replace the bound
\begin{equation*}
\sqrt{\sum_{a,b}|(GBG)_{ab}|^2} + \max_{a,b} |(GBG)_{ab}| \lesssim N^{3/2} \sqrt{\frac{\rho}{\eta}} \s^{\rm av}(B) 
\end{equation*}
by 
\begin{equation*}
	\sqrt{\sum_{a,b}|(G\widetilde{B}M)_{ab}|^2} + \max_{a,b} |(G\widetilde{B}M)_{ab}| \lesssim \frac{\eta}{\lambda^2 \rho + \eta} (1 + \beta^{-1})N^{3/2} \sqrt{\frac{\rho}{\eta}} \s^{\rm av}(B)  \,. 
\end{equation*}
Hence, following the estimates in the proof of Proposition \ref{prop:zagavGron}, we find that, by means of a weighted Young inequality, 
\begin{equation*}
\left| \E \big[ \langle M \underline{\lambda WG} \widetilde{B}\rangle \overline{R} |R|^{p-2} \big] \right| \lesssim \big( N^{ 4 \delta }  \s^{\rm av}(B)\big)^p  + N^{-\delta  } \E[|R|^p] \,. 
\end{equation*}

Therefore, combining with \eqref{eq:globalavstart} and \eqref{eq:quadav}, we obtain
\begin{equation*}
\E [|\Phi(B)|^p] \lesssim N^{ 4 \delta p }  + N^{-\delta p } \phi^p \,. 
\end{equation*}
Finally, since $p$ and $\delta$ are arbitrary, we immediately infer \eqref{eq:globaloutcomeav}, similarly to the conclusion of the proof of Lemma \ref{lem:bootstrapiso}, and have thus concluded the proof of Lemma \ref{lem:bootstrapav}. 
\end{proof}

\section{Two resolvent laws: Proof of Propositions \ref{prop:global2}, \ref{prop:zigstep2}, and \ref{prop:zagstep2}} \label{sec:2GLaw}
In this section, we collect the proofs of Propositions \ref{prop:global2}, \ref{prop:zigstep2}, and \ref{prop:zagstep2} regarding the two resolvent laws in Sections \ref{subsec:zigstep2}, \ref{subsec:zagstep2}, and \ref{subsec:global2}, respectively. We recall that throughout the proof we focus on the regime with $\theta(z_1, z_2) = 1$ and consider $A$'s that are regular w.r.t.~the target spectral parameters (see the discussion around \eqref{eq:s2time}). The modifications needed for the complementary case $\theta(z_1, z_2) = 0$ are outlined in Appendix \ref{subsec:regime2}. Moreover, throughout the whole section, we suppose that $\Im z_i > 0$ for simplicity of the presentation. 

\subsection{Characteristic flow: Proof of Proposition \ref{prop:zigstep2}} \label{subsec:zigstep2} 
 In Section~\ref{subsec:zigstep} we studied the evolution of a single resolvent along the characteristic flow. The main goal of this section is to study the evolution of products of two resolvents and deterministic regular observables $A$ (see Definition~\ref{def:regobs}). 
 We conduct the proof in the complex Hermitian case, the obvious modifications in the real symmetric case\footnote{For a detailed treatment, we refer to \cite[Section 4]{cipolloni2023eigenstate}.} are left to the reader. Also, note that it suffices to prove the statement only for fixed $z$'s and $t$'s, since uniformity follows by a standard simple grid argument. 
 
 In the following, we often drop the index $k$ as it remains fixed and introduce $t_{\rm init} := t_{k-1}$ and $t_{\rm fin} := t_{k}$. For $t \in [t_{\rm init}, t_{\rm final}]$, we denote $G_{i,t}:=\big(\mathfrak{F}_{\rm zig}^{t - t_{\rm init}}[H_{k-1}] - z_{i,t}\big)^{-1}$ for $z_{i,t} := \varphi_{t, t_{\rm final}}(z_i)$ with $\varphi$ defined in \eqref{eq:flowmap}, and $M_{i,t}:= M_{\lambda, t}(z_{i,t})$. 
Moreover, for the rest of this section we always use the simpler notation $A=\mathring{A}$. 
We are specifically interest in the evolution of the quantities $\langle G_{1,t}^{(*)}AG_{2,t}^{(*)}\rangle$ and $\langle \Im G_{1,t}A\Im G_{2,t}A^*\rangle$. 
For the sake of brevity, here we only prove a local law for the former and a bound for the latter. However, inspecting the proof is clear that with analogous estimates one can obtain a local law for $\langle \Im G_{1,t}A\Im G_{2,t}A^*\rangle$ as well. In fact, the only slightly different estimate would be the analog of \eqref{eq:diffllaw}; all the rest would remain the same (see \cite{cipolloni2023eigenstate} for the Wigner case). Additionally, below we only consider the evolution of $\langle G_{1,t}AG_{2,t}\rangle$, the proof of the local laws for all other combinations of stars is completely analogous.
For these objects we have the flows (see, e.g., \cite[Section 4]{cipolloni2023eigenstate}) 
\small
\begin{equation}
	\begin{split}
		\label{eq:GAG}
		\dd \langle G_{1,t}AG_{2,t}-M_{12,t}^A\rangle&=\frac{\lambda}{\sqrt{N}}\sum_{ab=1}^N \partial_{ab}\langle G_{1,t}AG_{2,t}\rangle (\dd B_t)_{ab}\\
		&\quad+\langle G_{1,t}AG_{2,t}-M_{12,t}^A\rangle (1 + \lambda^2\langle M_{12,t}^I\rangle) \dd t+\lambda^2\langle M_{12,t}^A\rangle \langle G_{1,t}G_{2,t}-M_{12,t}^I\rangle \dd t \\
		&\quad+\lambda^2\langle G_{1,t}AG_{2,t}-M_{12,t}^A\rangle \langle G_{1,t}G_{2,t}-M_{12,t}^I\rangle \dd t  \\
		&\quad+\lambda^2\langle G_{1,t}-M_{1,t}\rangle \langle G_{1,t}^2AG_{2,t}\rangle \dd t+\lambda^2\langle G_{2,t}-M_{2,t}\rangle \langle G_{1,t}AG_{2,t}^2\rangle \dd t
	\end{split}
\end{equation}
\normalsize
and
\small 
\begin{equation}
	\begin{split}
		\label{eq:imGAimGA}
		\dd \langle \Im G_{1,t}A &\Im G_{2,t}A^*\rangle \\
		&=\lambda \sum_{a,b=1}^N\partial_{ab} \langle \Im G_{1,t}A \Im G_{2,t} A^*\rangle\frac{(\dd B_t)_{ab}}{\sqrt{N}}+\langle \Im G_{1,t} A\Im G_{2,t} A^*\rangle \dd t \\
		&\quad+\lambda^2\left(\frac{\langle \Im G_{1,t} -\Im M_{1,t}\rangle}{\eta_{1,t}}+\frac{\langle \Im G_{2,t} -\Im M_{2,t}\rangle}{\eta_{2,t}}\right)\langle \Im G_{1,t} A\Im G_{2,t} A^*\rangle \dd t \\
		&\quad +\lambda^2\langle G_{2,t}^*A^*G_{1,t}\rangle \langle \Im G_{1,t}A\Im G_{2,t}\rangle \dd t+\lambda^2\langle G_{1,t}^* A G_{2,t} \rangle\langle \Im G_{2,t} A^* \Im G_{1,t}\rangle\dd t\\
		&\quad+\lambda^2\langle \Im G_{1,t} A G_{2,t}\rangle\langle\Im G_{2,t}A^*G_{1,t}\rangle \dd t+\lambda^2\langle G_{2,t}^*A^*\Im G_{1,t}\rangle\langle G_{1,t}^*A\Im G_{2,t}\rangle\dd t \\
		&\quad+\lambda^2\langle G_{1,t}-M_{1,t}\rangle \langle \Im G_{1,t} A\Im G_{2,t} A^* G_{1,t}\rangle\dd t +\lambda^2\langle G_{2,t}-M_{2,t}\rangle \langle \Im G_{2,t} A^*\Im G_{1,t} A G_{2,t}\rangle \dd t\\
		&\quad+\lambda^2\langle G_{1,t}^*-M_{1,t}^*\rangle \langle \Im G_{1,t} A\Im G_{2,t} A^* G_{1,t}^*\rangle \dd t+\lambda^2\langle G_{2,t}^*-M_{2,t}^*\rangle \langle \Im G_{2,t} A^*\Im G_{1,t} A G_{2,t}^*\rangle \dd t.
	\end{split}
\end{equation}
\normalsize

The terms in the rhs. of the two flows \eqref{eq:GAG}--\eqref{eq:imGAimGA} will be estimated together, in fact the bound on $\langle \Im G_{1,t}A\Im G_{2,t}A^*\rangle$ will be used to estimate the rhs. of \eqref{eq:GAG} too. For this reason we introduce the stopping time
\begin{equation}
\label{eq:stoptime}
\begin{split}
\tau:=\inf\bigg\{t \in  [\tinit, \tfin] : &\big|\langle G_{1,t}AG_{2,t}-M_{12,t}^A\rangle\big|\ge N^{2\xi}\sqrt{\frac{\rho_{1,t}\rho_{2,t}}{N \eta_{1,t}\eta_{2,t}}}\s_2(t), \\
  &\langle \Im G_{1,t}A\Im G_{2,t}A^*\rangle\ge N^\xi \rho_{1,t}\rho_{2,t}\s_2(t)^2 \bigg\}
\end{split}
\end{equation}
for $N\ell_t\ge N^\epsilon$ and $\xi\le \epsilon/10$, where we used the notation $\s_2(t)$ from \eqref{eq:s2time}. 

We now start with the estimate of the terms in the rhs. of \eqref{eq:GAG}. In the remainder of the section we always assume $t \in [\tinit, \tfin \wedge \tau]$, i.e., in particular, neglect the $\wedge \tau$ in the writing henceforth. 
We start by estimating the quadratic variation of the stochastic term in \eqref{eq:GAG}:
\begin{equation*}
	\begin{split}
		&\frac{\lambda^2}{N^2\eta_{1,s}^2}\langle \Im G_{1,s}AG_{2,s}\Im G_{1,s}G_{2,s}^*A^*\rangle \dd t+  \frac{\lambda^2}{N^2\eta_{2,s}^2}\langle \Im G_{2,s}A^*G_{1,s}^*\Im G_{2,s}G_{1,s}A\rangle \dd t \\
		&\qquad\quad \lesssim \frac{\lambda^2}{N^2\eta_{1,s}\eta_{2,s}}\left(\frac{1}{\eta_{1,s}^2}+\frac{1}{\eta_{2,s}^2}\right)\langle \Im G_{1,s}A\Im G_{2,s}A^*\rangle \dd t.
	\end{split}
\end{equation*}
By the BDG inequality (see Appendix B.6, Eq.~(18) in \cite{shorack2009empirical}) we thus have 
\begin{equation} \label{eq:BDG}
	\begin{split}
		\sup_{u \in [\tinit, t]} & \left|\sum_{a,b=1}^N\int_{\tinit}^{u}\partial_{ab} \langle  G_{1,s}A  G_{2,s} \rangle\frac{\lambda (\dd B_s)_{ab}}{\sqrt{N}}\right| \\[2mm]
		&\le N^\xi \left(\int_{\tinit}^{t} \frac{\lambda^2}{N^2\eta_{1,s}\eta_{2,u}}\left(\frac{1}{\eta_{1,s}^2}+\frac{1}{\eta_{2,s}^2}\right)\langle \Im G_{1,s}A\Im G_{2,s}A^*\rangle\,\dd s\right)^{1/2} \\[2mm]
		&\lesssim \frac{N^{3\xi/2}\mathfrak{s}_2(t)^{1/2}\sqrt{\rho_{1,t}\rho_{2,t}}}{N\sqrt{\eta_{1,t}\eta_{2,t}\ell_t}} \lesssim 
		\frac{N^{3\xi/2}\mathfrak{s}_2(t)^{1/2}\sqrt{\rho_{1,t}\rho_{2,t}}}{\sqrt{N \eta_{1,t}\eta_{2,t}}}	\,, 
	\end{split}
\end{equation}
where in the last step we used that $N \ell_t \ge 1$. We point out that here we also used that $\s_2(s)\lesssim s_2(t)$ for $s\le t$. This immediately follows from the definition of $\s_2(t)$ in \eqref{eq:defs2} and the fact that for the denominator in $\s_2(t)$ we have
\[
|1- \ee^{s-1} \lambda^2 \Re \langle M_1M_2\rangle| \gtrsim  |1- \ee^{t-1} \lambda^2 \Re \langle M_1M_2\rangle|, 
\]
which easily follows by distinguishing the cases $\Re \langle M_1 M_2 \rangle > 0$ (using additionally that, by definition of the terminal time, $\ee^t \lambda^2 \Re \langle M_1 M_2 \rangle < 1$) and $\Re \langle M_1 M_2 \rangle \le 0$. 
Throughout the remainder of the proof, we will frequently use $\s_2(s)\lesssim s_2(t)$, for $s\le t$, even if not stated explicitly.
  
Similarly, we estimate
\begin{equation*}
	\begin{split}
		\lambda^2\int_{\tinit}^t\langle G_{1,s}-M_{1,s}\rangle \langle G_{1,s}^2AG_{2,s}\rangle \dd s&\le \int_{\tinit}^t \frac{\lambda^2 N^\xi}{N\eta_{1,s}^2\eta_{2,s}^{1/2}}\langle \Im G_{1,s}\rangle^{1/2}\langle \Im G_{1,s}A\Im G_{2,s}A^*\rangle^{1/2}\,\dd s \\
		&\lesssim \int_{\tinit}^t \frac{\lambda^2 N^{3\xi/2} \rho_{1,s}\rho_{2,s}^{1/2} \mathfrak{s}_2(s)}{N\eta_{1,s}^2\eta_{2,s}^{1/2}} \, \dd s \\
		&\lesssim \frac{N^{3\xi/2} \rho_{2,t}^{1/2}\mathfrak{s}_2(t)}{N\eta_{1,t}\eta_{2,t}^{1/2}}\le  \frac{N^{3\xi/2}\mathfrak{s}_2(t)\sqrt{\rho_{1,t}\rho_{2,t}}}{N\sqrt{\eta_{1,t}\eta_{2,t}\ell_t}}\,, 
	\end{split}
\end{equation*}
additionally using a Schwarz inequality and a single resolvent bound $\langle \Im G_{1,s} \rangle \lesssim \rho_{1,s}$ as a consequence from the already established single resolvent local law \eqref{eq:avLL}. 
Using \eqref{eq:Mbounds} to bound the deterministic terms, the last term in the second line of \eqref{eq:GAG} is estimated by 
\begin{equation*}
	\begin{split}
		&\lambda^2\int_{\tinit}^t \sqrt{\frac{\rho_{2,s}}{\eta_{1,s}}\vee \frac{\rho_{1,s}}{\eta_{2,s}}}\sqrt{\rho_{1,s}\rho_{2,s}}\mathfrak{s}_2(s)\cdot \frac{N^\xi}{N\eta_{1,s}\eta_{2,s}}\, \dd s\\
		&\lesssim \frac{N^\xi}{\sqrt{N\eta_{1,t}\eta_{2,t}}}\sqrt{\rho_{1,t}\rho_{2,t}}\s(t)\int_{\tinit}^t \frac{\lambda^2}{\sqrt{N\eta_{1,s}\eta_{2,s}}} \left(\frac{\rho_{2,s}^{1/2}}{\eta_{1,s}^{1/2}}+\frac{\rho_{1,s}^{1/2}}{\eta_{2,s}^{1/2}}\right)\,\dd s \\
		&\lesssim  \frac{N^\xi}{\sqrt{N\eta_{1,t}\eta_{2,t}}}\sqrt{\rho_{1,t}\rho_{2,t}}\s(t)\int_{\tinit}^t\frac{\lambda^2}{\sqrt{N}}\left(\frac{1}{\eta_{1,s}^{3/2}}+\frac{\rho_{2,s}}{\eta_{1,s}^{1/2}\eta_{2,s}}+\frac{\rho_{1,s}}{\eta_{1,s}\eta_{2,s}^{1/2}}+\frac{1}{\eta_{2,s}^{3/2}}\right)\,\dd s \\
		&\le \frac{N^\xi\mathfrak{s}_2(t)\sqrt{\rho_{1,t}\rho_{2,t}}}{N\sqrt{\eta_{1,t}\eta_{2,t}\ell_t}} \lesssim \frac{N^\xi\mathfrak{s}_2(t) \sqrt{\rho_{1,t}\rho_{2,t}}}{\sqrt{N \eta_{1,t}\eta_{2,t}}}
	\end{split}
\end{equation*}
We point out that here we used
\begin{equation}
\label{eq:boundcauch}
\big|\langle G_{1,t}G_{2,t}-M_{12,t}^I\rangle\big|\le \frac{N^\xi}{N\eta_{1,t}\eta_{2,t}},
\end{equation}
which follows by Cauchy integral formula and the single resolvent local law\footnote{More precisely, using Cauchy integral formula we can write (see, e.g., \cite[Eq. (3.14)]{cipolloni2022optimal})
\[
\langle G_{1,t}G_{2,t}-M_{12,t}^I\rangle=\int_\R \langle \Im G_t(x+\ii\zeta)-\Im M_t(x+\ii\zeta)\rangle\prod_{i=1}^2\frac{1}{x-z_{i,t}+\ii\mathrm{sgn}(z_{i,t})\zeta}\, \dd x,
\]
for any positive $\zeta$ with $\zeta<\min_i |\Im z_{i,t}|$. Then, choosing $\zeta:=\frac{1}{2}\min_i |\Im z_{i,t}|$ and using the single resolvent local law \eqref{eq:avLL} (we point out that \eqref{eq:avLL} can easily be extended to any $\eta>0$ by a simple argument as in \cite[Appendix A]{cipolloni2021edge}) in this integral representation we immediately obtain \eqref{eq:boundcauch}.}.  (Alternatively this can be proven using the flow as well.)
Proceeding similarly we estimate the term in the third line of \eqref{eq:GAG} by
\begin{equation*}
	\lambda^2\int_{\tinit}^t\frac{N^{2\xi}}{\sqrt{N \eta_{1,s}\eta_{2,s}}}\sqrt{\rho_{1,s}\rho_{2,s}}\s_2(s)\cdot \frac{N^\xi}{N\eta_{1,s}\eta_{2,s}}\,\dd s\lesssim \frac{N^{3\xi}}{N\ell_t}\cdot \frac{N^{2\xi}}{\sqrt{N \eta_{1,t}\eta_{2,t}}}\sqrt{\rho_{1,t}\rho_{2,t}}\s_2(t).
\end{equation*}

Next, we denote $X_t:=\langle G_{1,t}AG_{2,t}-M_{12,t}^A\rangle$ and bound
\begin{equation*}
	\begin{split}
	|X_t|&\le |X_{\tinit}|+\int_{\tinit}^t \frac{1}{|1-\lambda^2\langle M_{1,s}M_{2,s}\rangle|}|X_s|\,\dd s+\frac{N^{5\xi/4}\mathfrak{s}_2(t)\sqrt{\rho_{1,t}\rho_{2,t}}}{\sqrt{N\eta_{1,t}\eta_{2,t}}} \\
	&\le \int_{\tinit}^t \frac{1}{|1-\lambda^2\langle M_{1,s}M_{2,s}\rangle|}|X_s|\,\dd s+\frac{N^{3\xi/2}\mathfrak{s}_2(t)\sqrt{\rho_{1,t}\rho_{2,t}}}{\sqrt{N\eta_{1,t}\eta_{2,t}}} \,,
	\end{split} 
\end{equation*}
where we used the explicit formula (following from \eqref{eq:m12} with $A \to I$)
\begin{equation*}
1 + \lambda^2 \langle M_{12,s}^I \rangle = \frac{1}{1 - \lambda^2 \langle M_{1,s} M_{2,s} \rangle}
\end{equation*}
for the first term in the second line of \eqref{eq:GAG} and employed 
\begin{equation*}
\big|X_{\tinit}\big|\lesssim \frac{N^\xi\sqrt{\rho_{1,{\tinit}}\rho_{2,{\tinit}}}}{\sqrt{N \eta_{1,{\tinit}}\eta_{2,{\tinit}}}} \mathfrak{s}_2({\tinit}) \lesssim  \frac{N^\xi\mathfrak{s}_2(t)\sqrt{\rho_{1,t}\rho_{2,t}}}{\sqrt{N\eta_{1,t}\eta_{2,t}}}
\end{equation*}
as follows by elementary monotonicity properties. 
Then, using a Gronwall inequality we obtain
\begin{equation}
	\label{eq:calerr}
	\begin{split}
	|X_t|&\lesssim \int_{\tinit}^{t}      \frac{N^{3\xi/2}\sqrt{\rho_{1,s}\rho_{2,s}}}{\sqrt{N \eta_{1,s}\eta_{2,s}}} \mathfrak{s}_2(s)   \sqrt{\frac{\eta_{1,{s}}\eta_{2,{s}}}{\eta_{1,t}\eta_{2,t}}} \left(1 + \lambda^2 \left(\frac{\rho_{1,s} }{\eta_{1,s} } + \frac{\rho_{2,s} }{\eta_{2,s} }\right)\right) \, \dd s+ \frac{N^{3\xi/2}\sqrt{\rho_{1,t}\rho_{2,t}}}{\sqrt{N \eta_{1,t}\eta_{2,t}}} \mathfrak{s}_2(t) \\
	&\lesssim \frac{N^{7\xi/4}\sqrt{\rho_{1,t}\rho_{2,t}}}{\sqrt{N \eta_{1,t}\eta_{2,t}}} \mathfrak{s}_2(t)
	\end{split}
\end{equation}
	where  we used that
	\[
\left| \frac{1}{1-\lambda^2\langle M_{1,s}M_{2,s}\rangle}\right| \lesssim  \left(1 + \lambda^2 \left(\frac{\rho_{1,s} }{\eta_{1,s} } + \frac{\rho_{2,s} }{\eta_{2,s} }\right)\right) \quad \text{and} \quad \exp\left( \int_{s}^{t}	\left| \frac{1}{1-\lambda^2\langle M_{1,u}M_{2,u}\rangle}\right|  \dd u\right) \lesssim \sqrt{\frac{\eta_{1,{s}}\eta_{2,{s}}}{\eta_{1,t}\eta_{2,t}}}\,, 
	\]
as follow by explicit computations similarly to \cite[Lemma 5.2, Eq.~(5.13a)]{cipolloni2024eigenvector}, see also \eqref{eq:beta12bound} below. 

	We now turn to the estimate of the rhs.~of \eqref{eq:imGAimGA}. The last term in the second line and the terms in the third line of \eqref{eq:imGAimGA} can be easily seen to be negligible.  Throughout we assume that $\langle G_{1,t}^{(*)}AG_{2,t}^{(*)}\rangle$ can be estimated by the corresponding $M$-term as a consequence of \eqref{eq:Mbounds} and the local law we just proved above. We first focus on the terms in the fourth and fifth line of \eqref{eq:imGAimGA}. For these terms, using \eqref{eq:Mbounds},  we estimate
\begin{equation}
	\label{eq:diffllaw}
		\int_{\tinit}^t \mathrm{fourth \,\, and \,\, fifth \,\, lines \,\, \eqref{eq:imGAimGA}} \, \dd s\le \lambda^2\int_{\tinit}^t \left(\frac{\rho_{1,s}\rho_{2,s}}{\eta_{1,s}\eta_{2,s}}\right)^{1/2}\mathfrak{s}_2(s)^2\rho_{1,s}\rho_{2,s} \, \dd s\lesssim \log N \mathfrak{s}_2(t)^2\rho_{1,t}\rho_{2,t}.
	\end{equation}
	We point out that here we used that while the bounds in \eqref{eq:Mbounds} are asymmetric in the indices $1$ and $2$, the terms appearing in the third and fourth line of \eqref{eq:imGAimGA} are so that the product of the bounds is symmetric. Next, we estimate
	\begin{equation*}
		\int_{\tinit}^t \mathrm{sixth \,\, and \,\, seventh \,\, lines \,\, \eqref{eq:imGAimGA}} \,\,\dd s\le \int_{\tinit}^t\frac{\lambda^2N^{2\xi}}{N}\left(\frac{1}{\eta_{1,s}^2}+\frac{1}{\eta_{2,s}^2}\right)\rho_{1,s}\rho_{2,s}\mathfrak{s}_2(s)^2 \, \dd s\lesssim \frac{N^{2\xi}\rho_{1,t}\rho_{2,t}}{N\ell_t}\mathfrak{s}_2(t)^2,
	\end{equation*}
	where we used that by a Schwarz inequality we have
	\[
	\begin{split}
		\big| \langle \Im G_{1,s} A\Im G_{2,s} A^* G_{1,s}\rangle\big|&\le\frac{1}{\eta_{1,s}}\langle \Im G_{1,s} A\Im G_{2,s} A^*\rangle, \\
		\big| \langle \Im G_{2,s} A\Im G_{1,s} A^* G_{2,s}\rangle\big|&\le\frac{1}{\eta_{2,s}}\langle \Im G_{1,s} A\Im G_{2,s} A^*\rangle.
	\end{split}
	\]
	By using again the BDG inequality, as in \eqref{eq:BDG}, we conclude with the estimate of the stochastic term in \eqref{eq:imGAimGA}:
	\begin{equation*}
		\begin{split}
			&\sup_{u \in [\tinit , t]}\left|\sum_{a,b=1}^N\int_{\tinit}^u\partial_{ab} \langle \Im G_{1,s}A \Im G_{2,s} A^*\rangle\frac{\lambda (\dd B_s)_{ab}}{\sqrt{N}}\right| \\
			&\le N^{\xi/2} \left(\int_{\tinit}^t\frac{\lambda^2}{N^2}\left(\frac{1}{\eta_{1,s}^2}+\frac{1}{\eta_{2,s}^2}\right)\langle \Im G_{1,s}A\Im G_{2,s}A^*\Im G_{1,s}A\Im G_{2,s} A^*\rangle\, \dd s\right)^{1/2} \\
			&\le  N^{\xi/2}  \left(\int_{\tinit}^t\frac{\lambda^2}{N}\left(\frac{1}{\eta_{1,s}^2}+\frac{1}{\eta_{2,s}^2}\right)\langle \Im G_{1,s}A\Im G_{2,s}A^*\rangle^2\, \dd s\right)^{1/2} \\
			&\lesssim  N^{\xi/2} \frac{\rho_{1,t}\rho_{2,t}}{\sqrt{N\ell_t}}\mathfrak{s}_2(t)^2,
		\end{split}
	\end{equation*}
	where in the second inequality we used the \emph{reduction inequality} 
	\[
	\langle \Im G_{1,s}A\Im G_{2,s}A^*\Im G_{1,s}A\Im G_{2,s}A^* \rangle\le N\langle \Im G_{1,s}A\Im G_{2,s}A^*\rangle^2
	\]
	that relies on an elementary trace inequality using operator positivity of $\sqrt{\Im G_{1,s}} A \Im G_{2,s} A^* \sqrt{\Im G_{1,s}}$. 
	
	 Finally, integrating \eqref{eq:imGAimGA} in time, and using the estimate
	\begin{equation*}
\langle \Im G_{1,\tinit}A\Im G_{2,\tinit}A^*\rangle\lesssim N^{\xi/2}\rho_{1,0}\rho_{2,\tinit} \mathfrak{s}_2(\tinit)^2\lesssim N^{\xi/2}\rho_{1,t}\rho_{2,t} \mathfrak{s}_2(t)^2,
	\end{equation*}
for the initial condition we conclude that
\begin{equation} \label{eq:GAGAfinal}
\langle \Im G_{1,t}A\Im G_{2,t}A^*\rangle \le N^{\xi/2} \rho_{1,t}\rho_{2,t} \mathfrak{s}_2(t)^2\,, 
\end{equation}
additionally using that $\xi\le \epsilon/10$. 

By combining \eqref{eq:calerr} and \eqref{eq:GAGAfinal}, and recalling the definition of the stopping time \eqref{eq:stoptime}, we conclude that $\tau =  \tfin$, which finishes the proof of Proposition \ref{prop:zigstep2}, similarly to the proof of Proposition \ref{prop:zigstep}. \qed
\subsection{Green function comparison: Proof of Proposition \ref{prop:zagstep2}} \label{subsec:zagstep2}  
The goal of this section is to give the proof of Proposition \ref{prop:zagstep2}, the \emph{zag} step, i.e., remove the Gaussian component introduced during the \emph{zig} step in Proposition \ref{prop:zigstep2} by a \emph{Green function comparison} (GFT) approach. Since throughout the argument, the time $t_k$ defined in \eqref{eq:tkdef} remains fixed, we shall henceforth drop the subscript/argument $t_k$, and, additionally,  abbreviate  $s_{\rm final} := s_k$. Recall that the spectral parameters remain fixed in the zag step, and also the deterministic approximation $M$ is kept unchanged. 

Contrary to earlier GFTs, in this paper, we do \emph{not} involve isotropic resolvent chains of the form $(GAG)_{ab}$ but conclude the \emph{zag} step self-consistently within the two types of average quantities $\langle GAG\rangle$ and $\langle \Im GA \Im GA^* \rangle$. This conceptual novelty of our argument considerably simplifies the entire proof, since tracking isotropic resolvent chains now becomes redundant in both the zig and zag step. It is possible to conduct such a \emph{non-isotropic zag step}, mainly due to three reasons: First, we consider only short chains; the tricks applied in the proofs below will at least become much more tedious for longer average chains, and might even fail entirely. Second, the main object of interest, $\langle  \Im G A \Im GA^* \rangle$ is inherently positive,\footnote{Even as an operator, after using cyclicity of the trace.} and our argument below crucially uses the positivity of $\Im G$ itself in various Schwarz estimates, as the ones collected in Lemma \ref{lem:Schwarz} below. Third, our fundamental control quantity $\s_2$ measures the observables only in terms of (a non-mean field version of) the Hilbert-Schmidt norm. As is known from the Wigner case \cite{cipolloni2023eigenstate, cipolloni2022rank, cipolloni2024out}, local law estimates involving the Hilbert-Schmidt norm of observables are suboptimal in terms of the $N \ell$-power (compared to operator norm or finer Schatten norms). Hence, in our estimates we can be ''generous'' with such $N\ell$-powers. 

Our proof of Proposition \ref{prop:zagstep2} is based on the following two Gronwall estimates, formulated in Propositions \ref{prop:zagavGron2G}--\ref{prop:zagavGron2G1A}. Their proofs are given in Sections \ref{subsubsec:zagav2G} and \ref{subsubsec:zagav2G1A}, respectively. 
\begin{proposition}[Gronwall estimate for $\langle \Im G A \Im GA^*\rangle$] \label{prop:zagavGron2G}
	Adopt the setting and notations from Proposition \ref{prop:zagstep2} and, for $s \in [0, \sfin]$, define
	\begin{equation} \label{eq:Rsdef}
		R_s := \big\langle \Im G^s_1 A \Im G^s_2A^*  -  \widehat{M}_{12}^A \big\rangle \,. 
	\end{equation}
	Then, for any large (even) $p\in \N$ and arbitrarily small $\delta > 0$ it holds that
	\begin{equation} \label{eq:avGron2G}
		\left| \frac{\dd }{\dd s} \E |R_s|^p \right| \lesssim \left(1 +   N^{-2\delta} \lambda^2 \max_i\big(\rho_i \eta_i^{-1}\big) \right) \, \left[\E |R_s|^p + \left(N^{3\delta} \frac{\rho_1 \rho_2}{\sqrt{N \ell}} \big(\s_2\big)^2\right)^p\right]
	\end{equation}
	uniformly in $s \in [0,\sfin]$ and  $(z_1, z_2)$-regular $A \in \C^{N \times N}$. 
\end{proposition}

\begin{proposition}[Gronwall estimate for $\langle  G A G\rangle$] \label{prop:zagavGron2G1A}
	Adopt the setting and notations from Proposition \ref{prop:zagstep2} and, for $s \in [0, \sfin]$, define
	\begin{equation} \label{eq:Sdef}
		S_s := \big\langle  G^s_1 A G^s_2  -  M_{12}^A \big\rangle \,. 
	\end{equation}
	Assume that $R_s$ from \eqref{eq:Rsdef} satisfies $|R_s| \lesssim \rho_1 \rho_2 \big(\s_2\big)^2$ with very high probability, uniformly in $s \in [0, \sfin]$. 
	Then for any large (even) $p\in \N$ and arbitrarily small $\delta > 0$, it holds that
	\begin{equation} \label{eq:avGron2G1A}
		\left| \frac{\dd }{\dd s} \E |S_s|^p \right| \lesssim \left(1 +   N^{-2\delta} \lambda^2 \max_i\big(\rho_i \eta_i^{-1}\big) \right) \, \left[\E |S_s|^p + \left(N^{3\delta} \sqrt{\frac{\rho_1 \rho_2}{N \eta_1 \eta_2 }} \s_2 \right)^p\right]
	\end{equation}
	uniformly in $s \in [0,\sfin]$ and $(z_1, z_2)$-regular $A \in \C^{N \times N}$. 
	
	The same statement holds with $G^s_1$ and/or $G^s_2$ replaced by their adjoints $(G^s_1)^*$ and $(G^s_2)^*$, respectively. 
\end{proposition}
\begin{proof}[Proof of Proposition \ref{prop:zagstep2}]
Armed with Propositions \ref{prop:zagavGron2G}--\ref{prop:zagavGron2G1A}, we can apply Gronwall's lemma (using the notation $\sfin = s_k$): First, from Proposition \ref{prop:zagavGron2G}, uniformly in $s \in [0, \sfin]$, we find that, analogously to \eqref{eq:integrateisoGron} and \eqref{eq:integrateavGron}, 
\begin{equation} \label{eq:integrate2G}
	\begin{split}
		\E |R_s|^p &\lesssim \exp\left[ \bigg(1 + N^{-2 \delta} \lambda^2 \max_i (\rho_i \eta_i^{-1})\bigg) (\sfin - s)\right] \, \left[\E |R_{\sfin}|^p + \left(N^{3\delta} \frac{\rho_1 \rho_2}{\sqrt{N \ell}} \big(\s_2\big)^2\right)^p\right] \\
		&\lesssim \exp(N^{-2 \delta}) \left[\E |R_{\sfin}|^p + \left(N^{3\delta} \frac{\rho_1 \rho_2}{\sqrt{N \ell}} \big(\s_2\big)^2\right)^p\right] \lesssim \left[\E |R_{\sfin}|^p + \left(N^{3\delta} \frac{\rho_1 \rho_2}{\sqrt{N \ell}} \big(\s_2\big)^2\right)^p\right]
	\end{split}
\end{equation}
where we used that $\lambda^2 \max_i (\rho_i \eta_i^{-1}) \lesssim N^{k \delta}$ by \eqref{eq:tkdef} and $\sfin \lesssim N^{-(k-1)\delta}$ by \eqref{eq:dtkdef}. 

To estimate $\E |R_{\sfin}|^p$, we recall that the $G_i^s$'s satisfy the two resolvent law \eqref{eq:2G2Azag} uniformly in $(z_1, z_2)$-regular $A$'s at $s = \sfin$. Hence, because $p$ and $\delta$ were arbitrary, we deduce that \eqref{eq:2G2Azag} holds uniformly in $(z_1, z_2)$-regular $A$ and times $s \in [0, \sfin]$. 

Second, now that we have established $|R_s| \lesssim \rho_1 \rho_2 \big(\s_2\big)^2$ with very high probability, uniformly in $s \in [0, \sfin]$, we can employ Proposition \ref{prop:zagavGron2G1A}. Indeed, uniformly in $s \in [0, \sfin]$, we find that, analogously to \eqref{eq:integrate2G}, 
\begin{equation*} 
		\E |S_s|^p  \lesssim \left[\E |S_{\sfin}|^p +\left(N^{3\delta} \sqrt{\frac{\rho_1 \rho_2}{N \eta_1 \eta_2}} \s_2 \right)^p\right] \,. 
\end{equation*}

To estimate $\E |S_{\sfin}|^p$, we recall that the $G_i^s$'s satisfy the two resolvent law \eqref{eq:2G1Azag} uniformly in $(z_1, z_2)$-regular $A$'s at $s = \sfin$. Hence, because $p$ and $\delta$ were arbitrary, we deduce that \eqref{eq:2G1Azag} holds uniformly in $(z_1, z_2)$-regular $A$ and times $s \in [0, \sfin]$. This finishes the proof of Proposition \ref{prop:zagstep2}. 
\end{proof}
Before going into the proofs of Propositions \ref{prop:zagavGron2G}--\ref{prop:zagavGron2G1A}, we collect a few bounds on isotropic resolvent chains, obtained by ''clever'' Schwarz inequalities, that shall be frequently used below. 
\begin{lemma}[Schwarz for isotropic chains] \label{lem:Schwarz} Dropping the $s$-superscript for notational simplicity, for any $\bm x, \bm y \in \C^N$, it holds that
\begin{align}
	\big(G_1 A G_2 \big)_{\bm x \bm y} &\le \Vert \bm y \Vert \sqrt{\big(G_1 A \Im G_2 A^* G_1^*\big)_{\bm x \bm x}} \eta^{-1/2}_2 \wedge \Vert \bm x \Vert \sqrt{\big(G_2 A \Im G_1 A^* G_2^*\big)_{\bm y \bm y}} \eta^{-1/2}_1 \label{eq:iso1Im0}\\
	&\prec \Vert \bm x \Vert \, \Vert \bm y \Vert \, \frac{\sqrt{N}}{\eta} \sqrt{\rho_1 \rho_2} \, \s_2\,,  \nonumber\\
	\big(\Im G_1 A G_2\big)_{\bm x \bm y} &\prec \sqrt{(\Im M_1)_{\bm x \bm x}} \sqrt{\big(G_2^* A^* \Im G_1 A G_2\big)_{\bm y \bm y}}  \label{eq:iso1Im1} \prec \Vert \bm y \Vert \, \sqrt{\frac{N}{\eta_2}} \sqrt{(\Im M_1)_{\bm x \bm x}} \, \sqrt{\rho_1 \rho_2} \, \s_2\,,  \\
	\big(\Im G_1 A \Im G_2\big)_{\bm x \bm y} &\prec \sqrt{(\Im M_1)_{\bm x \bm x}} \sqrt{\big(\Im G_2 A^* \Im G_1 A \Im G_2\big)_{\bm y \bm y}}  \nonumber \\
	&\prec  \sqrt{N} \sqrt{(\Im M_1)_{\bm x \bm x}(\Im M_2)_{\bm y \bm y}} \, \sqrt{\rho_1 \rho_2} \, \s_2 \,, \label{eq:iso1Im2}
\end{align}
as well as
\begin{align}
	\big(G_1 A \Im G_2 A^* G_1^*\big)_{\bm x \bm x} &\le \Vert \bm x \Vert \, \sqrt{\big( G_1 A \Im G_2 A^* \Im G_1 A \Im G_2 A^* G_1^*\big)_{\bm x \bm x}} \eta_1^{-1/2}\label{eq:iso2Im1} \\
	&\prec \Vert \bm x \Vert \, \sqrt{\frac{N}{\eta_1}} \sqrt{\rho_1 \rho_2} \s_2^{\rm av}(z_1, z_2;A) \sqrt{\big(G_1 A \Im G_2 A^* G_1^*\big)_{\bm x \bm x}}  \prec \Vert \bm x \Vert^2 \frac{N}{\eta_1} \rho_1 \rho_2 \big(\s_2\big)^2 \,,  \nonumber \\
	\big(\Im G_1 A \Im G_2 A^* \Im G_1\big)_{\bm x \bm x} &\prec N (\Im M_1)_{\bm x \bm x} \rho_1 \rho_2 \big(\s_2\big)^2 \,. \label{eq:iso2Im3}
\end{align}
\end{lemma}

\subsubsection{Proof of Proposition \ref{prop:zagavGron2G}} \label{subsubsec:zagav2G}
Similarly to the proofs of Propositions \ref{prop:zagisoGron} and \ref{prop:zagavGron}, we perform a cumulant expansion (cf.~\eqref{eq:cumexbasiciso} and \eqref{eq:cumexbasicav}). The main technical work of these arguments is to control the various terms arising by differentiation of resolvents, which is what we will focus on here, as the conclusion of the proof by application of various Young inequalities is routine. 

For ease of notation, we shall henceforth omit the time dependence, drop all the indices of $G_i, \eta_i, \rho_i$ and further assume that $A$ is self-adjoint, i.e.~$A = A^*$. In Section \ref{subsubsec:12}, we elaborate on a few exemplary terms, where we carry the precise index structure of $G$'s, $\eta$'s, and $\rho$'s. Further, note that, by differentiation of resolvents, the total number of $\Im G$'s is preserved, since 
\begin{equation}
	\partial_{ab} \Im G = - \Im  \Delta_{ab} G  - G^* \Delta_{ab} \Im G
\end{equation}
where $\Delta_{ab}$ is the matrix having all zero entries except at $(a,b)$. 
Finally, we also do not distinguish between $G$ and $G^*$ as our argument is agnostic to the difference.  Although sometimes, if we wish to stress the positivity of certain terms or use Ward identities, we explicitly write $G^*$ and $A^*$. We also point out that, whenever we encounter isotropic chains of the form $(GA\Im GAG)_{ab}$, the middle $\Im G$ is always an actual $\Im G$, since it has not been differentiated. The common thread of the whole proof is to make sure that the summations over matrix entries indexed by $a,b$ are \emph{effective}: This means that, as a toy example for a general matrix $B$, instead of $\sum_{a,b} |B_{ab}|^2 \le N^2 \max_{a,b} |B_{ab}|^2 \le N^2 \Vert B \Vert^2$ we use that $\sum_{a,b} |B_{ab}|^2 \le N \langle BB^* \rangle \le N \Vert B \Vert^2$, effectively saving an entire power of $N$. 

We start by controlling terms of order $k=3$ in the cumulant expansion and distinguish them by the number $n=1,2,3$ of "destroyed" resolvent chains, adopting the notations and conventions from the proofs of Propositions \ref{prop:zagisoGron} and \ref{prop:zagavGron}. For $n=1$, we estimate three exemplary terms: 
\begin{equation} \label{eq:n11}
	\begin{split}
		&\frac{\lambda^3}{N^{5/2}} \sum_{a,b} \left| \big(GA \Im G A G\big)_{aa} (\Im G)_{ab} G_{bb} \right| \\
		\lesssim &\frac{\lambda^3}{N^{5/2}} \frac{1}{\lambda \sqrt{\rho}}\sum_a \big(GA \Im G A^* G^*\big)_{aa} \sum_b \sqrt{(\Im M)_{bb}} (|M|)_{bb} \\
		\lesssim & \frac{\lambda^2}{N^{3/2} \sqrt{\rho} \eta} \langle \Im GA \Im GA^* \rangle \sqrt{\sum_b (\Im M)_{bb}} \sqrt{\sum_{ab} |M_{ab}|^2} 
		\lesssim  \frac{\lambda}{\sqrt{N} \eta} \langle \Im GA \Im GA^* \rangle
	\end{split}
\end{equation}
In the first step, we used a Schwarz inequality to estimate $(GA \Im G A G)_{aa} \le (GA \Im GA^*G^*)_{aa}$ (by writing $\Im G = \sqrt{\Im G} \sqrt{\Im G}$), and employed the single resolvent local law and a Schwarz inequality to bound $|(\Im G)_{ab}| \lesssim \sqrt{(\Im M)_{aa}} \sqrt{(\Im M)_{bb}}$ together with the norm bound $\Vert M \Vert \lesssim 1/(\lambda^2 \rho)$ for controlling $\sqrt{(\Im M)_{aa}}$, and the single resolvent law again for bounding $|G_{bb}| \lesssim (|M|)_{bb}$. To go to the last line, we then carried out the $a$-summation, used a Ward identity, and employed a Schwarz inequality for the $b$-summation, where we additionally \emph{extended the summation} to control
\begin{equation} \label{eq:sumextend}
	\sum_b \big|(|M|)_{bb}\big|^2 \le \sum_{ab} |M_{ab}|^2 \,. 
\end{equation}
In the ultimate step, we then used that $\langle |M|^2 \rangle \le 1/\lambda^2$ and $\langle \Im M \rangle \sim \rho$. 

As a second exemplary term, we consider
\begin{equation} \label{eq:n12}
	\begin{split}
		&\frac{\lambda^3}{N^{5/2}} \sum_{a,b} \left| \big(GA \Im G A \Im G\big)_{ba}  G_{aa} G_{bb} \right| 
		\lesssim \frac{\lambda^3}{N^{5/2}} \sqrt{\sum_{a,b} \big|  \big(GA \Im G A \Im G\big)_{ba} \big|^2} \sum_{a,b} |M_{ab}|^2 \\
		\lesssim & \frac{\lambda}{N} \langle GA\Im GA^* \Im G \Im G A \Im G A^* G^* \rangle^{1/2} \le \frac{\lambda}{\sqrt{N} \eta} \langle \Im G A \Im GA^* \rangle
	\end{split}
\end{equation}
where in the first step we used a Schwarz inequality together with the single resolvent law $|G_{aa}| \lesssim (|M|)_{aa}$ and the summation extension \eqref{eq:sumextend}. To go to the second line, we then used $\langle |M|^2 \rangle \le 1/\lambda^2$ and in the last step, a Ward indentity, the norm bound $\| \Im G\| \le \eta^{-1}$, and the trace inequality
\begin{equation*}
	\langle \Im GA \Im G A^* \Im G A \Im GA^* \rangle  \le N \langle \Im GA \Im GA^* \rangle
\end{equation*}
as follows by writing the first and third $\Im G$ in the fourfold chain as $\sqrt{\Im G} \sqrt{\Im G}$. 

As the third an final exemplary term, analogously to \eqref{eq:n11} and \eqref{eq:n12} above, we consider
\begin{equation} \label{eq:n13}
	\begin{split}
		&\frac{\lambda^3}{N^{5/2}} \sum_{a,b} \left| \big(GA  G\big)_{ba}  (\Im G A \Im G)_{aa} G_{bb} \right|  \\
		\lesssim &\frac{\lambda^3}{N^{5/2}} \sqrt{\sum_{a,b} \big|  \big(GA G \big)_{ba} \big|^2} \sqrt{\sum_{a,b} \big|  \big(\Im GA \Im G \big)_{ba} \big|^2} \sqrt{\sum_{a,b} |M_{ab}|^2} 
		\lesssim  \frac{\lambda^2}{N \eta^2} \langle \Im G A \Im GA^* \rangle
	\end{split}
\end{equation}
All other terms of order $k=3$ with $n=1$ destroyed chain can be handled similarly. 

Next, we turn to terms with $n=2$ destroyed chains and again consider two exemplary terms. As the first one, we estimate
\begin{equation} \label{eq:k3n2a}
	\begin{split}
		&\frac{\lambda^3}{N^{7/2}} \sum_{a,b} \left| (GA\Im GA^* \Im G)_{ab} (GA \Im G)_{ab} (GA \Im G)_{ba} \right| \\
		\lesssim & \frac{\lambda^3}{N^{7/2}} \frac{N}{\sqrt{\eta \lambda^2 \rho}} \langle \Im GA \Im GA^* \rangle \sum_{a,b} \big| (GA\Im G)_{ab}\big|^2 \lesssim \frac{\lambda^2 \rho}{\eta}\frac{1}{(N \eta \rho)^{3/2}} \langle \Im GA\Im GA^* \rangle^2 \,. 
	\end{split}
\end{equation}
In the first step, we used a Schwarz inequality and \eqref{eq:iso2Im1}--\eqref{eq:iso2Im3} to estimate
\begin{equation} \label{eq:mixedImnoIm}
	|(GA\Im GA^* \Im G)_{ab}| \le \sqrt{(GA\Im GA^* G^*)_{aa}} \sqrt{(\Im GA \Im GA^* \Im G)_{bb}} \lesssim \frac{N}{\sqrt{\eta \lambda^2 \rho}} \langle \Im GA \Im GA^* \rangle
\end{equation}
where we additionally used the norm bound $\Vert \Im M \Vert \lesssim 1/(\lambda^2 \rho)$ to control the $\sqrt{(\Im M)_{bb}}$ arising by estimating $\sqrt{(\Im GA \Im GA^* \Im G)_{bb}}$ via \eqref{eq:iso2Im3}. 

As a second exemplary term, we consider 
\begin{equation} \label{eq:k3n2}
	\begin{split}
		&\frac{\lambda^3}{N^{7/2}} \sum_{a,b} \left| (GA\Im GA^* \Im G)_{ab} (GA \Im G)_{aa} (GA \Im G)_{bb} \right| \\
		\lesssim & \frac{\lambda^3}{N^{7/2}} \frac{1}{\lambda^2 \rho} \left(\sum_{a} (GA\Im GA^* G^*)_{aa}\right)^{3/2} \sqrt{\sum_{a} (\Im GA\Im GA^* \Im G)_{aa}}  \\
		\lesssim &\sqrt{\frac{\lambda^2 \rho}{\eta}}\frac{1}{(N \eta \rho)^{3/2}} \langle \Im GA\Im GA^* \rangle^2 \,,
	\end{split}
\end{equation}
where in the first step, we estimated 
\begin{equation}
	|(GA \Im G)_{aa}| \le \sqrt{( GA \Im GA^* G^*)_{aa}} \sqrt{(\Im G)_{aa}} \lesssim \lambda^{-1} \rho^{-1/2} \sqrt{( GA \Im GA^* G^*)_{aa}}
\end{equation}
and the first bound in \eqref{eq:mixedImnoIm}. To go to the last line, we then used the Ward identity $GG^* = \Im G/\eta$ and the norm bound $\Vert \Im G \Vert \le 1/\eta$. All other terms of order $k=3$ with $n=2$ destroyed chain can be handled similarly. 

As the last possibility for terms of order $k=3$, we consider the case of $n=3$ destroyed chains. In this case, similarly to \eqref{eq:k3n2}, we estimate
\begin{equation} \label{eq:k3n3}
	\begin{split}
		&\frac{\lambda^3}{N^{9/2}} \sum_{a,b} \left| (GA\Im GA^* \Im G)_{ab}  \right|^3 \\
		\lesssim & \frac{\lambda^2}{N^{7/2}} \frac{1}{\sqrt{\eta \rho}} \left(\sum_{a} (GA\Im GA^* G^*)_{aa} \right) \left(\sum_{a} (\Im GA\Im GA^* \Im G)_{aa}\right)   \\
		\lesssim &{\frac{\lambda^2 \rho}{\eta}}\frac{1}{(N \eta \rho)^{3/2}} \langle \Im GA\Im GA^* \rangle^2 \,,
	\end{split}
\end{equation}
All the above terms of order $k=3$ in \eqref{eq:n11}, \eqref{eq:n12}, \eqref{eq:n13}, \eqref{eq:k3n2a}, \eqref{eq:k3n2}, and \eqref{eq:k3n3}, together with the $p-n$ not differentiated factors of $R_s$, can then be controlled as in \eqref{eq:avGron2G} with the aid of Young inequalities. 

Next, we consider terms of order $k \ge 4$ and distinguish two cases for the number $n$ of destroyed chains: (i) $n \in \{k-1, k\}$, and (ii) $n \in \{1, ..., k-2\}$. In case (i), it is a trivial combinatorial argument to see that there are least two chains which are hit by a derivative exactly once. The remaining $k-2$ derivatives are then distributed over the remaining $n-2$ chains, which are all hit by at least one derivative. For controlling the derivatives, we use \eqref{eq:iso1Im0}--\eqref{eq:iso2Im3} and \eqref{eq:mixedImnoIm} together with the single resolvent bound $|G_{ab}| \lesssim 1/(\lambda^2 \rho)$ to find that
\begin{equation} \label{eq:morederiv}
	\big| \partial_{ab}^m \partial_{ba}^\ell \langle \Im GA \Im GA^* \rangle \big| \lesssim \left(\frac{1}{\sqrt{\eta \lambda^2 \rho}}\right)^{m + \ell} \langle \Im GA \Im GA^*\rangle
\end{equation}
where we additionally used that $\eta \lesssim \lambda^2 \rho$ in our regime of interest. Hence, in case (i), keeping only the two factors, which are hit by a derivative exactly once and estimating the other ones by \eqref{eq:morederiv}, we find these terms to be bounded by
\begin{equation} \label{eq:kge4casei}
	\begin{split}
		\frac{\lambda^k}{N^{k/2+2}} \frac{1}{(\eta \lambda^2 \rho)^{(k-2)/2}} \sum_{a,b} \big|(GA \Im GA^* \Im G)_{ab}\big|^2 \langle \Im G A \Im GA^* \rangle^{n-2} 
		\lesssim   \lambda^2 \frac{\rho}{\eta} \frac{1}{(N \eta \rho)^{k/2}} \langle \Im GA \Im GA^* \rangle^n 
	\end{split}
\end{equation}
where we used that $ \sum_{a,b} \big|(GA \Im GA^* \Im G)_{ab}\big|^2 \lesssim (N/\eta)^2 \langle \Im GA \Im GA^* \rangle^2$. 

Next we turn to case (ii), where we will estimate almost all terms arising from taking derivatives with the aid of \eqref{eq:iso1Im0}--\eqref{eq:iso2Im3}. The only exception are terms from at most two of the destroyed chains. More precisely, we distinguish three (sub)cases: 
\begin{itemize}
	\item[(a)] there exists (at least) one chain, which is hit by at least four derivatives;
	\item[(b)] there exists (at least) one chain, which is is hit by exactly three derivatives;
	\item[(c)] there exist at least two chains which are hit by exactly two derivatives. 
\end{itemize}
Note that, by simple counting and the assumption $n \in \{1, ... , k-2\}$ this list exhausts all possible cases. 

We start with case (a) and discuss two exemplary terms. First, we consider
\begin{equation} \label{eq:casea1}
	\frac{\lambda^k}{N^{k/2+1}} \left(\frac{1}{\sqrt{\eta \lambda^2 \rho}}\right)^{k-4} \sum_{a,b} \left| (GAG)_{aa} G_{bb} \right|^2 \langle \Im GA \Im GA^*\rangle^{n-1}
\end{equation}
where we bounded all terms arising from the other $n-1$ destroyed chains (having an unspecified number of derivatives hitting the chain) with the aid of \eqref{eq:morederiv}. Moreover, potential further single resolvent factors from the specific chain hit by at least four derivatives are estimated via a single resolvent law and norm bound on $M$, yielding $\max_{a,b} |G_{ab}| \lesssim 1/(\lambda^2 \rho)$, and using that $\eta \lesssim \lambda^2 \rho$ in our regime. We point out that we dropped all (potentially) preserved $\Im G$'s in \eqref{eq:casea1}, as they are irrelevant for the following estimates. 
We then continue to bound \eqref{eq:casea1} as
\begin{equation} \label{eq:casea1explain}
	\begin{split}
		\eqref{eq:casea1} &\lesssim \frac{\lambda^4 \eta^2 \rho^2 }{(N \eta \rho)^{k/2}} \frac{1}{N}\bigg(\sum_{a,b} |(GAG)_{ab}|^2\bigg) \bigg(\sum_{a,b} |M_{ab}|^2\bigg) \langle \Im GA \Im GA^*\rangle^{n-1} \\
		&\lesssim \lambda^2 \frac{\rho}{\eta} \,  \frac{1}{(N \eta \rho)^{(k-2)/2}} \langle \Im GA \Im GA^*\rangle^n \lesssim \lambda^2 \frac{\rho}{\eta} \,  \frac{1}{(N \eta \rho)^{n/2}} \langle \Im GA \Im GA^*\rangle^n \,. 
	\end{split}
\end{equation}
In the first step, we used a Schwarz inequality together with a single resolvent bound and \eqref{eq:sumextend}. To go to the second line, we employed $\langle |M|^2\rangle \le 1/\lambda^2$ and two Ward identities. In the ultimate step, we then used that, by assumption, $n \le k-2$ and $N \eta \rho \ge 1$. As a second exemplary term in case (a) we consider 
\begin{equation} \label{eq:casea2} 
	\begin{split}
		&\frac{\lambda^k}{N^{k/2+1}} \left(\frac{1}{\sqrt{\eta \lambda^2 \rho}}\right)^{k-4} \sum_{a,b} \left| (GA\Im G A^* G^*)_{aa} G_{bb} G_{bb} G_{aa}\right| \langle \Im GA \Im GA^*\rangle^{n-1} \\
		\lesssim & \frac{\lambda^3 \eta^{3/2} \rho^{3/2}}{(N \eta \rho)^{k/2}} \frac{1}{N} \sum_a (GA\Im G A^* G^*)_{aa}\sum_{a,b} |M_{ab}|^2 \langle \Im GA \Im GA^*\rangle^{n-1} \\
		\lesssim & \sqrt{\frac{\lambda^2 \rho}{\eta}} \frac{1}{(N \eta \rho)^{n/2}}\langle \Im GA \Im GA^*\rangle^{n} \,. 
	\end{split}
\end{equation}
Here, to go the second line, we estimated $G_{aa}$ by single resolvent law and norm bound on $M$, together with $\lambda^2 \rho \gtrsim \eta$, and used a single resolvent law with \eqref{eq:sumextend} to bound $\sum_b |G_{bb}|^2$. In the second step argued similarly to \eqref{eq:casea1explain}, in particular using $N \eta \rho \ge 1$ and $n \le k-2$. All other terms in case (a) can be handled similarly as the ones in \eqref{eq:casea1} and \eqref{eq:casea2}. 

Next, we turn to case (b). Here, we may bound the terms arising by exactly three derivatives exactly as in \eqref{eq:n11} and \eqref{eq:n12}. The other differentiated chains can be bounded with the aid of \eqref{eq:morederiv}. Hence, we find that 
\begin{equation} \label{eq:caseb}
	\text{(b)--terms} \lesssim \frac{\lambda^k}{N^{k/2+ 1}} \left(\frac{1}{\sqrt{\eta \lambda^2 \rho}}\right)^{k-3} \frac{N^2}{\lambda^2 \eta} \langle \Im GA \Im GA^* \rangle^n \lesssim  \sqrt{\frac{\lambda^2 \rho}{\eta}}\frac{1}{(N \eta \rho)^{n/2}} \langle \Im GA \Im GA^*\rangle^n
\end{equation}
again using using $N \eta \rho \ge 1$ and $n \le k-2$ in the last step. 

Finally, we turn to case (c), where we again consider some exemplary terms. The main differences compared to cases (a) and (b) above is, that now, similarly to case (i), we need to consider the terms arising from differentiation of \emph{two} chains (compared to only one in cases (a) and (b)) in order to carry out the summation over $a,b$ in an effective way. By a Schwarz inequality, we may assume w.l.o.g.~that both chains, which are differentiated twice, exhibit the exact same index structure. Moreover, all other differentiated chains are simply bounded via \eqref{eq:morederiv}. 

Then, as the first exemplary term, we consider
\begin{equation} \label{eq:casec1}
	\begin{split}
		&\frac{\lambda^k}{N^{k/2 + 2}} \bigg(\frac{1}{\sqrt{\eta \lambda^2 \rho}}\bigg)^{k-4} \sum_{a,b} \left| \big(GAG\big)_{aa} (\Im G A \Im G)_{bb} \right|^2 \langle \Im GA \Im G A^* \rangle^{n-2} \\
		\lesssim & \frac{\lambda^2 \eta \rho }{(N \eta \rho)^{k/2}} \bigg(\frac{1}{N} \sum_{a} (GA\Im GA^* G^*)_{aa}\bigg)  \bigg(\frac{1}{N} \sum_{b} (\Im GA\Im GA^* \Im G)_{bb}\bigg) \langle \Im GA \Im G A^* \rangle^{n-2} \\
		\lesssim & \frac{\lambda^2 \rho}{\eta} \frac{1}{(N \eta \rho)^{n/2}}  \langle \Im GA \Im G A^* \rangle^{n} \,. 
	\end{split}
\end{equation}
In the first step, we bounded $(GAG)_{aa}$ and $(\Im G A \Im G)_{bb}$ with the aid of the first estimates in \eqref{eq:iso1Im0} and \eqref{eq:iso1Im2}, respectively. In the second step, we then employed the Ward identity, a norm bound $\Vert \Im G \Vert \le 1/\eta$, and used that (trivially) $n \le k$. 

As a second exemplary term, we consider
\begin{equation} \label{eq:casec2}
	\begin{split}
		&\frac{\lambda^k}{N^{k/2 + 2}} \left(\frac{1}{\sqrt{\eta \lambda^2 \rho}}\right)^{k-4} \sum_{a,b} \left| \big(GA\Im GA^* \Im G\big)_{aa} G_{bb} \right|^2 \langle \Im GA \Im G A^* \rangle^{n-2} \\
		\lesssim & \frac{N \lambda^2 \eta^2 \rho }{(N \eta \rho)^{k/2}} \bigg(\frac{1}{N} \sum_{a} (GA\Im GA^* G^*)_{aa}\bigg)  \bigg(\frac{1}{N} \sum_{a,b} |M_{ab}|^2\bigg) \langle \Im GA \Im G A^* \rangle^{n-1} \\
		\lesssim &  \frac{1}{(N \eta \rho)^{n/2}}  \langle \Im GA \Im G A^* \rangle^{n} \,. 
	\end{split}
\end{equation}
where in the first step we used a Schwarz inequality for $(GA\Im G A^* \Im G)_{aa}$ and then estimated the arising $(\Im GA \Im GA^* \Im G)_{aa}$ via \eqref{eq:iso2Im3}. For $G_{bb}$ we argued exactly as in \eqref{eq:casea2}. In the ultimate step, we employed a Ward identity, the fact that $\langle |M|^2 \rangle \le 1/\lambda^2$,  and used $n \le k-2$. 

All the above terms of order $k\ge 4$ in \eqref{eq:kge4casei}, \eqref{eq:casea1explain}, \eqref{eq:casea2}, \eqref{eq:caseb}, \eqref{eq:casec1}, and \eqref{eq:casec2}, together with the $p-n$ not differentiated factors of $R_s$, can then be controlled as in \eqref{eq:avGron2G} with the aid of Young inequalities. 
 \qed

\subsubsection{Proof of Proposition \ref{prop:zagavGron2G1A}} \label{subsubsec:zagav2G1A}
The proof of Proposition \ref{prop:zagavGron2G1A} is very similar to that of Proposition \ref{prop:zagavGron2G}, in particular relying on cumulant expansion and recursive moment estimates. Hence, we will be very brief and only discuss a few important changes. We remark that, similarly to the previous proof, we refrain from distinguishing $G$ and $G^*$ and, moreover, drop all the indices of $G_i, \eta_i, \rho_i$, and the arguments of $\s_2^{\rm av}$.  

As the main difference, since there is no $\Im G$ involved, we do not need to watch out for particular constellations of imaginary parts enabling certain "positivity tricks" used in the proof of Proposition~\ref{prop:zagavGron2G}. Except from this, the estimates needed for the proof of Proposition \ref{prop:zagavGron2G1A} are essentially the same as in the proof of Proposition \ref{prop:zagavGron2G}. We now discuss a few exemplary terms arising in the cumulant expansion, as direct analogs to certain terms in the proof of Proposition \ref{prop:zagavGron2G}. As in that proof, the argument is completed by application of Young inequalities. 

As the first exemplary term, we consider the analog of \eqref{eq:n11} (in the nomenclature of the proof of Proposition \ref{prop:zagavGron2G} corresponding to $k=3$ and $n=1$), which we bound as
\begin{equation} \label{eq:n11A}
	\begin{split}
		&\frac{\lambda^3}{N^{5/2}} \sum_{a,b} \left| \big(GA G G\big)_{aa} G_{ab} G_{bb} \right| \\
		\lesssim &\frac{\lambda^3}{N^{5/2}} \frac{1}{\eta }\sqrt{\sum_a \big(GA \Im G A^* G^*\big)_{aa}} \sqrt{\sum_a (\Im M)_{aa}}\sum_{a,b} |M_{ab}|^2 
		\lesssim \sqrt{\frac{\lambda^2 \rho}{\eta}} \frac{1}{\sqrt{N}  \eta} \langle \Im GA \Im GA^* \rangle^{1/2} 
	\end{split}
\end{equation}
Next, we consider the analog of \eqref{eq:n13}, which we estimate as 
\begin{equation} \label{eq:n13A}
	\begin{split}
		&\frac{\lambda^3}{N^{5/2}} \sum_{a,b} \left| \big(G G\big)_{ba}  (G A G)_{aa} G_{bb} \right|  \\
		\lesssim &\frac{\lambda^3}{N^{5/2}} \sqrt{\sum_{a,b} \big|  \big(GG \big)_{ba} \big|^2} \sqrt{\sum_{a,b} \big|  \big(GA G \big)_{ba} \big|^2} \sqrt{\sum_{a,b} |M_{ab}|^2} 
		\lesssim  \frac{\lambda^2 \rho}{\eta}  \frac{1}{\sqrt{N \eta \rho}}\frac{1}{\sqrt{N} \eta} \langle \Im G A \Im GA^* \rangle^{1/2}
	\end{split}
\end{equation}

For $k=3$, $n=2$, we consider one exemplary term as an analog of \eqref{eq:k3n2}, which we bound as
\begin{equation} \label{eq:k3n2A}
	\begin{split}
		\frac{\lambda^3}{N^{7/2}} \sum_{a,b} \left| (GA G G)_{ab} (GA  G)_{aa} (GG)_{bb} \right| 
		&\lesssim  \frac{\lambda^3}{N^{7/2}} \frac{1}{\eta^{5/2}} \sum_{a} (GA\Im GA^* G^*)_{aa} \sum_{b} |(\Im G)_{bb}|^{3/2} \\
		& \lesssim {\frac{\lambda^2 \rho}{\eta}}\frac{1}{\sqrt{N \eta \rho}} \left( \frac{1}{\sqrt{N} \eta}\langle \Im GA\Im GA^* \rangle^{1/2} \right)^2\,,
	\end{split}
\end{equation}
where in the first step, we bounded $|(GAGG)_{ab}| \le \eta^{-1}\sqrt{(GA\Im G A^* G^*)_{aa} (\Im G)_{bb}}$ and similarly used $|(GAG)_{aa}| \le \eta^{-1/2} \sqrt{(GA\Im GA^* G^*)_{aa}}$ as well as $|(GG)_{bb}| \le \eta^{-1} (\Im G)_{bb}$. To go to the second line, we then bounded $\sqrt{(\Im G)_{bb}} \lesssim \frac{1}{\lambda \sqrt{\rho}}$ with the aid of a single resolvent law, and carried out the $a$- and $b$-summations, additionally using $\langle \Im G \rangle \lesssim \rho$, again by a single resolvent law. 

Lastly for $k=3$, we consider the $n=3$ case and estimate
\begin{equation} \label{eq:k3n3A}
	\begin{split}
		\frac{\lambda^3}{N^{9/2}} \sum_{a,b} \left| (GA G G)_{ab}  \right|^3 
		\lesssim & \frac{\lambda^2}{N^{4}} \frac{1}{\eta^{7/2} \rho^{1/2}} \sum_{a} (GA\Im GA^* G^*)_{aa} \sum_{b} (\Im G)_{bb} \\
		\lesssim &{\frac{\lambda^2 \rho}{\eta}}\frac{1}{(N \eta \rho)^{1/2}} \left(\frac{1}{\sqrt{N} \eta}\langle \Im GA\Im GA^* \rangle^{1/2}\right)^3 \,,
	\end{split}
\end{equation}
where we used the last bound in
\begin{equation} \label{eq:GAGGbound}
	\begin{split}
		\big| (GAGG)_{ab} \big| &\le \sqrt{(GAGG^* A^* G^*)_{aa}} \sqrt{(GG^*)_{bb}} \\
		&\lesssim \sqrt{\frac{N}{\eta^3 \lambda^2 \rho}} \langle \Im GA \Im GA^* \rangle^{1/2} \quad \text{uniformly in} \quad a,b \in [N]
	\end{split}
\end{equation}
for one of the $(GAGG)_{ab}$ factors in the first line of \eqref{eq:k3n3A}, and the intermediate bound in \eqref{eq:GAGGbound} together with a Ward identity for the other two $(GAGG)_{ab}$ factors in the first line of \eqref{eq:k3n3A}. In the second step in \eqref{eq:k3n3A} we argued similarly to the last estimate in \eqref{eq:k3n2A}. By arguing exactly as in the proof of Proposition~\ref{prop:zagavGron2G}, this concludes the discussion for terms of order $k=3$ in the cumulant expansion, since all other index constellations not explicitly covered in \eqref{eq:n11A}--\eqref{eq:k3n3A} can be handled similarly. 

We now turn to terms of order $k \ge 4$ and make the same case distinction as in the proof of Proposition~\ref{prop:zagavGron2G} (cf.~the discussion around \eqref{eq:morederiv}--\eqref{eq:kge4casei}). As an analog of the auxiliary estimate \eqref{eq:morederiv}, we now have that
\begin{equation} \label{eq:morederivA}
	\big| \partial_{ab}^m \partial_{ba}^\ell \langle GAG \rangle \big| \lesssim \left(\frac{1}{\sqrt{\eta \lambda^2 \rho}}\right)^{m + \ell} \frac{1}{\sqrt{N} \eta}\langle \Im GA \Im GA^*\rangle^{1/2} \,. 
\end{equation}

Then, the analog of \eqref{eq:kge4casei} (case (i) in the nomenclature of the proof of Proposition \ref{prop:zagavGron2G})becomes, 
\begin{equation*}
	\begin{split}
		&\frac{\lambda^k}{N^{k/2+2}} \frac{1}{(\eta \lambda^2 \rho)^{(k-2)/2}} \frac{1}{(\sqrt{N} \eta)^{n-2}}\sum_{a,b} \big|(GA  G G)_{ab}\big|^2 \langle \Im G A \Im GA^* \rangle^{\frac{n-2}{2}}  \\
		\lesssim   &\lambda^2 \frac{\rho}{\eta} \frac{1}{(N \eta \rho)^{k/2}} \frac{1}{(\sqrt{N} \eta)^n}\langle \Im GA \Im GA^* \rangle^{n/2}  \lesssim  \frac{\lambda^2 \rho}{\eta}  \frac{1}{(\sqrt{N} \eta)^n}\langle \Im GA \Im GA^* \rangle^{n/2} 
	\end{split}
\end{equation*}
where here and in the following $n \le k$ denotes the number of differentiated resolvent chains. 

In the complementary case (ii), we consider only one exemplary term in the subcase (a). In fact, the analog of \eqref{eq:casea1}--\eqref{eq:casea1explain} can be estimated as
\begin{equation} \label{eq:casea1A}
	\begin{split}
		&\frac{\lambda^k}{N^{k/2+1}} \left(\frac{1}{\sqrt{\eta \lambda^2 \rho}}\right)^{k-4} \frac{1}{(\sqrt{N} \eta)^{n-1}}\sum_{a,b} \left| (GAG)_{aa} (GG)_{aa} G_{bb} G_{bb} \right| \langle \Im GA \Im GA^*\rangle^{\frac{n-1}{2}} \\
		\lesssim &\frac{\lambda^4 \eta^2 \rho^2}{(N \eta \rho)^{k/2}} \frac{1}{(\sqrt{N} \eta)^{n-1}} \frac{1}{N}\sqrt{\sum_{a,b} |(GAG)_{ab}|^2} \sqrt{\sum_{ab} |(GG)_{ab}|^2} \sum_{ab} |M_{ab}|^2 \langle \Im GA \Im GA^*\rangle^{\frac{n-1}{2}} \\
		\lesssim & \frac{\lambda^2 \rho}{\eta} \frac{1}{(N \eta \rho)^{(k-3)/2}} \frac{1}{(\sqrt{N} \eta)^n} \langle \Im GA \Im GA^*\rangle^{\frac{n}{2}} \lesssim \frac{\lambda^2 \rho}{\eta}  \frac{1}{(\sqrt{N} \eta)^n} \langle \Im GA \Im GA^*\rangle^{\frac{n}{2}} \,. 
	\end{split}
\end{equation}
To go to the second line, we used several Schwarz inequalities, single resolvent laws, and summation extensions \eqref{eq:sumextend}. In the penultimate step, we then carried out the summations and used Ward identities. The last step only used that $N \eta \rho \ge 1$ and $k \ge 4$. 

The terms in case (b) can be bounded using the exact same reasoning as in \eqref{eq:caseb}, yielding
\begin{equation*}
	\begin{split}
		\text{(b)-terms} &\lesssim \frac{\lambda^k}{N^{k/2 +1}} \left(\frac{1}{\sqrt{\eta \lambda^2 \rho}}\right)^{k-3} \frac{N^2 \rho}{\lambda \eta^2} \frac{1}{(\sqrt{N} \eta)^{n-1}} \langle \Im GA \Im GA^*\rangle^{\frac{n}{2}} \\
		&\lesssim \frac{\lambda^2 \rho}{\eta} \frac{1}{(N \eta \rho)^{(k-3)/2}} \frac{1}{(\sqrt{N} \eta)^{n}} \langle \Im GA \Im GA^*\rangle^{\frac{n}{2}} \lesssim \frac{\lambda^2 \rho}{\eta}  \frac{1}{(\sqrt{N} \eta)^{n}} \langle \Im GA \Im GA^*\rangle^{\frac{n}{2}}
	\end{split}
\end{equation*}

Finally, we turn to case (c) and estimate the analog of \eqref{eq:casec1} as
\begin{equation} \label{eq:casec1A}
	\begin{split}
		&\frac{\lambda^k}{N^{k/2 + 2}} \bigg(\frac{1}{\sqrt{\eta \lambda^2 \rho}}\bigg)^{k-4} \frac{1}{(\sqrt{N} \eta)^{n-2}}\sum_{a,b} \left| \big(GG\big)_{aa} (G A G)_{bb} \right|^2 \langle \Im GA \Im G A^* \rangle^{\frac{n-2}{2}} \\
		\lesssim & \frac{\lambda^2  \rho }{(N \eta \rho)^{k/2}}  \frac{1}{(\sqrt{N} \eta)^{n-2}}\bigg(\frac{1}{N} \sum_{a} (\Im G \Im G)_{aa}\bigg)  \bigg(\frac{1}{N} \sum_{b} (GA\Im GA^* G^*)_{bb}\bigg) \langle \Im GA \Im G A^* \rangle^{\frac{n-2}{2}} \\
		\lesssim & \frac{\lambda^2 \rho}{\eta} \frac{1}{(N \eta \rho)^{(n-2)/2}} \frac{1}{(\sqrt{N} \eta)^{n}} \langle \Im GA \Im G A^* \rangle^{\frac{n}{2}} \lesssim \frac{\lambda^2 \rho}{\eta} \frac{1}{(\sqrt{N} \eta)^{n}} \langle \Im GA \Im G A^* \rangle^{\frac{n}{2}}\,. 
	\end{split}
\end{equation}
To go to the last line, we bounded $\langle \Im G \Im G \rangle \lesssim \rho/\eta$, using a norm bound $\Vert \Im G \Vert \le \eta^{-1}$ and a single resolvent law to estimate $|\langle \Im G \rangle| \lesssim \rho$. In the ultimate step we then used $N \eta \rho \ge 1$ and $n \ge 2$. \qed

\subsubsection{On the precise index structure} \label{subsubsec:12} In the previous Sections \ref{subsubsec:zagav2G}--\ref{subsubsec:zagav2G1A}, we neglected the two different indices of $G$'s, $\eta$'s, and $\rho$'s for notational convenience. The goal of this section is to carry the precise index structure in a few exemplary "extreme" terms (in regard of indices involved) and hence check the bounds from the previous two sections in more details. 

In the case of Section \ref{subsubsec:zagav2G} (i.e.~the proof of Proposition \ref{prop:zagavGron2G}), we consider only those terms, where the derivatives hit both $G_1$ and $G_2$. Otherwise, the provided bounds in Section \ref{subsubsec:zagav2G} hold verbatim, we just need to add the index of the differentiated $G_i$ to the $\rho$'s and $\eta$'s involved in the bounds. Hence, as the first exemplary term, we consider \eqref{eq:n13}:
\begin{equation} \label{eq:n13indices}
	\begin{split}
		&\frac{\lambda^3}{N^{5/2}} \sum_{a,b} \left| \big(G_1A  G_2\big)_{ba}  (\Im G_2 A \Im G_1)_{aa} (G_1)_{bb} \right|  \\
		\lesssim &\frac{\lambda^3}{N^{5/2}} \sqrt{\sum_{a,b} \big|  \big(G_1A G_2 \big)_{ba} \big|^2} \sqrt{\sum_{a,b} \big|  \big(\Im G_2A \Im G_1 \big)_{ba} \big|^2} \sqrt{\sum_{a,b} |(M_1)_{ab}|^2}  \\
		\lesssim & \frac{\lambda^2}{N \eta_1 \eta_2} \langle \Im G_1 A \Im G_2A^* \rangle \lesssim \max_i \left(\frac{\lambda^2 \rho_i}{\eta_i} \right) \frac{\langle \Im G_1 A \Im G_2A^* \rangle  }{\sqrt{N \ell}}
	\end{split}
\end{equation}
where we recall the notation $\ell := \min_i (\eta_i \rho_i)$ and used $N \ell \ge 1$. Next we consider \eqref{eq:k3n3}: 
\begin{equation} \label{eq:k3n3indices}
	\begin{split}
		&\frac{\lambda^3}{N^{9/2}} \sum_{a,b} \left| (G_1A\Im G_2A^* \Im G_1)_{ab}  \right|^2 \left| (G_2 A^* \Im G_1 A \Im G_2)_{ab} \right| \\
		\lesssim & \frac{\lambda^2}{N^{7/2}} \frac{1}{\sqrt{\eta_2 \rho_2}} \left(\sum_{a} (G_1A\Im G_2A^* G_1^*)_{aa} \right) \left(\sum_{a} (\Im G_1A\Im G_2A^* \Im G_1)_{aa}\right)   \\
		\lesssim &{\frac{\lambda^2 \rho_1}{\eta_1}}\frac{1}{(N \ell)^{3/2}} \langle \Im G_1A\Im G_2A^* \rangle^3 \,. 
	\end{split}
\end{equation}
As a final example, we consider the higher order cumulant term \eqref{eq:casec1}. To estimate \eqref{eq:casec1} with a precise index structure, note that \eqref{eq:morederiv} admits the bound
\begin{equation} \label{eq:morederivindices}
	\big| \partial_{ab}^p \partial_{ba}^q \langle \Im GA \Im GA^* \rangle \big| \lesssim \left(\frac{1}{\sqrt{ \lambda^2 \ell}}\right)^{p+ q} \langle \Im GA \Im GA^*\rangle \qquad \text{for} \qquad p,q \in \N_0\,, 
\end{equation}
and we hence find (ignoring the untouched factor $\langle \Im G_1A \Im G_2 A^* \rangle^{n-2}$)
\begin{equation} \label{eq:casec1indices}
	\begin{split}
		&\frac{\lambda^k}{N^{k/2 + 2}} \bigg(\frac{1}{\sqrt{ \lambda^2 \ell}}\bigg)^{k-4} \sum_{a,b} \left| \big(G_1AG_2\big)_{aa} (\Im G_2 A^* \Im G_1)_{bb} \right|^2  \\
		\lesssim & \frac{\lambda^2 \ell }{(N \ell)^{k/2}} \bigg(\frac{1}{N} \sum_{a} (G_1A\Im G_2A^* G_1^*)_{aa}\bigg)  \bigg(\frac{1}{N} \sum_{b} (\Im G_1A\Im G_2A^* \Im G_1)_{bb}\bigg)  \\
		\lesssim & \frac{\lambda^2 \rho_1}{\eta_1} \frac{1}{(N \ell)^{n/2}}  \langle \Im GA \Im G A^* \rangle^{2} \,. 
	\end{split}
\end{equation}

We conclude this section by discussing three more exemplary terms from Section \ref{subsubsec:zagav2G1A}, where we mostly focus on situations, when the derivatives hit only copies of one of the $G_i$'s, say, $G_1$ for concreteness. 

First, we consider \eqref{eq:n11A}, which we control as
\begin{equation} \label{eq:n11Aindices}
	\begin{split}
		&\frac{\lambda^3}{N^{5/2}} \sum_{a,b} \left| \big(G_1A G_2 G_1\big)_{aa} (G_1)_{ab} (G_1)_{bb} \right| \\
		\lesssim &\frac{\lambda^3}{N^{5/2}} \frac{1}{\sqrt{\eta_1 \eta_2}  }\sqrt{\sum_a \big(G_1A \Im G_2 A^* G_1^*\big)_{aa}} \sqrt{\sum_a (\Im M_1)_{aa}}\sum_{a,b} |(M_1)_{ab}|^2  \\
		\lesssim &\sqrt{\frac{\lambda^2 \rho_1}{\eta_1}} \frac{1}{\sqrt{N \eta_1 \eta_2}  } \langle \Im G_1A \Im G_2A^* \rangle^{1/2} 
	\end{split}
\end{equation}

Next, we estimate \eqref{eq:k3n2A} (now one of the three derivatives hit a copy of $G_2$):
\begin{equation} \label{eq:k3n2Aindices}
	\begin{split}
		&\frac{\lambda^3}{N^{7/2}} \sum_{a,b} \left| (G_1A G_2 G_1)_{ab} (G_1A  G_2)_{aa} (G_2G_1)_{bb} \right| \\
		&\lesssim  \frac{\lambda^3}{N^{7/2}} \frac{1}{\eta_1 \eta_2^{3/2}} \sum_{a} (G_1A\Im G_2A^* G_1^*)_{aa} \sum_{b} (\Im G_1)_{bb} \sqrt{(\Im G_2)_{bb}} \\
		& \lesssim {\frac{\lambda^2 \rho_1}{\eta_1}}\frac{1}{\sqrt{N \eta_2 \rho_2}} \left( \frac{1}{\sqrt{N \eta_1 \eta_2}}\langle \Im G_1A\Im G_2A^* \rangle^{1/2} \right)^2\,. 
	\end{split}
\end{equation}

Finally, we consider a higher order cumulant term, namely a variant of \eqref{eq:casea1Aindices}, but where all derivatives hit copies of $G_1$. Using that 
\begin{equation} \label{eq:morederivAindices}
	\big| \partial_{ab}^p \partial_{ba}^q \langle G_1AG_2 \rangle \big| \lesssim \left(\frac{1}{\sqrt{ \lambda^2 \ell}}\right)^{p + q} \frac{1}{\sqrt{N \eta_1 \eta_2}}\langle \Im G_1A \Im G_2A^*\rangle^{1/2} \quad \text{for} \quad p,q \in \N_0\,, 
\end{equation}
we bound this variant of \eqref{eq:casea1Aindices} as (ignoring the untouched factor $\langle \Im G_1A \Im G_2A^*\rangle^{\frac{n-1}{2}}$)
\begin{equation} \label{eq:casea1Aindices}
	\begin{split}
		&\frac{\lambda^k}{N^{k/2+1}} \left(\frac{1}{\sqrt{ \lambda^2 \ell}}\right)^{k-4} \frac{1}{(N \eta_1 \eta_2 )^{\frac{n-1}{2}}}\sum_{a,b} \left| (G_1AG_2 G_1)_{aa} (G_1)_{aa} (G_1)_{bb} (G_1)_{bb} \right| \\
		\lesssim &\frac{\lambda^4 \ell^2}{(N \ell)^{k/2}} \frac{1}{(N \eta_1 \eta_2 )^{\frac{n-1}{2}}} \frac{1}{\lambda^2 \rho_1 \sqrt{\eta_1 \eta_2} } \frac{1}{N}\sqrt{\sum_{a} (G_1A\Im G_2 A^* G_1^*)_{aa}}  \sqrt{\sum_a (\Im G_1)_{aa}}\sum_{ab} |(M_1)_{ab}|^2  \\
		\lesssim &  \frac{1}{(N \ell)^{(k-3)/2}} \frac{1}{(N \eta_1 \eta_2 )^{\frac{n}{2}}} \langle \Im GA \Im GA^*\rangle^{\frac{1}{2}} \lesssim  \frac{1}{(N \eta_1 \eta_2 )^{\frac{n}{2}}} \langle \Im GA \Im GA^*\rangle^{\frac{1}{2}} \,. 
	\end{split}
\end{equation}

All the other terms can be handled in a similar way. This concludes our discussion of the precise index structure.

\subsection{Global law: Proof of Proposition \ref{prop:global2}} \label{subsec:global2} 
The goal of this section is to give the proof of the global law formulated in Proposition \ref{prop:global2}. Since all matrices and spectral parameters are constant, with a slight abuse of notation, we drop the subscript/argument $t=0$ introduced in Proposition \ref{prop:global2} throughout the whole section. Moreover, we will frequently use that our spectral parameters $z_i = z_{i,0}$ are such that 
\begin{equation*} 
\lambda^2 \max_i (\rho_i \eta_i^{-1}) \le 1/c \lesssim 1\,, \quad \text{with} \quad \eta_i := |\Im z_i|\,, \quad \rho_i := \pi^{-1} |\langle \Im M(z_i) \rangle| \,. 
\end{equation*}

Analogously to the single resolvent global law proven in Section \ref{subsec:global}, the proof of Proposition \ref{prop:global2} relies on a minimalistic cumulant expansion. In order to do so, we consider the  normalized differences (recall the notation $\ell := \min (\eta_i \rho_i)$)
\begin{align}
	\label{eq:Psi1}
	\Psi_1 = \Psi_1(z_1, z_2;A) &:= \max_{*}\sqrt{N \ell} \sqrt{\frac{N \eta_1 \eta_2}{\rho_1 \rho_2}} \big(\s_2\big)^{-1}\left| \big\langle G_1^{(*)} A G_2^{(*)} \big\rangle - \big\langle M_{1^{(*)}2^{(*)}}^A\big) \big\rangle \right| \\ \label{eq:Psi2}
	\Psi_2 = \Psi_2(z_1,z_2;A) &:= \frac{\sqrt{N\ell}}{\rho_1 \rho_2 \, \big(\s_2\big)^2}	\left| \big\langle \Im G A \Im G_2 A^* \big\rangle - \big\langle \widehat{M}_{12}^A \big\rangle \right| 
\end{align}
where  $\max_{*}$ denotes the maximum over all possibilities of adjoints taken. In Proposition \ref{prop:master} below, we show certain \emph{master inequalities} for $\Psi_1, \Psi_2$, which will immediately allow us to conclude the proof of Proposition \ref{prop:global2}. We remark that we defined $\Psi_1$ in \eqref{eq:Psi1} in such a way, that when showing $\Psi_1 \prec 1$, we will have even proven a stronger bound (by a factor $1/\sqrt{N \ell}$) than the one claimed in \eqref{eq:2G1Aglobal}. 
\begin{proposition}[Master inequalities] \label{prop:master}
	Using the above notations, we have the following: Assume that $\Psi_1 \prec \psi_1$ and $\Psi_2 \prec \psi_2$ for some deterministic control quantities $\psi_1, \psi_2 \in [1,N^D]$ for some $D < \infty$. Then it holds that 
	\begin{equation} \label{eq:master}
		\Psi_1 \prec 1 + \frac{\psi_1}{N \ell}+ \frac{(\psi_2)^{1/2}}{(N \ell)^{1/4}} \quad \text{and} \quad \Psi_2 \prec 1 + \frac{\psi_1}{N \ell} + \frac{\psi_2}{(N \ell)^{1/2}} \,. 
	\end{equation}
\end{proposition}
\begin{proof}[Proof of Proposition \ref{prop:global2}]
	Having Proposition \ref{prop:master} at hand, we can immediately deduce $\Psi_1 + \Psi_2 \prec 1$ by iterating \eqref{eq:master} finitely many times. This, in turn, directly yields \eqref{eq:2G1Aglobal}--\eqref{eq:2G2Aglobal}. 
\end{proof}
It thus remains to give the proof of Proposition \ref{prop:master}. 
\subsection{Proof of Proposition \ref{prop:master}}
As mentioned above, our argument rests on a minimalistic cumulant expansion, similarly to the single resolvent global in Section \ref{subsec:global}, and the zag steps in Sections~\ref{subsec:zagstep} and~\ref{subsec:zagstep2}. Hence, we will be rather brief and only present the key estimates. 

For both relations in \eqref{eq:master}, the starting point is the identity (denoting $M_i = M(z_i)$)
\begin{equation} \label{eq:2Gbasicident}
	\begin{split}
		\mathcal{B}_{12} \big[G_1 A G_2 - M(z_1, A, z_2)\big] = &M_1 A (G_2 - M_2) - M_1 \underline{\lambda W G_1 A G_2} \\
		&+ \lambda^2 M_1 \langle G_1 A G_2 \rangle (G_2 - M_2) + \lambda^2 M_1 \langle G_1 - M_1 \rangle G_1 A G_2
	\end{split}
\end{equation}
where we denoted the \emph{two-body stability operator} by $\mathcal{B}_{12}[\cdot] := \mathbf{1} - \lambda^2 M_1 \langle \, \cdot \, \rangle M_2$ and introduced the \emph{underline} term in \eqref{eq:2Gbasicident} as
\begin{equation*}
	\underline{\lambda W G_1 A G_2} := \lambda WG_1 A G_2 + \lambda^2 \langle G_1 \rangle G_1 A G_2 + \lambda^2 \langle G_1 A G_2 \rangle G_2\,. 
\end{equation*}
The same identity \eqref{eq:2Gbasicident} holds true when replacing $G_1 \to G_1^*$ and/or $G_2 \to G_2^*$, but for simplicity of the notation, we shall henceforth focus on the case without any adjoint taken. 

Now, after inverting $\mathcal{B}_{12}$ and taking the trace in \eqref{eq:2Gbasicident}, we arrive at
\begin{equation} \label{eq:2Gbasicident2}
	\begin{split}
		\beta_{12}	 \big\langle G_1 A G_2 - M(z_1, A, z_2)\big\rangle = &\langle M_1 A (G_2 - M_2)\rangle - \langle M_1 \underline{\lambda W G_1 A G_2}\rangle \\
		&+ \lambda^2 \langle G_1 A G_2 \rangle \langle (G_2 - M_2) M_1 \rangle+ \lambda^2  \langle G_1 - M_1 \rangle \langle G_1 A G_2 M_1 \rangle
	\end{split}
\end{equation}
where $\beta_{12} := 1 - \lambda^2 \langle M_1 M_2 \rangle$ satisfies (see, e.g., \cite[Proof of (5.12) in Appendix A.5]{cipolloni2024eigenvector})
\begin{equation} \label{eq:beta12bound}
	|\beta_{12}|^{-1} \le \frac{1}{1 - \lambda^2 \langle |M_1|^2 \rangle^{1/2} \langle |M_2|^2 \rangle^{1/2}} \le 1 + \frac{\pi}{2}\left(\frac{\lambda^2 \rho_1}{\eta_1} + \frac{\lambda^2 \rho_2}{\eta_2}\right) \lesssim 1
\end{equation}
with the last estimate being valid in the global regime. The same bound holds for $\beta_{1^*2}, \beta_{12^*}$, and $\beta_{1^*2^*}$. 

As a first step towards proving the first of the two master inequalities \eqref{eq:master} in Proposition \ref{prop:master}, we have the following \emph{underline representation lemma}. 
\begin{lemma}[Underline representation for $\langle GAG \rangle$] \label{lem:underlineGAG}
	It holds that
	\begin{equation*}
		\beta_{1^{(*)}2^{(*)}}	 \big\langle G_1^{(*)} A G_2^{(*)} - M^A_{1^{(*)}2^{(*)}}\big\rangle = - \langle M_1 \underline{\lambda W G_1^{(*)} A G_2^{(*)}}\rangle + \mathcal{O}_\prec\left( \sqrt{\frac{\rho_1 \rho_2}{N^2 \ell \eta_1 \eta_2}} \s_2 \left[1 + \frac{\psi_1}{N \ell} + \frac{(\psi_2)^{1/2}}{(N \ell)^{1/4}}\right] \right) \,. 
	\end{equation*}
\end{lemma}
Moreover, we also have the following underline representation for $\langle \Im G_1 A \Im G_2 A^* \rangle$, lying the basis for proving the second master inequality in \eqref{eq:master}. 
\begin{lemma}[Underline representation for $\langle \Im GA\Im GA\rangle$] \label{lem:underlineGAGA}
	It holds that 
	\begin{equation} \label{eq:underlineGAGA}
		\begin{split}
			&\langle \Im G_1 A \Im G_2 A^*\rangle - \langle \widehat{M}_{12}^A \rangle \\
			= &\frac{\ii}{2} \left( \langle M_1 \underline{\lambda W G_1 A \Im G_2A^*}\rangle - \langle M_1^* \underline{\lambda W G_1^* A \Im G_2A^*}\rangle  \right) \\
			&+ \frac{\lambda^2}{4 \beta_{12}}  \langle M_1 \underline{\lambda W G_1 A G_2}\rangle \langle M_1 A^* M_2\rangle - \frac{\lambda^2}{4 \beta_{1^*2}}  \langle M_1^* \underline{\lambda W G_1^* A G_2}\rangle \langle M_1^* A^* M_2\rangle \\
			&- \frac{\lambda^2}{4 \beta_{12^*}}  \langle M_1 \underline{\lambda W G_1 A G_2^*}\rangle \langle M_1 A^* M_2^*\rangle + \frac{\lambda^2}{4 \beta_{1^*2^*}}  \langle M_1^* \underline{\lambda W G_1^* A G_2^*}\rangle \langle M_1^* A^* M_2^*\rangle \\
			& + \mathcal{O}_\prec \left(\frac{\rho_1 \rho_2}{\sqrt{N \ell}} \big(\s_2\big)^2 \left[ 1 +\frac{\psi_1}{N \ell} + \frac{(\psi_2)^{1/2}}{(N \ell)^{1/4}} \right]\right) \,. 
		\end{split}
	\end{equation}
\end{lemma}
The proofs of Lemmas \ref{lem:underlineGAG}--\ref{lem:underlineGAGA} are given in Section \ref{subsec:underlineproof} below. Armed with the above lemmas, we now turn to proving Proposition \ref{prop:master}. Similarly to the proofs of Propositions \ref{prop:zagavGron2G}--\ref{prop:zagavGron2G1A}, we will drop all indices of $G_i, \eta_i, \rho_i$ and refrain from distinguishing $G, G^*$ and $A, A^*$ (unless we wish to stress positivity of certain matrices). The estimates required for this argument are, moreover, very similar to the ones presented in the proofs of Propositions \ref{prop:zagavGron2G}--\ref{prop:zagavGron2G1A}. However, one new term corresponding to a second order cumulant arises, which is given special attention in the proofs below. Other terms are discussed only very briefly. 
\subsubsection{Proof of the first master inequality in \eqref{eq:master}} \label{subsubsec:master1}
As described above, the first master inequality can be obtained by a minimalistic cumulant expansion as in \eqref{eq:isocumex} and \eqref{eq:avcumex}, relying on the underline representation from Lemma \ref{lem:underlineGAG}. Exactly as in the proof of the isotropic and average single resolvent global laws, the cumulant expansion produces a second order term (analog of the first terms on the rhs.~of \eqref{eq:isocumex} and \eqref{eq:avcumex}), that is not present in a similar cumulant expansion employed in the \emph{zag} step performed in Section~\ref{subsec:zagstep2}. 

For this new term, we bound, analogously to \eqref{eq:global2ndiso} and \eqref{eq:2ndordav},
\begin{equation}
	\begin{split}
		\frac{\lambda^2}{N^3} \bigg| \sum_{a,b} (GAGM)_{ab} (GAGG)_{ab}\bigg| &\lesssim \frac{\lambda^2}{N^2} \langle GAGMM^*G^* A^* G^* \rangle^{1/2} \langle GAGGG^*G^* A^* G^* \rangle^{1/2}  \\
		&\lesssim  \left(\frac{1}{\sqrt{N \eta \rho}} \sqrt{\frac{\rho^2}{N \eta^2}} \s_2\left[1 + \frac{(\psi_2)^{1/2}}{(N \eta \rho)^{1/4}}\right]\right)^{2}
	\end{split}
\end{equation}
using $GG^* = \Im G/\eta$, norm bounds $\Vert \Im G \Vert \lesssim \eta^{-1}$ and $\Vert MM^*\Vert \lesssim \eta^{-2}$, and that $\lambda^2 \rho/\eta \lesssim 1$ in the global regime.  

Higher order cumulant terms can be estimated similarly to the ones arising in the \emph{zag} step in \eqref{eq:n11A}, \eqref{eq:n13A}, \eqref{eq:k3n2A}, \eqref{eq:k3n3A}, \eqref{eq:casea1A}, or \eqref{eq:casec1A}, as they carry the same index structure (recall the discussions around \eqref{eq:global3rdiso} and at the end of the proof of Lemma \ref{lem:bootstrapav}), and shall hence not be discussed further. This concludes the proof of the first master inequality in \eqref{eq:master} and we proceed with proving the second one. 
\subsubsection{Proof of the second master inequality in \eqref{eq:master}} \label{subsubsec:master2}
As described above, we also prove the second master inequality by cumulant expansion, now relying on the underline representation from Lemma \ref{lem:underlineGAGA}. Contrary to Section \ref{subsubsec:master1}, there are now multiple underlined terms, that have to be controlled. We start with an underlined term from the second line of \eqref{eq:underlineGAGA}. The new second order cumulant term can be bounded as
\begin{equation}
	\begin{split}
		&\frac{\lambda^2}{N^3} \bigg| \sum_{a,b} (GA\Im GA^*M)_{ab} (GA\Im G A^* \Im G)_{ab}\bigg|  \\
		\lesssim &\frac{\lambda^2}{N^3} \sum_a  (GA\Im G A^* G^*)_{aa} \sqrt{\sum_b (M^* A \Im G A^* M)_{bb}} \sqrt{\sum_b (\Im G A \Im G A^* \Im G)_{bb}} \\
		\lesssim  &\left(\frac{\rho^2 \big(\s_2\big)^2}{\sqrt{N \eta \rho}} \left[1 + \frac{\psi_2}{\sqrt{N \eta \rho}}\right]\right)^{2}
	\end{split}
\end{equation}
using that 
\begin{equation*}
	\langle \Im GA \Im GA^* \rangle \lesssim \rho^2 \big(\s_2\big)^2 \left[1 + \frac{\psi_2}{\sqrt{N \eta \rho}}\right]  \quad  \text{and} \quad 	\langle \Im G A^* MM^* A\rangle \lesssim \frac{1}{\eta} \langle \Im M A \Im M A^* \rangle \lesssim \frac{\rho^2 \big(\s_2\big)^2}{\eta}
\end{equation*}
which follow by a single resolvent local law, similarly to \eqref{eq:singleGpos} below. 

For an underlined term from the third and fourth line of \eqref{eq:underlineGAGA}, we use that $|\beta_{1^{(*)}2^{(*)}}| \gtrsim 1$ in the global regime, and the $M$-bound \eqref{eq:Mbound}. Then, the second order cumulant term can be bounded as
\begin{equation}
	\begin{split}
		&\frac{\lambda^4}{N^3} \sqrt{\frac{\rho}{\eta}} \rho \,  \s_2 \bigg| \sum_{a,b} (GAGM)_{ab} (GA\Im GA^* \Im G)_{ab} \bigg| \\
		\lesssim &\left(\frac{\lambda^2 \rho}{\eta}\right)^2 \frac{1}{(N \eta \rho)^{3/2}} \rho \, \s_2 \langle \Im GA \Im GA^* \rangle^{3/2} \lesssim \left(\frac{\rho^2 \big(\s_2\big)^2}{\sqrt{N \eta \rho}} \left[1 + \frac{\psi_2}{\sqrt{N \eta \rho}}\right]\right)^{2}
	\end{split}
\end{equation}
where we used a Schwarz inequality, $\lambda^2 \rho/\eta \lesssim 1$ and $N \eta \rho \ge 1$. 

Higher order cumulant terms can be estimated similarly to the ones arising in the \emph{zag} step in \eqref{eq:n11}, \eqref{eq:n12}, \eqref{eq:n13}, \eqref{eq:k3n2a}, \eqref{eq:k3n2}, \eqref{eq:k3n3}, or \eqref{eq:kge4casei}--\eqref{eq:casec2}, as they carry the same index structure (recall the discussions around \eqref{eq:global3rdiso} and at the end of the proof of Lemma \ref{lem:bootstrapav}), and shall hence not be discussed further. This concludes the proof of the second master inequality in \eqref{eq:master} and hence the proof of Proposition \ref{prop:master}. 
\qed
\subsubsection{Proofs of Lemmas \ref{lem:underlineGAG} and \ref{lem:underlineGAGA}} \label{subsec:underlineproof}
In this section, we collect the proofs of the underline representations in Lemmas \ref{lem:underlineGAG} and \ref{lem:underlineGAGA}. 
\begin{proof}[Proof of Lemma \ref{lem:underlineGAG}]
We focus on the case of no adjoints taken, the other cases are handled similarly. 	To establish Lemma \ref{lem:underlineGAG}, we now estimate all but the underline term in \eqref{eq:2Gbasicident2} separately. For the first one, we have that
	\begin{equation} \label{eq:cont1}
		|\langle M_1 A (G_2 - M_2)\rangle| \prec \frac{1}{N \eta_2} \sqrt{\frac{\langle \Im M_2 A^* M_1 M_1^* A \rangle}{\rho_2}} \lesssim \frac{1}{\sqrt{N \ell}} \sqrt{\frac{\rho_1 \rho_2}{N \eta_1 \eta_2}} \s_2
	\end{equation}
	as a consequence of the average single resolvent law and estimating $M_1M_1^* \lesssim \Im M_1/\eta_1$. 
	
	Next, we consider the first term in the second line of \eqref{eq:2Gbasicident2} and subtract and add the deterministic approximation $\langle M(z_1, A, z_2) \rangle$ to $\langle G_1 A G_2 \rangle$. For the deterministic approximation, we use the overestimate (since in the case without any adjoint taken, it is actually zero, but we aim to treat any constellation of stars) by the worst case bound in \eqref{eq:Mbounds} as
	\begin{equation} \label{eq:Mbound}
		|	\langle M_{1^{(*)}2^{(*)}}^A \rangle| = \frac{|\langle M_1^{(*)} A M_2^{(*)} \rangle|}{|\beta_{1^{(*)}2^{(*)}}|} \lesssim \sqrt{\frac{\rho_1}{\eta_2} + \frac{\rho_2}{\eta_1}} \sqrt{\rho_1 \rho_2}\s_2\,. 
	\end{equation}
	The $M$-contribution can then be bounded as 
	\begin{equation} \label{eq:cont2}
		\begin{split}
			&|\lambda^2 \langle (G_2 - M_2) M_1 \rangle| \sqrt{\frac{\rho_1}{\eta_2} + \frac{\rho_2}{\eta_1}} \sqrt{\rho_1 \rho_2} \s_2\\
			\prec \ & \lambda^2 \langle \Im M_1 \Im M_2\rangle^{1/2} \sqrt{\frac{\rho_1}{\eta_2} + \frac{\rho_2}{\eta_1}} \frac{1}{\sqrt{N \ell}}\sqrt{\frac{\rho_1 \rho_2}{N \eta_1 \eta_2}} \s_2 \\
			\lesssim \ &  \left(\frac{\lambda^2 \rho_1}{\eta_1} + \frac{\lambda^2 \rho_2}{\eta_2}\right)\frac{1}{\sqrt{N \ell}}\sqrt{\frac{\rho_1 \rho_2}{N \eta_1 \eta_2}} (\s_2)^{1/2} \lesssim \frac{1}{\sqrt{N \ell}}\sqrt{\frac{\rho_1 \rho_2}{N \eta_1 \eta_2}} \s_2
		\end{split}
	\end{equation}
	where in the first step, we used a single resolvent law and $M_1M_1^* \lesssim \Im M_1/\eta_1$. To go to the last line, we then bounded $\langle \Im M_1 \Im M_2\rangle^{1/2} \lesssim \sqrt{\rho_1/\eta_2 \wedge \rho_2/\eta_1}$ and used a Young inequality. In the ultimate step, we finally used that we are in the global regime. The remaining $(G-M)$-type contribution can then be estimated as
	\begin{equation} \label{eq:cont2G-M}
		\min_i \left(\frac{\lambda^2 \rho_i}{\eta_i}\right) \frac{1}{\sqrt{N \ell}}\sqrt{\frac{\rho_1 \rho_2}{N \eta_1 \eta_2}} (\s_2)^{1/2} \frac{\psi_1 }{N \ell} \lesssim \frac{1}{\sqrt{N \ell}}\sqrt{\frac{\rho_1 \rho_2}{N \eta_1 \eta_2}} (\s_2)^{1/2} \frac{\psi_1 }{N \ell} \,. 
	\end{equation}
	
	Finally, we turn to the last term in the second line of \eqref{eq:2Gbasicident2}, which we estimate as
	\begin{equation} \label{eq:cont3}
		\begin{split}
			\lambda^2  |\langle G_1 - M_1 \rangle \langle G_1 A G_2 M_1 \rangle| &\prec \frac{\lambda^2}{N \eta_1} \frac{\rho_1^{1/2}}{\eta_1 \eta_2^{1/2}} \langle \Im G_1 A \Im G_2 A^* \rangle^{1/2} \\
			&\prec \frac{\lambda^2 \rho_1}{\eta_1} \frac{1}{\sqrt{N \ell}} \sqrt{\frac{\rho_1 \rho_2}{N \eta_1 \eta_2}} \s_2\left[1 + \frac{\sqrt{\psi_2}}{(N \ell)^{1/4}}\right] \\
			&\lesssim \frac{1}{\sqrt{N \ell}} \sqrt{\frac{\rho_1 \rho_2}{N \eta_1 \eta_2}}  \s_2\left[ 1 + \frac{\sqrt{\psi_2}}{(N \ell)^{1/4}}\right] \,. 
		\end{split}
	\end{equation}
	Combining the estimates in \eqref{eq:cont1}, \eqref{eq:cont2}, \eqref{eq:cont2G-M}and \eqref{eq:cont3}, we conclude the proof of Lemma \ref{lem:underlineGAG}. 
\end{proof}

\begin{proof}[Proof of Lemma \ref{lem:underlineGAGA}] 
	From \eqref{eq:2Gbasicident}, we find that
	\begin{equation} \label{eq:fourfold} 
		\begin{split}
		&\langle G_1^{(*)} A G_2^{(*)} A^*- M_{1^{(*)}2^{(*)}}^A A^* \rangle   \\
		= \ &\langle M_1^{(*)} A (G_2^{(*)} - M_2^{(*)}) \widetilde{A}^*_{1^{(*)}2^{(*)}} \rangle - \langle M_1^{(*)} \underline{\lambda W G_1 A G_2} \widetilde{A}^*_{1^{(*)}2^{(*)}} \rangle \\
			&+ \lambda^2 \langle M_1^{(*)} \langle G_1^{(*)} A G_2^{(*)} \rangle (G_2^{(*)} - M_2^{(*)})\widetilde{A}^*_{1^{(*)}2^{(*)}} \rangle + \lambda^2 \langle M_1^{(*)} \langle G_1^{(*)} - M_1^{(*)} \rangle G_1^{(*)} A G_2^{(*)} \widetilde{A}^*_{1^{(*)}2^{(*)}} \rangle
		\end{split}
	\end{equation}
	where we denoted 
		\begin{equation} \label{eq:XA}
	\widetilde{A}^*_{1^{(*)}2^{(*)}} := 	\left(\big(\mathcal{B}_{1^{(*)}2^{(*)}}^*\big)^{-1}[A^*]\right)^* = A^* + \lambda^2 \frac{\langle M^{(*)}_1 A^* M_2^{(*)} \rangle}{\beta_{1^{(*)}2^{(*)}}} \,. 
	\end{equation}
	The three lines of underline terms in \eqref{eq:underlineGAGA} arise by taking fourfold linear combinations of $G_1/G_1^*$ and $G_2/G_2^*$ in \eqref{eq:fourfold} and additionally using the definition of $	\widetilde{A}^*_{1^{(*)}2^{(*)}}$ in \eqref{eq:XA}. 	It now remains to estimate all the other non-underlined terms arising from the fourfold linear combination of \eqref{eq:fourfold} together with the representation \eqref{eq:XA}, for which we will always discuss the two contributions in \eqref{eq:XA} separately. 
	
	Corresponding to the first term on the rhs.~of \eqref{eq:fourfold} and the first contribution of \eqref{eq:XA}, we bound
	\begin{equation} \label{eq:under11}
		|\langle \Im M_1 A (G_2 - M_2) A^* \rangle| \prec \frac{\langle \Im M_1 A \Im M_2 A^* \rangle}{\sqrt{N \ell}} \lesssim \frac{\rho_1 \rho_2}{\sqrt{N \ell}} \big(\s_2\big)^2\,, 
	\end{equation}
	where we used that, for $B \ge 0$ (here applied to $A^* \Im M_1 A$), using spectral decomposition of $B$ and the isotropic single resolvent law, it holds that
	\begin{equation} \label{eq:singleGpos}
		|\langle (G-M) B \rangle| \prec \frac{\langle \Im M B \rangle}{\sqrt{N \eta \rho}} \,. 
	\end{equation}
	We point out that in \eqref{eq:under11}, we combined the contributions of $G_1$ and $G_1^*$, which is possible since the first term on the rhs.~of \eqref{eq:XA} is unaffected by the choice of adjoints. 
	
	This is different for the second term in \eqref{eq:XA}, for which, additionally using  that $|\beta_{12}| \gtrsim 1$ in the global regime, we estimate (focusing on the case of no adjoints for simplicity of the presentation)
	\begin{equation} \label{eq:under12}
		\begin{split}
			\lambda^2 |\langle M_1 A (G_2 - M_2)\rangle| |\langle M_1 A^* M_2 \rangle| &\prec \lambda^2 \frac{\sqrt{\langle \Im M_2 A^* M_1 M_1^* A \rangle}}{N \eta_2 \rho_2^{1/2}} \sqrt{\frac{\rho_1}{\eta_2} + \frac{\rho_2}{\eta_1}} \sqrt{\rho_1 \rho_2} \s_2 \\
			& \lesssim \left(\frac{\lambda^2 \rho_1}{\eta_1} + \frac{\lambda^2 \rho_2}{\eta_2}\right) \frac{\rho_1 \rho_2}{N \ell} \big(\s_2\big)^2 \lesssim \frac{\rho_1 \rho_2}{\sqrt{N \ell}} \big(\s_2\big)^2 \,. 
		\end{split}
	\end{equation}
	Here, we used a single resolvent law, $M_1 M_1^* \lesssim \Im M_1/\eta_1$, the bound \eqref{eq:Mbound}, the fact, that we are in the global regime, and $N \ell \ge 1$. 
	
	Next we turn to the first term in the third line of \eqref{eq:fourfold} with the first contribution of \eqref{eq:XA}, and consider, for simplicity of the notation, the case of no adjoints taken. Moreover, similarly to the discussion above \eqref{eq:Mbound}, we subtract and add the deterministic approximation $\langle M(z_1, A, z_2)\rangle$ to $\langle G_1 A G_2 \rangle$ and bound the deterministic by \eqref{eq:Mbound}. The $M$-contribution can then be estimated, analogously to \eqref{eq:cont2}, as 
	\begin{equation} \label{eq:under21}
		\begin{split}
			&\lambda^2 |\langle (G_2 - M_2)A^* M_1 \rangle| \sqrt{\frac{\rho_1}{\eta_2} + \frac{\rho_2}{\eta_1}} \sqrt{\rho_1 \rho_2} \s_2 \\
			\prec \, &\lambda^2 \frac{1}{N \eta_2 \eta_1^{1/2}\rho_2^{1/2}}\sqrt{\frac{\rho_1}{\eta_2} + \frac{\rho_2}{\eta_1}} {\rho_1 \rho_2} \big(\s_2\big)^2 \lesssim \left(\frac{\lambda^2 \rho_1}{\eta_1} + \frac{\lambda^2 \rho_2}{\eta_2}\right) \frac{\rho_1 \rho_2 \big(\s_2\big)^2}{N \ell} \lesssim \frac{\rho_1 \rho_2 \big(\s_2\big)^2}{(N \ell)^{1/2}} \,,
		\end{split}
	\end{equation}
	where we again used a single resolvent law, $M_1 M_1^* \lesssim \Im M_1/\eta_1$, the characterization of the global regime, and $N \ell \ge 1$. Analogously, the $(G-M)$-type contribution, we bound as
	\begin{equation} \label{eq:under21G-M}
		\frac{\rho_1 \rho_2 \big(\s_2\big)^2}{\sqrt{N \ell}} \frac{\psi_1}{(N \ell)^{3/2} } \,. 
	\end{equation}
	For the first term in the third line of \eqref{eq:fourfold} in now remains to discuss the second contribution of \eqref{eq:XA}. Again, decomposing $\langle G_1 A G_2 \rangle$, we bound the $M$-type term, using \eqref{eq:Mbound}, the characterization of the global domain, and $\langle \Im M_1 \Im M_2\rangle^{1/2} \lesssim \sqrt{\rho_1/\eta_2 \wedge \rho_2/\eta_1}$, as
	\begin{equation} \label{eq:under22}
		\begin{split}
			&\lambda^4 |\langle (G_2 - M_2) M_1 \rangle| \left(\frac{\rho_1}{\eta_2} + \frac{\rho_2}{\eta_1}\right) \rho_1 \rho_2 \big(\s_2\big)^2 \\
			\prec \, &\frac{\lambda^4}{N \eta_2 \rho_2^{1/2} \eta_1^{1/2}} \langle \Im M_1 \Im M_2\rangle^{1/2}  \left(\frac{\rho_1}{\eta_2} + \frac{\rho_2}{\eta_1}\right) \rho_1 \rho_2 \big(\s_2\big)^2 \\
			\lesssim \, & \left(\max_i\frac{\lambda^2 \rho_i}{\eta_i} \right)^2 \frac{\rho_1 \rho_2 \big(\s_2\big)^2}{N \ell} \lesssim \frac{\rho_1 \rho_2 \big(\s_2\big)^2}{(N \ell)^{1/2}}
		\end{split}
	\end{equation}
	and the $(G-M)$-type term as
	\begin{equation} \label{eq:under22G-M}
		\begin{split}
			&\lambda^4 |\langle (G_2 - M_2) M_1 \rangle| \frac{1}{(N \eta_1 \eta_2)^{1/2}}\sqrt{\frac{\rho_1}{\eta_2} + \frac{\rho_2}{\eta_1}} \rho_1 \rho_2 \big(\s_2\big)^2 \psi_1  \\
			\prec \, &\frac{\lambda^4 \langle \Im M_1 \Im M_2\rangle^{1/2} }{N^{3/2} \eta_2^{3/2} \rho_2^{1/2} \eta_1}  \sqrt{\frac{\rho_1}{\eta_2} + \frac{\rho_2}{\eta_1}} \rho_1 \rho_2 \big(\s_2\big)^2 \psi_1   \lesssim \left(\max_i\frac{\lambda^2 \rho_i}{\eta_i} \right)^2 \frac{\rho_1 \rho_2\big(\s_2\big)^2}{(N \ell)^{3/2}} \lesssim \frac{\rho_1 \rho_2 \big(\s_2\big)^2}{(N \ell)^{1/2}} \frac{\psi_1}{N \ell} \,. 
		\end{split}
	\end{equation}
	
	We are now left with estimating the contributions arising from the last term in the third line of \eqref{eq:fourfold}, where for the first contribution of \eqref{eq:XA}, similarly to \eqref{eq:under11}, we can combine the contributions of $G_2$ and $G_2^*$, which is possible since the first term on the rhs.~of \eqref{eq:XA} is unaffected by the choice of adjoints. Hence, this contribution admits the bound
	\begin{equation} \label{eq:under31}
		\begin{split}
			\lambda^2 |\langle G_1 - M_1 \rangle| \, |\langle G_1 A \Im G_2 A^*M_1 \rangle | &\prec \frac{\lambda^2}{N \eta_1^{2}} \langle \Im G_1 A \Im G_2 A^* \rangle^{1/2} \langle \Im G_2  A^* \Im M_1 A \rangle^{1/2} \\
			&\prec \, \frac{\lambda^2 \rho_1}{\eta_1} \frac{\rho_1 \rho_2}{N \ell} \big(\s_2\big)^2 \left[1 + \frac{(\psi_2)^{1/2}}{(N \ell)^{1/4}}\right] \lesssim \frac{\rho_1 \rho_2}{\sqrt{N \ell}} \big(\s_2\big)^2\left[1 + \frac{(\psi_2)^{1/2}}{(N \ell)^{1/4}}\right] \,,
		\end{split}
	\end{equation}
	where in the first step, we used a single resolvent law and a Schwarz inequality. In the second step, we then employed the bound \eqref{eq:singleGpos} with $G \to \Im G_2$. 
	
	As the final term to estimate, we consider the second contribution of \eqref{eq:XA}. Here, no combination of $G_2$ and $G_2^*$ is accessible, and we need to estimate each contribution separately. Again, for notational simplicity, we focus on the case of no adjoints taken. Then, this term can be bounded, using \eqref{eq:Mbound}, as
	\begin{equation} \label{eq:under32}
		\begin{split}
			&\lambda^4 |\langle G_1 - M_1 \rangle| \, |\langle G_1 A G_2 M_1\rangle| \, \sqrt{\frac{\rho_1}{\eta_2} + \frac{\rho_2}{\eta_1}} \sqrt{\rho_1 \rho_2} \s_2 \\
			\prec \, & \frac{\lambda^4\rho_1^{1/2}}{N \eta_1^{2} \eta_2^{1/2} } \sqrt{\frac{\rho_1}{\eta_2} + \frac{\rho_2}{\eta_1}} \sqrt{\rho_1 \rho_2} \s_2 \langle \Im G_1 A \Im G_2 A^* \rangle^{1/2} \\ \lesssim \, & \left(\max_i\frac{\lambda^2 \rho_i}{\eta_i} \right)^2 \frac{\rho_1 \rho_2}{N \ell} \big(\s_2\big)^2 \left[1 + \frac{(\psi_2)^{1/2}}{(N \ell)^{1/4}}\right] \lesssim \frac{\rho_1 \rho_2}{\sqrt{N \ell}} \big(\s_2\big)^2 \left[1 + \frac{(\psi_2)^{1/2}}{(N \ell)^{1/4}}\right] \,. 
		\end{split}
	\end{equation}
	
	Collecting the estimates \eqref{eq:under11}, \eqref{eq:under12}, \eqref{eq:under21}, \eqref{eq:under21G-M}, \eqref{eq:under22}, \eqref{eq:under22G-M}, \eqref{eq:under31}, and \eqref{eq:under32}, we conclude the proof of Lemma \ref{lem:underlineGAGA}. 
\end{proof}

\appendix

\section{Additional technical results and proofs}
\label{app:techn}

\subsection{Proof of Lemma~\ref{lem:m12bounds}} \label{subsec:Mbounds}
The proof of the fifth relation in \eqref{eq:Mbounds} is immediate by the definitions \eqref{eq:m12} and \eqref{eq:defs2}, additionally involving simple algebra and elementary trace inequalities after having written out the formula \eqref{eq:hatMfull}. The proof of the remaining four bounds is similar among them, we thus just focus on the proof of the first bound in \eqref{eq:Mbounds}. To prove this bound we will use a \emph{meta-argument}, see, e.g., \cite{cook2018non}. 

We present only the main steps, since this proof is analogous to \cite[Proof of Lemma 4.3]{cipolloni2024optimal}. Let the matrix size $N$ be fixed, and consider the $(dN)\times (dN)$ Wigner matrix $W^{(d)}$ with entries having variance $1/(dN)$, and the deformation $H_0^{(d)}:=H_0\otimes I_d\in\C^{dN\times dN}$, where $I_d\in\C^{d\times d}$ is the identity matrix. Define the resolvent $G_\lambda^{(d)}(z):=(H_0^{(d)}+\lambda W^{(d)}-z)^{-1}$, then its deterministic approximation is still given by the solution of \eqref{eq:MDE} with $H_0$ being replaced with $H_0^{(d)}$, which we denote by $M_\lambda^{(d)}(z)$. In the following by $M_\lambda^{(d)}(z_1,A^{(d)},z_2)$ we denote the analog of \eqref{eq:m12} when $M_\lambda$ is replaced with $M_\lambda^{(d)}$ and $A^{(d)}:=A\otimes I_d$. The tensorization of the matrices is designed in such a way that we have
\begin{equation} \label{eq:metakey1}
\langle M_\lambda^{(d)}(z_1,A^{(d)},z_2) \rangle = \langle M_\lambda(z_1,A,z_2) \rangle \qquad \text{for all} \quad d \in \N
\end{equation}

The key idea of the meta-argument is to consider \eqref{eq:2Gbasicident} for $W = \mathrm{GUE}$ and using single resolvent laws, which yield that 
\begin{equation} \label{eq:metakey2}
\lim\limits_{d \to \infty} \left[ \langle M_\lambda^{(d)}(z_1,A^{(d)},{z_2})\rangle - \E_{\mathrm{GUE}}\langle G_\lambda^{(d)}(z_1)A^{(d)}G_\lambda^{(d)}({z_2})\rangle \right] = 0
\end{equation}
for any fixed $z_1, z_2 \in \C \setminus \R$, $A \in \C^{N \times N}$, Hermitian $H_0 \in \C^{N \times N}$, and $\lambda > 0$. 

Combining \eqref{eq:metakey1}--\eqref{eq:metakey2}, we can estimate $\langle M_\lambda(z_1, A, \overline{z_2})\rangle$ by writing 
\begin{equation}
\begin{split}
\label{eq:almth}
\hspace{-5mm}\langle G_\lambda^{(d)}(z_1)A^{(d)}G_\lambda^{(d)}(\overline{z_2})\rangle&=\langle G_\lambda^{(d)}(z_1)A^{(d)}G_\lambda^{(d)}(z_2)\rangle-2\ii\langle G_\lambda^{(d)}(z_1)A^{(d)}\Im G_\lambda^{(d)}(z_2)\rangle \\
&=\langle G_\lambda^{(d)}(z_1)A^{(d)}G_\lambda^{(d)}(z_2)\rangle-\langle M_\lambda^{(d)}(z_1,A^{(d)},z_2)\rangle-2\ii\langle G_\lambda^{(d)}(z_1)A^{(d)}\Im G_\lambda^{(d)}(z_2)\rangle,
\end{split}
\end{equation}
where we used that $\langle M_\lambda(z_1,A,z_2)\rangle= \langle M_\lambda^{(d)}(z_1,A^{(d)},z_2) \rangle =0$ since $A$ is $(z_1,z_2)$-regular. The first term in the second line of \eqref{eq:almth} can be controlled using \eqref{eq:metakey2}. Finally, by a Schwarz inequality we estimate
\begin{equation}
\label{eq:lastestmeta}
\big|\langle G_\lambda^{(d)}(z_1)A^{(d)}\Im G_\lambda^{(d)}(z_2)\rangle\big|\le \sqrt{\frac{| \langle \Im G_\lambda^{(d)}(z_1)A^{(d)}\Im G_\lambda^{(d)}(z_2)(A^{(d)})^*\rangle| \, |\langle\Im G_\lambda^{(d)}(z_2)\rangle|}{\eta_1}}\lesssim \sqrt{\frac{\rho_2}{\eta_1}}\mathfrak{s}_2,
\end{equation}
where in the last inequality we used that $\langle\Im G_\lambda^{(d)}(z_2)\rangle| \lesssim \rho_2$ and applied a meta-argument, analogous to \eqref{eq:metakey2}, reading 
\begin{equation} \label{eq:metaend}
\lim\limits_{d \to \infty}\left[\langle \widehat{M}_\lambda^{(d)}(z_1, A^{(d)}, z_2) (A^{(d)})^* \rangle  -  \E_{\mathrm{GUE}}\langle \Im G_\lambda^{(d)}(z_1)A^{(d)}\Im G_\lambda^{(d)}(z_2)A^*\rangle  \right] = 0
\end{equation}
together with the fifth bound in \eqref{eq:Mbounds}. Combining \eqref{eq:metakey1}--\eqref{eq:metaend}, we conclude the first bound in \eqref{eq:Mbounds} and hence the proof of Lemma \ref{lem:m12bounds}. \qed

\subsection{Proof of Lemma \ref{lem:invstab}} The identity in \eqref{eq:invstab} can easily be checked involving the definition of the stability operator in \eqref{eq:stabop}. The lower bound in \eqref{eq:stabbound} follows by writing (we drop the subscript $\lambda$ and the argument $z$, and write $\eta = |\Im z|$)
\begin{equation}
	\begin{split}
|1 - \lambda^2 \langle M^2 \rangle| &= |1 - \lambda^2 \langle MM^* \rangle - 2 \ii \lambda^2 \langle M\Im M \rangle| \\
&\ge 1 - \lambda^2 \langle MM^* \rangle + 2 \lambda^2 \langle (\Im M)^2 \rangle \gtrsim \big(\eta/(\lambda^2 \rho) \wedge 1\big) + \lambda^2 \rho^2 \,. 
	\end{split}
\end{equation}
Here, in the first step, we wrote $M = M^*+ 2 \ii \Im M$. To go to the second line, we dropped the imaginary part of $2 \ii \langle M \Im M \rangle$ and used that, as a simple consequence of the MDE \eqref{eq:MDE}, $1 - \lambda^2 \langle MM^* \rangle \ge 0$ and the positivity of $(\Im M)^2$. The last step follows since $1 - \lambda^2 \langle MM^* \rangle = \eta /(\eta + \lambda^2 |\langle \Im M \rangle| )$ and by estimating $\langle (\Im M)^2 \rangle \ge |\langle \Im M \rangle|^2$. This concludes the proof of Lemma \ref{lem:invstab}. \qed 

\subsection{Modifications in the proof of the two-resolvent bound for $\theta(z_1, z_2) = 0$} \label{subsec:regime2} 

In this section, we briefly explain, how the proof of the two resolvent bound in Theorem \ref{thm:main2G} has to be modified for the case $\theta(z_1, z_2) = 0$, complementary to the case $\theta(z_1, z_2) = 1$, which was treated in the main part of this paper (see Section \ref{subsec:proof2G}).
Completely analogously to \eqref{eq:theta=1domain}, for $\mathcal{D} \subset \C \setminus \R$, we introduce the notation
\begin{equation*}
(\mathcal{D})^{\times 2}_{\theta = 0} := \left\{(z_1, z_2) \in \mathcal{D}^{ 2} \, : \, \theta(z_1, z_2) = 0 \right\} \,. 
\end{equation*}

In the following, we only outline the changes in the key steps of the argument, compared to the proof in Sections \ref{subsec:proof2G} and \ref{sec:2GLaw}, and leave some details to the reader. In fact, the main difference is that now \emph{every} observable $A$ is regular according to Definition~\ref{def:regobs}, and we shall hence drop the $\mathring{A}$ notation entirely. We recall that, for $\theta(z_1, z_2) = 0$, the fundamental control parameter $\s_2$ is given by
\begin{equation} \label{eq:s2theta0}
\big(\s_2(z_1, z_2;A)\big)^2 = \langle \Gamma_1 A \Gamma_2 A^*\rangle + \lambda^4 \max_* \big| \langle M_1^{(*)} A M_2^{(*)} \rangle \big|^2 \langle \Gamma_1 \Gamma_2 \rangle
\end{equation}

As the first modification of the overall proof, the analog of Lemma \ref{lem:m12bounds} reads as follows. 
\begin{lemma}
	\label{lem:m12boundstheta0}
	Fix $(z_1, z_2) \in (\C\setminus\R)_{\theta = 0}^{\times 2}$  and adopt the notations from Lemma \ref{lem:m12bounds}. Let $A  \in \C^{N\times N}$. Then it holds that, for $\s_2 = \s_2(z_1, z_2;A)$, 
	\begin{equation}
		\label{eq:Mboundstheta0}
		\begin{split}
			\big|\langle M_\lambda(z_1,A,{z_2})\rangle\big|&\lesssim \big| \langle M_\lambda(z_1)A M_\lambda(z_2) \rangle\big|, \\
			\big|\langle M_\lambda(z_1,A,z_2) - M_\lambda(\overline{z}_1, A, z_2)\rangle\big|&\lesssim \sqrt{\frac{\rho_1}{\eta_2}}\sqrt{\rho_1\rho_2} \mathfrak{s}_2, \\
			\big|\langle M_\lambda({z_1},A,{z_2}) - M_\lambda(z_1, A, \overline{z}_2)\rangle\big|&\lesssim \sqrt{\frac{\rho_2}{\eta_1}}\sqrt{\rho_1\rho_2} \mathfrak{s}_2 \,. 
		\end{split}
	\end{equation}
	Moreover, for $A,B \in \C^{N \times N}$, it holds that
	\begin{equation}	\label{eq:Mboundstheta02}
	\max_*	\big|\langle \widehat{M}_\lambda(z_1^{(*)},A,z_2^{(*)})B\rangle\big|\lesssim {\rho_1\rho_2} \, \s_2(z_1, z_2;A) \,  \s_2(z_2, z_1;B) 
	\end{equation}
	and we additionally have 
	\begin{equation} \label{eq:s2Id}
\big(\s_2(z_1, z_2; \mathbf{1}) \big)^2 \lesssim \langle \Gamma_1 \Gamma_2 \rangle \lesssim \frac{1}{\rho_1 \eta_1} \wedge \frac{1}{\rho_2 \eta_2}
	\end{equation}
\end{lemma}
The proof of Lemma \ref{lem:m12boundstheta0} is completely analogous to the proof of Lemma \ref{lem:m12bounds} and so omitted. We only point out the key differences to Lemma \ref{lem:m12bounds}: 
 First, since $\theta=0$, we have $|1 - \lambda^2 \langle M_\lambda(z_1) M_\lambda(z_2)\rangle| \ge 1/2$,
hence the denominator in the definition of $M_\lambda(z_1, A, z_2)$ in~\eqref{eq:m12} is harmless. 
This gives the first bound in \eqref{eq:Mboundstheta0}.
Second,  since every observable is regular, we do not have any orthogonality relation of the form $\langle MAM\rangle = 0$.
In fact, the second and third bound in \eqref{eq:Mboundstheta0} directly correspond to the first and second bound in \eqref{eq:Mbounds}, just in the latter regular case one of the two $\langle M_\lambda\rangle$ terms is zero
since $\langle M_\lambda(z_1^+, A, z_2^-)\rangle =  \langle M_\lambda (z_1^+) A M_\lambda(z_2^-) \rangle = 0$
by regularity. Thus in the regular case, in Lemma \ref{lem:m12bounds},
$M$-terms of “mixed”-type objects $\langle GA \Im G  \rangle$ are
expressed in terms of $M$-terms of $\langle GAG\rangle$ with appropriately chosen combination of $\pm$ on $z_1, z_2$, see~\eqref{eq:Mbounds}. Next,  \eqref{eq:Mboundstheta02} is the direct analogue of the last
estimate in \eqref{eq:Mbounds} for the $M$-term of a
$\langle \Im GA \Im G B \rangle$-type object (here we also need it for two different observables $A$ and  $B$).
Finally, the last estimate~\eqref{eq:s2Id} does not have an analogue in Lemma \ref{lem:m12bounds},
but it directly follows from the definition of $\s_2(z_1, z_2; \mathbf{1})$ in \eqref{eq:defs2}
with a harmless denominator. In the last step we used the trivial $\| \Gamma \| \le 1/\rho\eta$ bound
following from~\eqref{eq:boundM} and \eqref{eq:Gammadef}.
As we will see in the proof of the local law (around  \eqref{eq:modification} below), this
additional bound~\eqref{eq:s2Id} is needed to compensate for the lack of an $\s_2$ factor in
the first bound in~\eqref{eq:Mboundstheta0}.

As the second modification of the overall proof, the two resolvent local laws in the regime $\theta(z_1, z_2) = 0$, as an analog of Theorem \ref{thm:2GLL}, read as follows. 
\begin{theorem}[Two resolvent local laws for $\theta = 0$] \label{thm:2GLLtheta0}
	Adopt the notations and assumptions of Theorem~\ref{thm:2GLL}, but suppose that $(z_1, z_2) \in \big(\mathcal{D}^{\rm abv}\big)^{\times 2}_{\theta = 0}$. Then it holds that
	\begin{equation} \label{eq:2G1Atheta0}
		\left| \big\langle G_1^{(*)} A G_2^{(*)} \big\rangle - \big\langle M_\lambda\big(z_1^{(*)},  A ,z_2^{(*)}\big) \big\rangle \right| \prec \sqrt{\frac{\rho_1 \rho_2}{N \eta_1 \eta_2}} \s_2(z_1, z_2; A)
	\end{equation}
	and
	\begin{equation} \label{eq:2G2Atheta0}
		\left| \big\langle \Im G_1 A \Im G_2 B \big\rangle - \big\langle \widehat{M}_\lambda(z_1,  A ,z_2)B \big\rangle \right| \prec \frac{\rho_1 \rho_2}{\sqrt{N \ell}} \,  \s_2(z_1, z_2; A) \, \s_2(z_2, z_1;B)
	\end{equation}
all	uniformly in deterministic $A, B  \in \C^{N \times N}$. 
\end{theorem}
Exactly as in Section \ref{subsec:proof2G}, the estimate \eqref{eq:2G2Atheta0} for $B = A^*$ immediately yields Theorem \ref{thm:main2G} for $\theta(z_1, z_2) = 0$.

Analogously to Theorem \ref{thm:2GLL}, also Theorem \ref{thm:2GLLtheta0} is obtained by following the \emph{Zigzag strategy}, whose key steps are discussed in Section \ref{subsubsec:zigzagproof}. As now $\theta(z_1, z_2) = 0$ (and hence there is no denominator term from \eqref{eq:defs2}, recall \eqref{eq:s2theta0}), the "time-dependent" version of the basic control parameter $\s_2$ (cf.~\eqref{eq:s2time}) is simply given by the time-independent $\s_2(z_1, z_2;A)$, i.e., for all times $t \in [0,1]$, 
\begin{equation*}
\big(\s_2(t)\big)^2 := \big(\s_2(z_1, z_2;A)\big)^2 \,. 
\end{equation*}

With this definition, all the main ingredients of the zigzag induction, formulated in Propositions~\ref{prop:global2},~\ref{prop:zigstep2}, and \ref{prop:zagstep2}, hold, when naturally adapted to the modified \eqref{eq:2G2Atheta0} (compared to \eqref{eq:2G2A}). Moreover, their proofs follow along the same lines as the arguments given in Sections \ref{subsec:global2}, \ref{subsec:zigstep2}, and \ref{subsec:zagstep2}, respectively. 
 The only somewhat nontrivial difference is that 
$\langle M_\lambda(z_1, A, z_2)\rangle$ in Lemma  \ref{lem:m12boundstheta0} is 
not estimated in terms of the basic parameter $\s_2$, therefore its power counting seems to be off.
A careful inspection of the leading $M$-terms of the quantities  in~\eqref{eq:GAG} and \eqref{eq:imGAimGA} shows 
that $\langle M_\lambda\rangle$ never stands alone, it comes in combination with a $\langle \widehat M_\lambda\rangle$ 
term. Specifically, after decomposing the terms in the fourth line of \eqref{eq:imGAimGA}, schematically, according to 
$G = M + (G- M)$, we encounter leading terms of the form $\lambda^2 \langle M_\lambda\rangle \langle \widehat M_\lambda\rangle$. Using the bounds from Lemma \ref{lem:m12boundstheta0}, we easily obtain
the following combined "compensated" estimate
\begin{equation} \label{eq:modification}
	\begin{split}
		&\max_* \lambda^2 			\big|\langle M_\lambda(z_1^{(*)},A,{z_2^{(*)}})\rangle\big| \, \big|\langle \widehat{M}_\lambda(z_1^{(*)},A,z_2^{(*)})\rangle\big|  \\ 
		\lesssim & \ \rho_1   \rho_2 \, \s_2(z_1, z_2;A) \s_2(z_1, z_2, \mathbf{1}) \max_* \lambda^2 			\big|\langle M_\lambda(z_1^{(*)},A,{z_2^{(*)}})\rangle\big|  \lesssim  \,\rho_1 \rho_2 \, \big( \s_2(z_1, z_2;A)\big)^2  \,.
	\end{split}
\end{equation}
We argue similarly for the $\lambda^2 \langle M_\lambda\rangle \langle \Im G A \Im G- \widehat{M}_\lambda \rangle$ term (that arises from the fourth line of \eqref{eq:imGAimGA} as well), involving the bound \eqref{eq:2G2Atheta0}.
After this detour, the necessary  $\s_2(z_1, z_2;A)$ factors are regained. The rest of the proof is analogous to the  
treatment in the $\theta(z_1,z_2) = 1$ case, just in \eqref{eq:diffllaw}, we now use
the "compensated" bound \eqref{eq:modification}. 

This completes our discussion of the modifications for the two-resolvent bound in the case $\theta(z_1, z_2) = 0$. We leave further details to the reader.

\printbibliography

\end{document}